\documentclass[a4paper, 12pt]{article}

\usepackage{amsthm}
\usepackage{amsmath}
\usepackage{amssymb}
\usepackage{graphicx}
\usepackage[round]{natbib}
\usepackage{textcomp}
\usepackage{bbm}
\usepackage{bm}
\usepackage{url}
\usepackage[margin=2cm]{geometry}
\usepackage{mathrsfs} 
\usepackage[dvipsnames]{xcolor}
\usepackage{csquotes}
\usepackage{float}

\usepackage{array}
\usepackage{comment}

\usepackage{algorithm}
\usepackage{algpseudocode}

\usepackage{listings}
\usepackage{setspace}

\usepackage{yfonts}
\usepackage{pdfpages}
\usepackage[colorlinks=true,allcolors=blue]{hyperref}
\usepackage{mathtools}
\usepackage{stmaryrd} 
\usepackage{scalerel} 

\renewcommand{\subset}{\subseteq}

\newtheorem{thm}{Theorem}[section] 
\newtheorem{lem}[thm]{Lemma} 
\newtheorem{prop}[thm]{Proposition}

\theoremstyle{definition}  
\newtheorem{defn}{Definition}[section]

\newtheorem{rem}{Remark}[section] 
\newtheorem{example}{Example}[section]

\newtheorem{assumption}{}

\algrenewcommand\algorithmicrequire{\textbf{Input:}}
\algrenewcommand\algorithmicensure{\textbf{Output:}}

\newcommand{\A}{\mathcal{A}} 
\newcommand{\C}{\mathcal{C}} 
\newcommand{\D}{\mathcal{D}} 
\newcommand{\F}{\mathcal{F}} 
\renewcommand{\L}{\mathcal{L}} 
 
\newcommand{\cN}{\mathcal{N}}
\renewcommand{\S}{\mathcal{S}} 
\newcommand{\T}{\mathcal{T}}
\newcommand{\X}{\mathcal{X}}

\newcommand{\Exp}{\mathbb{E}}
\newcommand{\Prob}{\mathbb{P}}
\newcommand{\Qrob}{\mathbb{Q}}
\newcommand{\R}{\mathbb{R}} 

\newcommand{\BR}{\mathcal{B}(\mathbb{R})} 

\newcommand{\one}{\mathbbm{1}} 

\newcommand{\prob}{\mathsf{P}} 

\DeclareMathOperator*{\argmin}{arg\,min}
\DeclarePairedDelimiter{\norm}{\lVert}{\rVert} 

\newcommand{\fhat}{\widehat{f}}
\newcommand{\ftilde}{\widetilde{f}}

\newcommand{\iso}{\mathsf{iso}}

\title{Uncertainty quantification for expectation-calibrated predictions}
\date{}
\author{Georgios Gavrilopoulos\thanks{Seminar for Statistics, ETH Zurich, Zurich, Switzerland.} \and Johanna Ziegel\footnotemark[1]}

\begin{document}
    \maketitle
    \begin{abstract}
        The existing literature on model calibration focuses mainly on classification and probabilistic predictions.
        In this work, we address calibrated point predictions for the conditional mean.
        Although existing impossibility results preclude exact out-of-sample calibrated predictions, we develop calibrated confidence intervals that provide uncertainty quantification around such predictions.
        Our results come with distribution-free theoretical guarantees and are applicable in model-agnostic, finite-sample settings under exchangeability by leveraging conformal prediction.
        Calibrated confidence intervals rely on a general underlying binning scheme.
        We present two examples of such a binning scheme, one based on a data-independent partition and the other on isotonic regression.
        Under structural assumptions, we prove that calibrated confidence intervals based on isotonic regression come with strong asymptotic consistency properties and have an asymptotically vanishing width.
        We illustrate the empirical performance of our methods by applying them first in a simulated setting and then to a highly imbalanced insurance dataset.
    \end{abstract}


    \section{Introduction}


    The concept of forecast calibration was developed in the meteorology literature, to evaluate probabilistic forecasts for binary events, such as the probability of precipitation.
    A probabilistic forecast for a binary event is calibrated if its predicted class frequencies align with the observed frequencies.
    Formally, if \(Y\) is a binary random variable, a prediction \(p=p(X)\) for the probability \(\Prob(Y=1\, | \, X)\) is called (out-of-sample) calibrated if
    \begin{equation}
        \label{eq: calibration_binary}
        \mathbb{P}\left(Y=1\, | \, p(X)\right)=p(X)\quad \text{almost surely}.
    \end{equation}
    
    Throughout the 1960s and 1970s, calibration appeared in the literature with various names, such as unbiasedness in the small, realism, realism of confidence, and reliability \citep{Lichtenstein_1977}.
    A weaker notion of calibration, called unbiasedness in the large, had also been proposed \citep{Murphy_1967}.
    A concise review of probabilistic forecasting in those early days is given by \citet{Dawid_1986}.

    Calibration originally attracted a lot of attention because it incentivizes honest predictions \citep{Dawid_1986}.
    Recently, calibration has also been considered important because it facilitates interpretability of black-box models.
    For example, classification models often rely on the computation of a score in \([0,1]\), which is then turned into a binary output via some prespecified decision rule.
    This score is often interpreted as a class-membership probability, and it is used to quantify model uncertainty.
    If the original model is calibrated, this interpretation is valid.
    
    However, for some commonly used classification models, such as naive Bayes classifiers, SVMs, and modern neural networks, the predicted scores often do not align with the observed class frequencies \citep{Zadrozny_2002, Guo_2017}.
    This has prompted the development of post-processing methods to improve calibration properties.
    Examples of such methods include histogram binning \citep{Zadrozny_2001, Gupta_2020}, isotonic calibration \citep{Zadrozny_2002}, and Platt scaling \citep{Platt_2000}.

    Relying on histogram binning and isotonic calibration, \citet{Vovk_2004} and \citet{Vovk_2014} developed Venn and Venn-Abers predictors, respectively.
    These are post-processing recalibration methods with set-valued outputs, which are guaranteed to contain a calibrated prediction.
    The advantages of Venn and Venn-Abers predictors are their theoretical validity guarantees and their empirical performance.

    Calibration can be generalized beyond binary outcomes.
    Given random variables \(X\in \X\) and \(Y\in \R\), a (random) function \(f:\X\to \R\) is called \emph{(expectation-)calibrated} if
    \begin{equation}
        \label{eq: perfect_calibration}
        \mathbb{E}[Y| f(X)]=f(X)\quad \text{almost surely.}
    \end{equation}
    We can think of \(f\) as a predictive model for the conditional mean \(\Exp[Y|X]\).
    The definition of expectation-calibration was introduced by \citet{Tarpey_1996} using the term \emph{self-consistency}.
    The same property later appeared in \citet{Gupta_2020, wuetr_ziegel_23, Vander_Laan_2024} as auto-calibration and perfect calibration.
    \citet{Gneiting_2023} generalized it to other functionals by introducing the concept of T-calibration.

    A more general interpretation of \eqref{eq: perfect_calibration} is through nested information sets.
    Ideally, \(f(X)\) makes use of all the relevant information in \(X\) concerning $Y$, and in that case, \eqref{eq: perfect_calibration} implies that \(f(X)\) is equal to its target functional \(\mathbb{E}[Y|X]\).
    If, on the other hand, \(f(X)\) makes no use of that information and is calibrated, then \(f(X)\) is equal to the marginal expectation \(\mathbb{E}[Y]\), which is uninformative.
    The information set ($\sigma$-algebra) generated by \(f(X)\) itself can be thought of as a good compromise between these two extremes.
    This motivates definition \eqref{eq: perfect_calibration}.
    According to that, \(f(X)\) captures the information carried by itself, which explains the term \emph{self-consistency}.

    In finite samples, even the uninformative marginal expectation \(\mathbb{E}[Y]\) is impossible to access.
 Therefore, in practice, the task of deriving precisely calibrated predictions is generally impossible.
    \citet{Gupta_2020} formally showed the impossibility of this task for binary classification and proved that, even for \eqref{eq: perfect_calibration} to be satisfied approximately, \(f\) must induce a finite or countably infinite partition on the space \(\X\).
    They used histogram binning to obtain such approximately calibrated models \(f\) for binary classification and they showed that such models can be used to construct confidence intervals that contain \(\mathbb{E}[Y|f(X)]\).
    However, careful hyperparameter tuning is needed to avoid inflated intervals.

    \citet{wuetr_ziegel_23} used a slightly different technique, which borrows ideas from isotonic calibration \citep{Zadrozny_2001} and Venn-Abers predictors \citep{Vovk_2014}.
    After training a model \(\fhat\), they fit isotonic regression to \((\fhat(X_1),Y_1),\ldots,(\fhat(X_n),Y_n)\), where \(\{(X_i,Y_i)\}_{i=1}^n\) is a holdout calibration set.
    They show that the fitted function satisfies \eqref{eq: perfect_calibration} \emph{in-sample}, that is, when \((X,Y)\) is sampled from the empirical distribution of \(\{(X_i,Y_i)\}_{i=1}^n\).

    In this paper, we extend Venn-Abers predictors \citep{Vovk_2014} beyond binary classification.
    Using in-sample calibrated predictions, 
    we construct \emph{calibrated confidence intervals}, that is, intervals \(\C_{n+1,\alpha}=\C_{n+1,\alpha}(X_{n+1})\) such that
    \begin{equation}
        \label{eq: calibrated_confidence_intervals}
        \mathbb{P}\Big(f(X_{n+1})\in \C_{n+1,\alpha}\ \Big | \ \A\Big)\geq 1-\alpha,
    \end{equation}
    where \(f\) is calibrated, \(\alpha\) is a prespecified error level, and \(\A\) is a \(\sigma\)-algebra, which may depend on the training set, the calibration set, and even the test point \((X_{n+1},Y_{n+1})\).
    Since \(f\) is calibrated, this condition can be rewritten as \(\mathbb{P}\big(\mathbb{E}[Y_{n+1}|f(X_{n+1})]\in \C_{n+1,\alpha}\ \big | \ \A\big)\geq 1-\alpha\), which matches the definition of \citet{Gupta_2020}.
    Figure \ref{fig: GAM100_intervals} provides an illustration of these calibrated confidence intervals, computed for multiple test points \(X_{n+1}\).
    Further details are given in Section \ref{sec: simulation}.

    \begin{figure}
        \centering
        \includegraphics[width=0.65\linewidth]{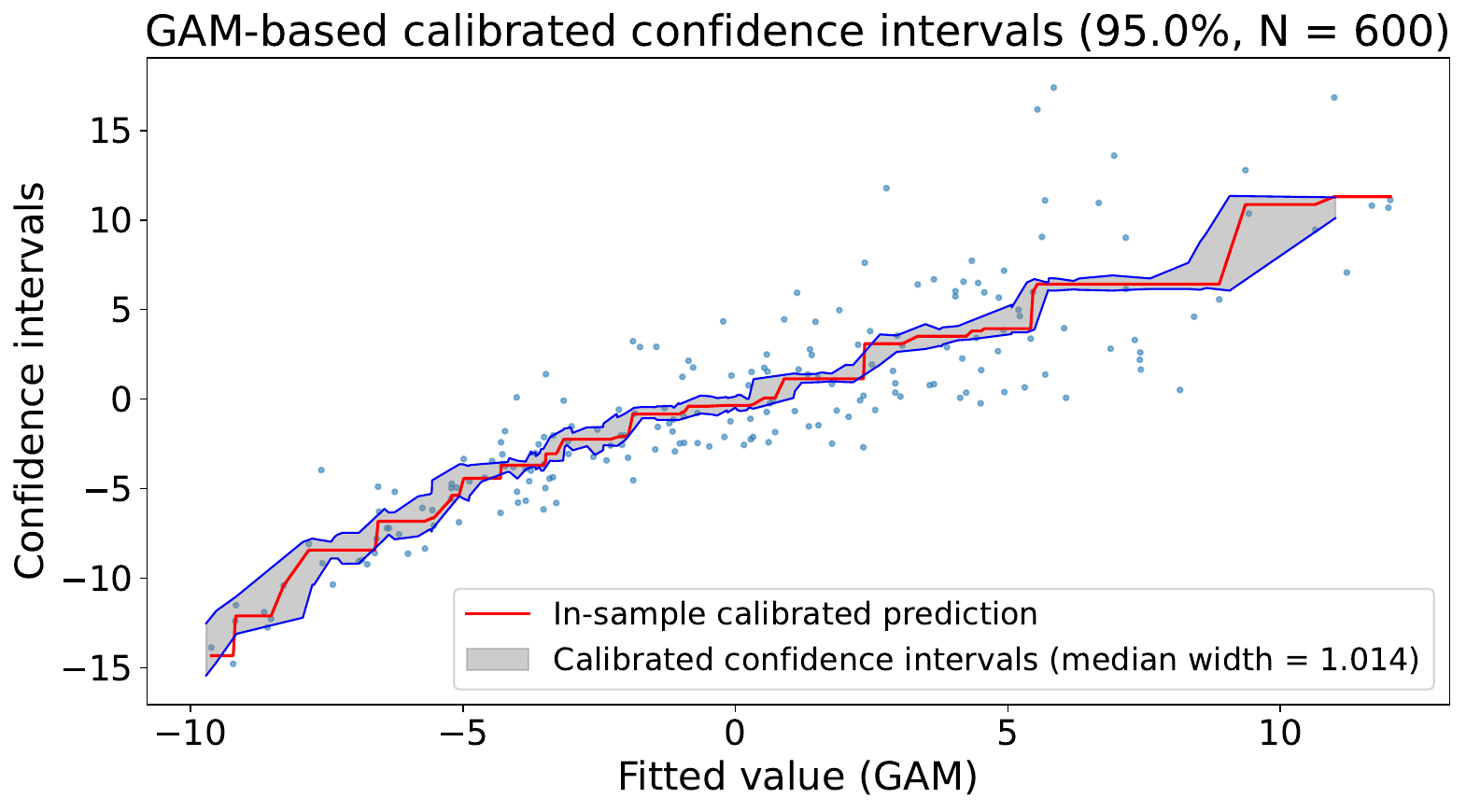}
        \caption{Calibrated confidence intervals computed for multiple test points \(X_{n+1}\).
        The in-sample calibrated prediction of \citet{wuetr_ziegel_23} is shown in red.}
        \label{fig: GAM100_intervals}
    \end{figure}
    
    These confidence intervals enjoy the same validity guarantees as the ones constructed by \citet{Gupta_2020}, but this time in the more general setting of a real-valued response variable \(Y\).
    In practice, they can be computed explicitly, without any approximations, and without cost-intensive algorithms.
    
    All our results come with appropriate theoretical guarantees.
    Calibrated confidence intervals have at least nominal finite-sample coverage and their width shrinks to zero asymptotically.
    We also propose a best-candidate point prediction from these intervals and show that it is consistent for the true conditional expectation \(\mathbb{E}[Y|X]\) under structural assumptions.
    This best candidate relates to the recalibrated model \(\ftilde\) in \citet{wuetr_ziegel_23}.
    Our results show that their proposed recalibrated prediction is also consistent for the conditional mean.

    This is not the first time that Venn-Abers predictors are generalized to the continuous setting.
    Recently, \citet{Vander_Laan_2024} developed conditionally valid prediction sets for \(Y_{n+1}\), i.e. sets \( \widehat{C}_n(X_{n+1})\) such that \(\Prob(Y_{n+1}\in \widehat{C}_n(X_{n+1})\, | \, f(X_{n+1}))\geq 1-\alpha\), where \(f\) is calibrated.
    However, their approach is computationally intensive, and \(\widehat{C}_n(X_{n+1})\) can only be approximated in practice.
    Moreover, our targets are confidence intervals that contain a calibrated prediction \(\ftilde(X_{n+1})\), not prediction sets for \(Y_{n+1}\) itself.

    The work of \citet{Petej_2026}, which was developed concurrently with ours, also extends Venn-Abers predictors to continuous labels.
    However, they only develop intervals with marginal coverage guarantees, which can be expressed as a special case of our confidence intervals.
    Moreover, they do not study the asymptotic behavior of their intervals or the asymptotic properties of recalibrated predictions.

    To construct calibrated confidence intervals, we make use of prediction sets for \(Y_{n+1}\).
    Our method does not put any restrictions on how these prediction sets are generated.
    Nevertheless, conformal prediction offers a theoretically solid, interpretable, and distribution-free uncertainty quantification framework that yields prediction sets with almost exact coverage guarantees.
    Systematic treatments of conformal prediction can be found both in the seminal work of \citet{Vovk_2005} and in the most recent books by \citet{Angelopoulos_2023} and \citet{Angelopoulos_2025}, which also provide insight into more recent developments in the field.
    
    In our applications, we use both the standard version of conformal prediction and conformalized quantile regression \citep{Romano_2019}, which outputs prediction sets with locally adaptive width.
    The prediction sets derived by \citet{Vander_Laan_2024} could also be used, but at a much higher computational cost.

    While we focus on point forecasts in this paper, a large strand of literature has proposed generalizations of \eqref{eq: calibration_binary} to probabilistic forecasts.
    There, the target is the prediction of the entire (conditional) distribution of \(Y\).
    As it turns out, \eqref{eq: calibration_binary} can be generalized in various ways, giving rise to probabilistic calibration \citep{Dawid_1984, Diebold_1998}, threshold calibration \citep{henzi2021isotonic}, marginal calibration \citep{Gneiting_2007} and auto-calibration.
    An extensive analysis of these notions of calibration is given in \citet{Gneiting_2023}.

    Various methods to construct calibrated probabilistic predictions have been proposed.
    Conformal predictive systems, which were introduced by \citet{Vovk_2019}, give rise to probabilistically calibrated predictive distributions.
    Recent works by \citet{Bostrom_2021} and \citet{Allen_2025} used binning and isotonic distributional regression \citep{henzi2021isotonic}, respectively,  to construct predictive distributions with stronger calibration properties. 

    This paper has the following structure: In Section \ref{sec: main_result}, we define \emph{calibrated confidence intervals} and we present our main result about the construction of such intervals.
    In Section \ref{sec: methodology}, we provide methodological details for the implementation of calibrated confidence intervals, and in Section \ref{sec: predictive_performance}, we discuss the predictive performance and the consistency of calibrated predictions.
    We conclude the theoretical part of the paper with Section \ref{sec: interval_width}, where we provide asymptotic results about the width of calibrated confidence intervals.
    Section \ref{sec: simulation} contains an implementation of our methods on a simulated dataset. Section \ref{sec: case_study} presents a case study on third-party motor liability claims that demonstrates the adaptivity of our method to different types of data and unbalanced datasets, such as those often seen in insurance.
    Code and replication material for the simulations and the case study are available at \url{https://github.com/GGavrilos/Calibrated_Intervals}.

    \section{Calibrated confidence intervals}\label{sec: main_result}

    Throughout this paper, we assume that \(X\) is a random covariate vector taking values in a standard Borel space \(\X\) and that \(Y\) is a real-valued random variable.
    Both \(X\) and \(Y\) are defined on a probability space \((\Omega, \F, \Prob)\).
    We denote the joint distribution of \((X,Y)\) by \(\prob_{(X,Y)}\).
    We observe realizations of this random pair, which we divide into a training set \(\S_1 = \{(X_{-i},Y_{-i})\}_{i=0}^r,\)
    and an independent calibration set \(\S_2 = \{(X_i,Y_i)\}_{i=1}^n\).
    We also observe a test covariate point \(X_{n+1}\), but not its label \(Y_{n+1}\).
    We denote by \(\Prob_n\) the empirical distribution of the calibration set and by \(\Prob_{n+1}\) the (unknown) empirical distribution of the augmented set \(\S_2\cup \{(X_{n+1},Y_{n+1})\}\).
    For a given \(y\in \R\), we denote by \(\Prob_{y}\) the empirical distribution of the augmented set \(\S_2\cup\{(X_{n+1},y)\}\).

    We say that a family \(\{Z_i\}_{i=1}^N\) of random variables is exchangeable if, for any fixed permutation \(\sigma\) of \(\{1,\ldots,N\}\), it holds that
    \(\big(Z_1,\ldots,Z_{N}\big) \overset{d}{=} \big(Z_{\sigma(1)},\ldots,Z_{\sigma(N)}\big)\).
    Exchangeability is weaker than the assumption that \(Z_1,\ldots,Z_N\) are iid.
    However, exchangeable points are still identically distributed.
    One important implication of exchangeability is that, conditionally on the empirical measure \(\Prob_{n+1}\), the distribution of each point \((X_i,Y_i)\) is \(\Prob_{n+1}\) itself \citep[Proposition~2.2]{Angelopoulos_2025}.
    In short,
    \begin{equation}
        \label{eq: exchangeability_implication}
        (X_i,Y_i)\, | \, \Prob_{n+1}\sim \Prob_{n+1},\quad \text{for all } i\in \{1,\ldots,n+1\}.
    \end{equation}
    This means that, under exchangeability, points are indistinguishable.
    In the following, we assume that \(\S_2\cup\{(X_{n+1},Y_{n+1})\}\) is exchangeable and independent from \(\S_1\).

    \subsection{In-sample calibration and binning procedures} 

    Out-of-sample calibration is a strong property.
    For binary classification, \citet{Gupta_2020} showed that calibration is impossible without access to \(\prob_{(X,Y)}\).
    As an alternative, they proposed asymptotic calibration, which can be achieved using binning strategies.
    A property related to asymptotic calibration is in-sample calibration, introduced by \citet{wuetr_ziegel_23}.

    \begin{defn}[In-Sample Calibration]
        An \emph{in-sample calibrated procedure} is a (possibly random) mapping $G$ that assigns to each sample \((X_1,Y_1),\ldots,(X_n,Y_n)\in \X\times \R\) a function \(G(\Prob_n,\cdot):\X\to \R\) such that
        \begin{equation}
            \label{eq: in_sample_calibration}
            \Exp_{(X,Y)\sim \Prob_n}\Big[Y -G(\Prob_n, X)\ \Big| \ G(\Prob_n, X)\Big] = 0,\quad \Prob_n\text{-almost surely}. 
        \end{equation}
    \end{defn}

    \begin{rem}
        \label{rem: finite_values}
        In the definition of in-sample calibration, we assume that \((X,Y)\sim \Prob_n\).
        Therefore, in-sample calibration only depends on the values of \(G(\Prob_n, \cdot)\) at \(X_1,\ldots,X_n\).
        Changing any value \(G(\Prob_{n}, x)\) for \(x\notin \{X_1,\ldots,X_n\}\) does not affect the validity of \eqref{eq: in_sample_calibration}.
        Thus, we may assume that \(G(\Prob_{n}, \X)=\{G(\Prob_{n}, X_1),\ldots,G(\Prob_{n}, X_n)\}\).
        Henceforth, when referring to in-sample calibrated procedures, we implicitly exclude procedures \(G\) that do not have this form.
    \end{rem}
    
    A map \(G\) is called a \emph{binning procedure} if it assigns to any sample \(\D = \{(X_i,Y_i)\}_{i=1}^n\) a function \(G_{\D}:\X\to \R\) that takes finitely many values \(\beta_1,\ldots,\beta_s\), which may depend on \(\D\).
    Equivalently, there exists an integer \(s_{\D}\geq 1\) and a \(\D\)-dependent partition \(B_1,\ldots,B_s\) of \(\X\), such that \(G_{\D}(x)=\sum_{k=1}^{s_{\D}} \beta_k \one\{x\in B_k\}\), where \(\beta_1,\ldots,\beta_{s_{\D}}\) are distinct.
    The following proposition shows that in-sample calibrated procedures and binning procedures are very closely related.

    \begin{prop}
        \label{prop: equivalence_in_sample_binning}
        Any in-sample calibrated procedure is a binning procedure.
        Conversely, a binning procedure \(G\) with representation
        \begin{equation}
            \label{eq: binning_procedure}
            G(\Prob_{n},x)=\sum_{k=1}^{s_{\D}} \beta_k \one\{x\in B_k\}
        \end{equation}
        is in-sample calibrated if and only if \(\beta_k = \left(\sum_{i\in D_k} Y_i\right)/(\#D_k)\),
        where \(D_k=\{i\in \{1,\ldots,n\}: X_i\in B_k\}\) and \((X_i,Y_i)\) denotes the \(i\)-th point of the sample \(\D\).
    \end{prop}

    \begin{proof}
        By Remark \ref{rem: finite_values}, it follows that any in-sample calibrated function takes finitely many values, so it is a binning procedure.
        For the converse, the binning procedure \eqref{eq: binning_procedure} is in-sample calibrated if and only if, for all indices \(k\in \{1,\ldots,s_{\D}\}\), it holds that
        \begin{align*}
            \mathbb{E}_{(X,Y)\sim \Prob_n}\Big[Y\ \Big | \  G(\Prob_{n},X)=\beta_k\Big]=\beta_k
            &\Leftrightarrow \frac{\Exp_{(X,Y)\sim \Prob_n}\Big[Y\cdot \one\{G(\Prob_{n},X)=\beta_k\}\Big]}{\Prob_n(G(\Prob_{n},X)=\beta_k)}=\beta_k\\
            &\Leftrightarrow \frac{\frac{1}{n}\sum_{i=1}^n Y_i \one\{G(\Prob_{n},X_i)=\beta_k\}}{\frac{1}{n}\sum_{i=1}^n \one\{G(\Prob_{n},X_i)=\beta_k\}}=\beta_k\\
            & \Leftrightarrow \frac{\sum_{i\in D_k} Y_i}{\# D_k}=\beta_k.\qquad \qedhere
        \end{align*}
    \end{proof}

    Hence, in-sample calibrated procedures are those that partition \(\X\) into bins \(B_1,\ldots,B_{s_{\D}}\) and average out the responses in every bin.
    We now give some common examples.

    \begin{example}
        \label{ex: regression_trees}
        Given \((X_1,Y_1),\ldots,(X_n,Y_n)\in \X\times \R\), where \(\X\subseteq \R^d\), a regression tree \(T:\X\to \R\) partitions \(\X\) into disjoint regions \(R_1,\ldots,R_s\).
        For a given input \(x\in R_k\), the tree outputs
        \begin{equation}
            \label{eq: regression_tree}
            T(x)=\frac{\sum_{i\in D_k} Y_i}{\# D_k}=\vcentcolon r_k,
        \end{equation}
        where \(D_k=\{i\in \{1,\ldots,n\}:X_i\in R_k\}\).
        The partition is chosen in such a way that the sum of squared errors \(\sum_{k=1}^s \sum_{i\in D_k} (Y_i-r_k)^2\) is minimized.
        In practice, it is computationally infeasible to search through all possible partitions, so trees make use of greedy algorithms to find simple partitions.
        In any case, \eqref{eq: regression_tree} is satisfied, so trees are in-sample calibrated.
    \end{example}

    \begin{example}
        \label{ex: isotonic_regression}
        Isotonic regression \citep{Ayer_1955} is relevant when \(\X\) is a subset of a totally ordered space and it is expressed through the minimization problem
        \begin{equation}
            \label{eq: isotonic_regression_minimization}
            \argmin_{T}\sum_{i=1}^n (Y_i-T(X_i))^2\quad \text{over all increasing functions }T:\X\to \R.
        \end{equation}
        Isotonic regression can be generalized to partially ordered spaces, but in this paper we restrict ourselves to the special case of a total order.
        \citet{BarlowBartholomewETAL1972} proved the existence of a solution \(\widehat{T}\), determined by a data-dependent partition \(D_1,\ldots,D_s\) of the set \(\{1,\ldots,n\}\).
        For an index \(j\in D_k\), it holds that
        \begin{equation}
            \label{eq: isotonic_regression_solution}
            \widehat{T}(X_j)=\frac{\sum_{i\in D_k} Y_i}{\# D_k}.
        \end{equation}
        Then, \(\widehat{T}\) can be extended to a right-continuous step function over the entire domain \(\X\), with jumps only at observed design points.
        By \eqref{eq: isotonic_regression_solution}, \(\widehat{T}\) is in-sample calibrated.

        Isotonic regression is going to play a central role in the development of our methods.
        One of its advantages over regression trees is that \eqref{eq: isotonic_regression_minimization} can be solved exactly using the PAV Algorithm \citep{BarlowBartholomewETAL1972}, whose complexity is \(O(n)\) if \(X_1,\ldots,X_n\) are already sorted \citep{Grotzinger_1984}.
        Accelerated versions have also been proposed \citep{Henzi_2022}.

        Another advantage of isotonic regression is its hyperparameter-free nature.
        Other binning methods, such as regression trees, require pruning or hyperparameter tuning.
        $k$-nearest neighbours and clustering methods also rely on hyperparameters related to the model complexity.
    \end{example}

    \subsection{Isotonic calibration}

    Isotonic regression offers in-sample calibration guarantees, but it is not recommended for prediction because it is applicable only when \(\X\subseteq \R\) and its performance relies heavily on the validity of the isotonic shape constraint.
    Nevertheless, it can be used as a post-processing, \emph{recalibration} tool.
    The idea to use isotonic regression for recalibration in binary classification originated in \citet{Zadrozny_2002}, who used it to recalibrate probabilistic predictions made by SVMs, Naive Bayes models and decision trees.
    This method was later refined by \citet{Vovk_2014} and it was generalized to a continuous setting by \citet{wuetr_ziegel_23}.
    
    The main idea is to fit isotonic regression to the sample \((\fhat(X_1),Y_1),\ldots,(\fhat(X_n),Y_n)\),
    where \(\fhat(X)\) is any prediction for the conditional mean \(\Exp[Y|X]\).
    If \(\widehat{T}:\R\to \R\) denotes the isotonic regression fit, then \(\widehat{T}(\fhat(X_i))\) is in-sample calibrated, because the binning procedure \(\widehat{T}\) satisfies the conditions of Proposition \ref{prop: equivalence_in_sample_binning}; see Example \ref{ex: isotonic_regression}.
    This method is known as \emph{isotonic calibration}.
    
    Isotonic calibration combines the calibration properties of isotonic regression with the predictive power of \(\fhat\).
    It relies on the assumption that a flexible model, like a neural network or a boosted tree, can capture the order relationships among \(\Exp[Y| X=x_1], \ldots, \Exp[Y| X=x_n]\), i.e.
    \begin{equation}
        \label{eq: estimating_order}
        \text{if }\fhat(X_i)\leq \fhat(X_j), \text{ then } \Exp[Y| X=X_i]\leq \Exp[Y| X=X_j].
    \end{equation}
    In practice, we would like to use a model \(\fhat\) that satisfies \eqref{eq: estimating_order} for as many pairs \(i,j\in \{1,\ldots,n\}\) as possible, so that isotonic regression is appropriate for the points \(\{(\fhat(X_i),Y_i)\}_{i=1}^n\).

    One advantage of isotonic calibration is that it can be used on top of any predictive model \(\fhat\), which enables its application to a vast variety of domains and data types.
    Additionally, it does not compromise predictive performance.
    As \citet{Vander_Laan_2024} noted, the in-sample performance of \(\widehat{T}(\fhat(X_i))\) is at least as good as that of \(\fhat(X_i)\), as isotonic regression minimizes \(\sum_{i=1}^n (Y_i-T(\fhat(X_i)))^2\) over all increasing functions \(T\), including the identity map.
    
    On the downside, \(\widehat{T}(\fhat(X_{n+1}))\) is only in-sample calibrated.
    This is not sufficient for most practical scenarios, where \((X_{n+1},Y_{n+1})\) is actually sampled from the same distribution as \(\{(X_i,Y_i)\}_{i=1}^n\) and not from their empirical distribution.

    \subsection{Calibrated confidence intervals}\label{subsec: calibrated_intervals}

    In this subsection, we present our main theoretical result.
    We show that, starting from an observed sample \(\{(X_i,Y_i)\}_{i=1}^n\) and \(X_{n+1}\), an in-sample calibrated procedure \(G\), and a prediction set for \(Y_{n+1}\), we can construct a calibrated confidence set, i.e. a set that satisfies \eqref{eq: calibrated_confidence_intervals}.
    This holds under the following assumption:

    \begin{assumption}
        \label{assum: A1}
        The random variables \(Y_1,\ldots,Y_{n+1}\) have a finite first moment.
    \end{assumption}
    The proof of the main theorem can be found in Appendix \ref{sec: appA}.
    
    \begin{thm}\label{thm:I}
    Let \(\S_1\) be a training sample and let \(\{(X_i,Y_i)\}_{i=1}^{n+1}\) be exchangeable and independent from \(\S_1\), with all variables defined on a probability space \((\Omega, \F, \Prob)\).
    Let \(G\) be an in-sample calibrated procedure that depends only on \(\S_1\).
    Denote by \(\Prob_{n+1}\) and \(\Prob_{y}\) the empirical distributions of \((X_1,Y_1),\ldots,(X_{n+1},Y_{n+1})\) and \((X_1,Y_1),\ldots,(X_n,Y_n),(X_{n+1},y)\)
    respectively, where \(y\in \R\).
    Furthermore, let \(I_\alpha\left(X_{n+1}\right)\) be a (random) set that satisfies 
    \begin{equation}\label{eq:Ialpha}
    \Prob\left(Y_{n+1}\in I_\alpha\left(X_{n+1}\right)\, | \, \A \right) \geq 1-\alpha, \quad \text{almost surely},
    \end{equation}
    for some $\sigma$-algebra \(\A\subseteq \F\).
    Set
    \begin{equation}
        \label{eq: C_alpha}
        \C_{n+1,\alpha}=\Big\{G(\Prob_{y},X_{n+1}) : y \in I_\alpha(X_{n+1})\Big\}.
    \end{equation}
    Under Assumption \ref{assum: A1}, it holds that \(G(\Prob_{n+1},X_{n+1})\) is calibrated and that
    \begin{equation}
        \label{eq: C_alpha_guarantee}
        \Prob\Big(G(\Prob_{n+1},X_{n+1}) \in \C_{n+1,\alpha} \, \Big|\, \A\Big) \ge 1-\alpha, \quad \text{ almost surely}.
    \end{equation}
    \end{thm}
    
    The proof that \(G(\Prob_{n+1},X_{n+1})\) is calibrated generalizes \citet[Theorem~4.1]{Vander_Laan_2024}, who assume the existence of a finite second moment.
    We also provide a different proof idea, which only uses the basic properties of exchangeability, expressed through \eqref{eq: exchangeability_implication}.

    The construction of $\C_{n+1,\alpha}$ necessitates a prediction set $I_\alpha(X_{n+1})$ for $Y_{n+1}$.
    In Subsection \ref{subsec: motivating_prediction_set}, we provide some further insight on the importance of this set, and on possible alternatives that do not make use of it.
    In theory, there is no restriction on how \(I_\alpha(X_{n+1})\) is constructed, but in practice, conformal prediction would be the most reliable method in non-parametric settings.
    Thus, we may think of \(I_\alpha(X_{n+1})\) as a conformal prediction set.
    If \(I_\alpha(X_{n+1})\) has conditional coverage guarantees, these are passed on to $\C_{n+1,\alpha}$.
    Even though covariate-conditional coverage has been shown to be impossible in totally agnostic settings \citep{Barber_2023}, several variants such as Mondrian conformal prediction \citep{Vovk_2005}, conformalized quantile regression \citep{Romano_2019}, and distributional conformal prediction \citep{Chernozhukov_2021} provide certain guarantees in this direction.
    Recently, \citet{Gibbs_2025} and \citet{Duchi_2025} developed conformal prediction sets with stronger conditional coverage guarantees.
    In Remark \ref{rem: CP_Gibbs_2025} we explain how some of these stronger properties also carry over to \(\C_{n+1,\alpha}\).

    The prediction sets of \citet{Vander_Laan_2024}, which offer coverage of \(Y_{n+1}\) conditionally on \(G(\Prob_{n+1},X_{n+1})\), can also be used, but they can only be approximated.
    If the \(Y\)-space is bounded, discretization may offer sufficient approximation accuracy, albeit at a high computational cost, incurred by having to run isotonic calibration for many different values of \(y\).
    As we will show, computing our calibrated confidence set \(\C_{n+1, \alpha}\) does not require iterating over \(y\in I_{\alpha}(X_{n+1})\), which reduces the cost significantly.

    \section{Methodology for calibrated confidence intervals} \label{sec: methodology}

    To construct the calibrated confidence intervals of Theorem \ref{thm:I}, we need a calibration procedure \(G\) and a prediction set \(I_{\alpha}(X_{n+1})\) for \(Y_{n+1}\).
    In this section, we provide methodological details on the choice of \(G\) and on the effects of this choice on the resulting calibrated confidence intervals.

    In Subsections \ref{subsec: histogram_binning} and \ref{subsec: isotonic_calibration_methodology}, we present two different calibration procedures, inspired by Examples \ref{ex: regression_trees} and \ref{ex: isotonic_regression}.
    In Subsections \ref{subsec: motivating_prediction_set} and \ref{subsec: best_candidates}, we provide further methodological insight into calibrated confidence intervals, including comments on the choice of the prediction set \(I_\alpha(X_{n+1})\), and on best-candidate point predictions from the calibrated confidence interval \(\C_{n+1,\alpha}\).

    \subsection{Calibration via histogram binning}\label{subsec: histogram_binning}

    We begin with a purposefully simple binning procedure, which is described by Algorithm \ref{alg: simple_binning}.
    We do not expect it to perform well in practice, because it does not adapt well to the calibration dataset.
    However, it allows for explicit formulation of the calibrated confidence intervals.
    \begin{algorithm}[ht]
        \caption{Calibrated confidence intervals (simple binning)}
        \label{alg: simple_binning}
        \begin{algorithmic}[1]
            \Require Training set \(\S_1\), calibration set \(\S_2\), test point \(X_{n+1}\).
            \State Estimate \(f^\star(x)\vcentcolon = \Exp[Y|X=x]\) by fitting a model \(\fhat\) on \(\S_1\).
            \State Determine a partition \(B_1,\ldots,B_K\) of \(\R\) using \(\S_1\).
            \State Compute \(\fhat(X_i)\) for all indices \(i\in \{1,\ldots,n+1\}\) and determine the bins
            \begin{equation*}
                D_j \vcentcolon = \left\{i\in \{1,\ldots,n+1\}\, | \, \fhat(X_i)\in B_j\right\},\quad j=1,\ldots,K.
            \end{equation*}
            \Ensure Given a compact prediction interval \(I_{\alpha}(X_{n+1})=[\ell_{n+1},h_{n+1}]\), output
            \begin{equation}
                \label{eq: calibrated_interval_binning}
                \C_{n+1,\alpha} = \left[\frac{1}{\# D_{j^\star}}\left(\ell_{n+1} + \sum_{i \in D_{j^\star}\setminus\{n+1\}} Y_i\right), \frac{1}{\# D_{j^\star}}\left(h_{n+1} + \sum_{i \in D_{j^\star}\setminus\{n+1\}} Y_i\right)\right],
            \end{equation}
            where \(j^\star\) is such that \(n+1\in D_{j^\star}\).
        \end{algorithmic}
    \end{algorithm}
    
    The output of this algorithm matches the calibrated confidence interval \eqref{eq: C_alpha} in Theorem \ref{thm:I}.
    Indeed, the relevant in-sample calibration procedure, evaluated at \(X_{n+1}\), is given by
    \begin{equation*}
        G(\Prob_y, X_{n+1})=\frac{1}{\# D_{j^\star}}\left(y + \sum_{i\in D_{j^\star}\setminus \{n+1\}} Y_i\right).
    \end{equation*}
    Therefore, by \eqref{eq: C_alpha}, it follows that \(\left\{\left(y + \sum_{j\in D_{j^\star}\setminus \{n+1\}} Y_j\right)/\# D_{j^\star}\, \middle| \, \ell_{n+1}\leq y\leq h_{n+1}\right\}\) is equal to the output of Algorithm \ref{alg: simple_binning} in \eqref{eq: calibrated_interval_binning}.
    From Theorem \ref{thm:I}, it follows that the calibrated prediction contained in \(\C_{n+1,\alpha}\) with high probability is given by \(G(\Prob_{n+1},X_{n+1})=(1/\# D_{j^\star})\sum_{i\in D_{j^\star}} Y_i\).
    Since \(Y_{n+1}\) is unobserved, we do not have access to this prediction in practice.

    \begin{rem}
        The binning procedure \(G(\Prob_y,\cdot)\) in Algorithm \ref{alg: simple_binning} is the following: given \(x\in \X\), let \(j_x\in \{1,\ldots,K\}\)  be an index such that \(\fhat(x)\in B_{j_x}\).
        Then, we define \(G(\Prob_y,x)\) to be the sample average of the responses with index in \(D_{j_x}\).
        According to Proposition \ref{prop: equivalence_in_sample_binning}, this procedure is in-sample calibrated.
        Therefore, \(\C_{n+1,\alpha}\) inherits the validity properties given by Theorem \ref{thm:I}.
    \end{rem}
    
    Algorithm \ref{alg: simple_binning} provides total freedom in the choice of the partition \(B_1,\ldots,B_K\).
    One possibility is to fit isotonic regression to \(\{(\fhat(X_{-i}),Y_{-i})\}_{i=0}^r\), and use the associated bins, or to partition the covariate space \(\X\) using clustering methods or regression trees.
    All these options suffer from various weaknesses.
    Apart from the fact that the partition \(B_1,\ldots,B_K\) does not adapt to the calibration set, clustering and tree-based models rely additionally on hyperparameters, such as the number of bins, the minimum number of leaves per node, and the maximum tree depth.
    
    The width of \(\C_{n+1,\alpha}\) is equal to \(w(\C_{n+1,\alpha})=(h_{n+1}-\ell_{n+1})/\# D_{j^\star}\).
    This depends both on the width of the prediction set \(I_{\alpha}(X_{n+1})=[\ell_{n+1},h_{n+1}]\) and the size of the bin \(D_{j^\star}\).
    This observation recovers a problem already noted by earlier works, such as \citet{Gupta_2020}, namely that the width of the calibrated prediction interval inflates if the bin \(D_{j^\star}\) contains only a few points.
    As we explain in Section \ref{sec: predictive_performance}, this has an important effect on the convergence properties of the associated calibrated confidence intervals.

    \citet{Gupta_2020} proposed the construction of bins with large numbers of points.
    Although this solves some of the above problems and provides upper bounds for the bin-related calibration error, it still relies on hyperparameters.
    It is also unclear how much the prediction error of the base model \(\fhat\) is compromised under a recalibration based on simple binning.

    \subsection{Calibration via isotonic regression}\label{subsec: isotonic_calibration_methodology}

    We now show how we can construct calibrated confidence intervals using ideas based on isotonic calibration.
    This method is adaptive, hyperparameter-free, and it has reasonable computational complexity.
    It is described in Algorithm \ref{alg: isotonic_recalibration}.
    As we show in Sections \ref{sec: predictive_performance} and \ref{sec: interval_width}, it is also consistent, and it produces calibrated intervals with asymptotically vanishing width.

    \begin{algorithm}
        \caption{Calibrated confidence intervals (isotonic regression)}
        \label{alg: isotonic_recalibration}
        \begin{algorithmic}[1]
            \Require Training set \(\S_1\), calibration set \(\S_2\), test point \(X_{n+1}\).
            \State Estimate \(f^\star(x)\vcentcolon = \Exp[Y|X=x]\) by fitting a model \(\fhat\) on \(\S_1\).
            \State Given a compact prediction interval \(I_{\alpha}(X_{n+1})=[\ell_{n+1},h_{n+1}]\), denote by \(T_{\ell},T_h\) the isotonic regression fits to the augmented sets
            \begin{equation*}
                (\fhat(X_1),Y_1),\ldots,(\fhat(X_{n+1}),\ell_{n+1})\quad \text{and}\quad (\fhat(X_1),Y_1),\ldots,(\fhat(X_{n+1}),h_{n+1}).
            \end{equation*}
            \Ensure The interval
            \begin{equation}
                \label{eq: calibrated_interval_isotonic}
                \C_{n+1,\alpha} = \left[T_{\ell}\left(\fhat(X_{n+1})\right),T_h\left(\fhat(X_{n+1})\right)\right].
            \end{equation}
        \end{algorithmic}
    \end{algorithm}

    To see why the output \(\C_{n+1,\alpha}\) of Algorithm \ref{alg: isotonic_recalibration} is a special case of the calibrated interval of Theorem \ref{thm:I}, notice that the in-sample calibrated procedure used in Algorithm \ref{alg: isotonic_recalibration}, evaluated at \(\fhat(X_{n+1})\), is given by \(G(\Prob_y, X_{n+1}) =\vcentcolon \iso_{n+1}\left[(\fhat(X_1),Y_1),\ldots,(\fhat(X_{n+1}),y)\right]\), where \(\iso_j(\{(w_i,z_i)\}_{i=1}^{k})\) denotes the isotonic regression fit to the points \(\{(w_i,z_i)\}_{i=1}^k\), evaluated at \(w_j\).
    Isotonic regression has the property that, if \(z_1\leq z_1^\prime, \ldots, z_k\leq z_k^\prime\), then
    \begin{equation}
        \label{eq: isotonic_regression_ordering}
        \iso_j((w_1,z_1),\ldots,(w_k,z_k))\leq \iso_j((w_1,z_1^\prime),\ldots,(w_k,z_k^\prime)).
    \end{equation}
    By \eqref{eq: C_alpha}, \eqref{eq: isotonic_regression_ordering} and continuity, it follows that the output of Algorithm \ref{alg: isotonic_recalibration} is equal to
    \begin{equation*}
        \left[T_{\ell}\left(\fhat(X_{n+1})\right),T_h\left(\fhat(X_{n+1})\right)\right] = \left\{\iso_{n+1}\left((\fhat(X_1),Y_1),\ldots,(\fhat(X_{n+1}),y)\right) \, \middle| \, \ell_{n+1}\leq y\leq h_{n+1}\right\},
    \end{equation*}
    exactly as in Theorem \ref{thm:I}.
    The calibrated prediction contained in \(\C_{n+1,\alpha}\) with high probability is  \(G(\Prob_{n+1},X_{n+1})=\iso_{n+1}\left((\fhat(X_1),Y_1),\ldots,(\fhat(X_{n+1}),Y_{n+1})\right)\), which is unknown at prediction time.
    If there are ties among \(\fhat(X_1),\ldots,\fhat(X_{n+1})\), we merge the corresponding points, we replace their labels by their average, and we run Algorithm \ref{alg: isotonic_recalibration} using weighted isotonic regression.
    For simplicity, we assume that there are no ties.

    Expression \eqref{eq: calibrated_interval_isotonic} offers a substantial reduction in computational cost and enables exact computation of the calibrated confidence intervals \(\C_{n+1,\alpha}\).
    Instead of having to fit isotonic regression to \((\fhat(X_1),Y_1),\ldots,(\fhat(X_{n+1}),y)\) for every point \(y\in I_\alpha(X_{n+1})\), it is only necessary to do this twice, once for every endpoint of the prediction set \(I_\alpha(X_{n+1})\).
    This becomes possible due to the natural order induced by isotonic regression, as expressed by \eqref{eq: isotonic_regression_ordering}.
    Without this property, we would have to resort to approximations of \(\C_{n+1,\alpha}\) via a discretization of \(I_\alpha(X_{n+1})\), which would also come with a higher computational complexity.

    \subsection{Importance of the prediction set \(I_\alpha(X_{n+1})\)} \label{subsec: motivating_prediction_set}
    
    In Theorem \ref{thm:I}, we can replace the prediction set \(I_\alpha(X_{n+1})\) by the entire support \(I\) of \(Y\).
    In this case, the theorem would yield the calibrated confidence interval \(\C_{n+1}=\Big\{G(\Prob_y,X_{n+1}):y\in I\Big\},\) which contains a calibrated prediction almost surely.
    For binary classification, this is the main idea behind Venn-Abers predictors \citep{Vovk_2014}.
    More generally, if \(I\) is bounded, \(\C_{n+1}\) provides a good alternative to \(\C_{n+1,\alpha}\), particularly when almost-sure coverage is necessary.
    
    However, if \(I\) is unbounded, the set \(\C_{n+1}\) is also unbounded.
    Indeed, in both algorithms, \(G(\Prob_y, X_{n+1})=(1/\# D)\left(y+\sum_{i\in D\setminus\{n+1\}} Y_i\right)\),
    where \(D\) denotes the bin containing \((\fhat(X_{n+1}),y)\).
    If \(y\) ranges over an unbounded set, \(G(\Prob_y,X_{n+1})\) also takes arbitrarily small and large values.
    The use of the prediction set \(I_\alpha(X_{n+1})\) instead of \(I\) is one of the key elements that differentiates our method from other related works and it offers higher stability and informativeness. 
    
    \subsection{Best candidates} \label{subsec: best_candidates}

    Even though the main goal of this paper is to provide uncertainty quantification around calibrated predictions, decision-making often also relies on point predictions.
    In classical statistics, confidence intervals are centered around existing point predictions.
    Our setting is the opposite: confidence intervals contain a calibrated prediction, but this prediction is inaccessible.
    A natural question that arises is how to gauge this calibrated prediction, using only the available data.
    In this subsection, we answer this question by making some methodologically motivated proposals about how to choose a \emph{best candidate} from the calibrated confidence interval.
    
    For binary classification, \citet{Gupta_2020} show that the midpoint of the confidence interval \(\C_{n+1,\alpha}\), denoted by \(\text{mid}(\C_{n+1,\alpha})\), is approximately calibrated.
    This proposal can naturally be extended beyond binary responses.
    Especially if \(\C_{n+1,\alpha}\) has a small width, the midpoint is a reasonable choice.
    However, it may not be reliable in cases of high uncertainty.

    The recent paper by \citet{Vander_Laan_2024} also searched for a reliable best-candidate prediction.
    Borrowing ideas from \citet{Vovk_2014}, they proposed an adjusted-midpoint solution.
    For our method, this solution would be expressed as
    \begin{equation}
        \label{eq: best_candidate_laan}
        \ftilde_{n+1}(X_{n+1})\vcentcolon = \text{mid}(\C_{n+1,\alpha}) + \frac{T_h\left(\fhat(X_{n+1})\right)-T_{\ell}\left(\fhat(X_{n+1})\right)}{\sup I_\alpha(X_{n+1})-\inf I_\alpha(X_{n+1})}\left(\frac{1}{n}\sum_{i=1}^n Y_i - \text{mid}(\C_{n+1,\alpha})\right).
    \end{equation}
    The terms \(T_\ell, T_h\) are defined in Algorithm \ref{alg: isotonic_recalibration}.
    The idea behind this proposal is to look at the sensitivity of \(\iso_{n+1}\left((\fhat(X_1),Y_1),\ldots,(\fhat(X_{n+1}),y)\right)\) in \(y\).
    If this isotonic regression fit is robust against variations in \(y\), then the set \(\C_{n+1,\alpha}= \left[T_\ell\left(\fhat(X_{n+1})\right), T_h\left(\fhat(X_{n+1})\right)\right],\)
    will have a small width compared to \(I_\alpha(X_{n+1})\).
    Therefore, the midpoint prediction is reliable.
    Accordingly, the multiplicative factor
    \begin{equation*}
        \frac{T_h\left(\fhat(X_{n+1})\right)-T_{\ell}\left(\fhat(X_{n+1})\right)}{\sup I_\alpha(X_{n+1})-\inf I_\alpha(X_{n+1})}
    \end{equation*}
    vanishes, and \(\ftilde_{n+1}(X_{n+1})\approx \text{mid}(\C_{n+1,\alpha})\).
    On the contrary, high sensitivity in \(y\) would indicate that the confidence interval is too wide, and any possible \emph{guess} has a low chance of being close to the calibrated prediction.
    In that case, the value of the multiplicative factor will approximate \(1\), resulting in \(\ftilde_{n+1}(X_{n+1})\) being approximately equal to the marginal average \(\frac{1}{n}\sum_{i=1}^n Y_i\).

    Both midpoint-based candidates suffer from certain weaknesses.
    First, the only property that motivates the use of \(\text{mid}(\C_{n+1,\alpha})\) is that it is \emph{worst-case optimal}, minimizing the maximum distance from other points of \(\C_{n+1,\alpha}\).
    However, this is a generic property of the midpoint, without any case-specific insight.
    Apart from that property, there is no methodological motivation for using the midpoint or midpoint-based best-candidate predictions.

    Another weakness of \eqref{eq: best_candidate_laan} is that, under moderate or high uncertainty, it reduces to the average \(\overline{Y}_n\vcentcolon = \frac{1}{n}\sum_{i=1}^n Y_i\), which may lie far from \(\C_{n+1,\alpha}\).
    \citet{Vander_Laan_2024} touch upon that issue by explaining that \(\overline{Y}_n\) can be replaced by any other reference predictor.

    Our proposed best candidate relies on the in-sample calibrated prediction introduced by \citet{wuetr_ziegel_23}.
    This prediction is the isotonic regression fit \(G(\Prob_n,\cdot)\) to the points \(\{(\fhat(X_i),Y_i)\}_{i=1}^n\).
    By Theorem \ref{thm:I}, the out-of-sample calibrated prediction is given by the isotonic regression fit \(G(\Prob_{n+1},X_{n+1})\) to the full sample \(\{(\fhat(X_i),Y_i)\}_{i=1}^{n+1}\), evaluated at \(X_{n+1}\).
    Adding only one point will likely not change the isotonic regression fit dramatically, so \(G(\Prob_n,\cdot)\) should be approximately equal to \(G(\Prob_{n+1},\cdot)\).
    
    In case we only want to report predictions in \(\C_{n+1,\alpha}\), we can clip \(G(\Prob_n,\cdot)\) at the lower and upper endpoints of that interval.
    This would yield the point prediction
    \begin{equation}
        \label{eq: best_candidate}
        \ftilde_{\text{clip}}(X_{n+1})\vcentcolon = \min\Big\{\sup (\C_{n+1,\alpha}), \max\left\{G(\Prob_n, X_{n+1}), \inf (\C_{n+1,\alpha})\right\} \Big\}.
    \end{equation}
    This proposal is motivated by Proposition \ref{prop: over_under_calibration}, which shows that with high probability, the predictions \(T_\ell(\fhat(X_{n+1})), T_h(\fhat(X_{n+1}))\) are \emph{under-calibrated} and \emph{over-calibrated} respectively.
    
    All candidates proposed in this subsection are motivated by the statistical convention of coupling uncertainty quantification with point predictions and are largely based on heuristic arguments.
    As shown in Section \ref{sec: case_study}, these candidates are not always reliable in practice.

    \section{Consistency of isotonic recalibration} \label{sec: predictive_performance}

    In binary classification, the concept of calibration draws much of its credibility from the fact that the true success probability \(p^\star(x) = \Prob(Y=1|X=x)\) is calibrated.
    This is one of the reasons why calibration is considered important when evaluating forecasts using proper scoring rules \citep{Dawid_1984}.
    This property continues to hold in the continuous setting, for the true conditional mean \(\mu(x)=\Exp[Y|X=x]\).
    Indeed, the \(\sigma\)-algebra generated by \(\mu(X)\) is smaller than the one generated by \(X\), so \(
        \Exp[Y\, |\, \mu(X)]
        = \Exp\Big[\Exp[Y\, |\, X]\, \Big | \, \mu(X)\Big]
        = \Exp[\mu(X)\, |\, \mu(X)]
        = \mu(X)\).

    In this section, we provide insight on the predictive power of calibrated predictions.
    First, we derive non-asymptotic bounds for the prediction error of predictions obtained by Algorithms \ref{alg: simple_binning} and \ref{alg: isotonic_recalibration} under the assumption that \(\fhat\) is equal to the true conditional mean of \(Y\) given \(X\).
    In the second part, we show that, under structural assumptions on the base model \(\fhat\), calibrated predictions based on isotonic regression are consistent for the true conditional mean.
    
    \begin{lem}
        \label{lem: predictive_power_variance}
        Let \((X,Y)\in \X\times \R\) be a random pair such that \(Y\in L^2\).
        For all \(x\in \X\), denote the conditional expectation \(\Exp[Y\, | \, X=x]\) by \(\mu(x)\).
        If \(\ftilde(X)\) is a calibrated prediction for \(Y\), then
        \begin{equation}
            \label{eq: predictive_power_variance}
            \Exp\Big[(Y-\mu(X))^2\Big] + \mathrm{Var}(\mu(X)) = \Exp\Big[(Y-\ftilde(X))^2\Big] + \mathrm{Var}(\ftilde(X)).
        \end{equation}
    \end{lem}

    Lemma \ref{lem: predictive_power_variance} shows that the expected squared error of calibrated predictions depends on their complexity (variance). 
    For example, the unconditional mean $\ftilde(X) = \Exp[Y]$ is trivially a calibrated prediction, but since its variance is zero, its prediction error on the right-hand side of \eqref{eq: predictive_power_variance} may be much larger than our target MSE on the left-hand side.
    In other words, good calibrated predictions should have higher complexity, matching that of \(\mu(X)\).
    The proof of Lemma \ref{lem: predictive_power_variance} can be found in Appendix \ref{sec: app_Proofs_4}.

    \subsection{Predictive power of calibrated predictions}
    \label{subsec: predictive_power}

    Algorithms \ref{alg: simple_binning}, \ref{alg: isotonic_recalibration} are applied on top of a prediction model \(\fhat\) and yield (unknown) calibrated predictions, denoted by \(G(\Prob_{n+1},X_{n+1})\).
    A natural question that could be posed is how the prediction error of those calibrated predictions compares to that of the original model \(\fhat\).
    Ideally, recalibration should not cause a decline in predictive performance.
    
    In this subsection we compare the two recalibration algorithms and we show that Algorithm \ref{alg: isotonic_recalibration} does not compromise predictive performance asymptotically.
    On the contrary, the performance of Algorithm \ref{alg: simple_binning} can vary, depending on the binning scheme.
    The proofs of all results can be found in Appendix \ref{sec: app_Proofs_4}.
    We make the following assumption:
    
    \begin{assumption}
        \label{assum: pred_performance_iid}
        For all \(i\in [n+1]\), it holds that \(Y_i = f^\star(X_i)+\varepsilon_i\).
        The noise terms \(\varepsilon_1,\ldots,\varepsilon_{n+1}\) are iid, independent from \(\S_1, \{X_i\}_{i=1}^{n+1}\), and they satisfy \(\Exp[\varepsilon_i]=0\), \(\mathrm{Var}(\varepsilon_i)=\vcentcolon \sigma_\varepsilon^2<\infty\) for all \(i\in [n+1]\).
    \end{assumption}
    
    The intention of this assumption is to present here a worst-case analysis of isotonic recalibration.
    We will show that, when \(\fhat=f^\star\), in which case it would be unnecessary to recalibrate \(\fhat\), Algorithm \ref{alg: isotonic_recalibration} outputs predictions with an asymptotically vanishing excess risk.
    For Algorithm \ref{alg: simple_binning} (simple binning), we can only infer a nonvanishing upper bound for the excess risk of \(\mathsf{c}(\fhat(X_{n+1}))\).
    In Appendix \ref{sec: app_Proofs_4}, we provide an example showing that this bound is tight, so the excess risk does not necessarily vanish asymptotically.
    
    \begin{prop}
        \label{prop: pred_performance_binning}
        Consider the setting of Algorithm \ref{alg: simple_binning} and let \(k\) denote the number of occupied bins.
        Assume that \(\fhat = f^\star\) and that \(X_1, \ldots,X_{n+1}\) are iid.
        If \(\mathrm{Var}(f^\star(X)) = \vcentcolon\sigma_{f^\star}^2<\infty\) and \(\mathsf{c}(f^\star(X_{n+1})) \vcentcolon = G(\Prob_{n+1},X_{n+1})=(1/\#D_{j^\star})\sum_{i \in D_{j^\star}} Y_i\), then under Assumption \ref{assum: pred_performance_iid},  it holds that
        \begin{equation}\label{eq: pred_loss_binning}
            \Exp\Big[\left(Y_{n+1}-\mathsf{c}(f^\star(X_{n+1}))\right)^2\Big]-\Exp\Big[(Y_{n+1}-f^\star(X_{n+1}))^2\Big]\leq \frac{n}{n+1}\sigma_{f^\star}^2-\frac{\Exp[k]}{n+1}\sigma_\varepsilon^2.
        \end{equation}
    \end{prop}
    
    The number of occupied bins \(k\) depends on the training set and on \(X_1,\ldots,X_{n+1}\).
    In the above proposition we omit this dependence for the sake of a lighter notation.
    
    Equation \eqref{eq: pred_loss_binning} provides an upper bound on the excess risk of \(\mathsf{c}(f^\star(X_{n+1}))\).
    When \(\Exp[k]\) is large, the excess risk can become negative.
    In particular, this happens if
    \begin{equation*}
        \frac{\Exp[k]}{n+1}\sigma_\varepsilon^2-\frac{n}{n+1}\sigma_{f^\star}^2>0\Longleftrightarrow \frac{\Exp[k]}{n}>\frac{\sigma_{f^\star}^2}{\sigma_{\varepsilon}^2}.
    \end{equation*}
    This may seem counterintuitive, but it is explained by the fact that \(\mathsf{c}(f^\star(X_{n+1}))\) is unobserved and it has access to the true response \(Y_{n+1}\), while \(f^\star(X_{n+1})\) does not.
    Especially if there are many small bins, in which case \(D_{j^\star}\) contains very few points, the true, unobserved value \(Y_{n+1}\) has a large influence on \(\mathsf{c}(f^\star(X_{n+1}))\).
    In the extreme case where \(D_{j^\star}=\{n+1\}\), it holds that \(\mathsf{c}(f^\star(X_{n+1}))=Y_{n+1}\).
    However, if \(\#D_{j^\star}\) is too small, \eqref{eq: calibrated_interval_binning} suggests that \(\C_{n+1,\alpha}\) will be very wide and uninformative, so it is generally not advised to enforce small bin sizes.

    The necessity to tune \(k\) and find the right tradeoff between predictive performance and informative calibrated confidence intervals is one of the main drawbacks of Algorithm \ref{alg: simple_binning}.
    On the contrary, Algorithm \ref{alg: isotonic_recalibration} provides automatic bin selection.
    As we show in the following proposition, it also provides asymptotic guarantees about the excess risk of \(\mathsf{c}(\fhat(X_{n+1}))\).

    \begin{prop}
        \label{prop: pred_performance_isotonic}
        Consider the setting of Algorithm \ref{alg: isotonic_recalibration} and assume that \(\fhat = f^\star\).
        Suppose that \(\mathrm{Var}(f^\star(X))=\vcentcolon \sigma_{f^\star}^2<\infty\) and set
        \begin{equation*}
            \mathsf{c}(f^\star(X_{n+1}))\vcentcolon = G(\Prob_{n+1},X_{n+1}) = \iso_{n+1}\left((f^\star(X_1),Y_1),\ldots,(f^\star(X_{n+1}),Y_{n+1})\right).
        \end{equation*}
        Under Assumption \ref{assum: pred_performance_iid}, there exists a constant \(C_{f^\star,\varepsilon}\) depending on \(\sigma_{f^\star},\sigma_{\varepsilon}\) such that
        \begin{align}
            \label{eq: pred_loss_isotonic}
             \Big|\Exp\Big[(Y_{n+1}-\mathsf{c}
             &(f^\star(X_{n+1})))^2\Big] - \Exp\Big[(Y_{n+1}-f^\star(X_{n+1}))^2\Big]\Big|\nonumber \\
             & \leq C_{f^\star,\varepsilon}\left[\Exp\left(\frac{\sigma_{\varepsilon}^2 \Big[\max_i f^\star(X_i)-\min_i f^\star(X_i)\Big]}{n+1}\right)^{2/3}+\frac{\sigma_{\varepsilon}^2\log (n+1)}{n+1}\right]^{1/2}.
        \end{align}
    \end{prop}

    \begin{rem}
        If \(\Exp[\max_i f^\star(X_i) - \min_i f^\star(X_i)]=o(n)\), the above proposition implies that the excess risk vanishes asymptotically.
        In particular, if the distribution of \(\{f^\star(X_i)\}_{i=1}^{n+1}\) has a bounded support, we obtain a convergence rate equivalent to \(n^{-1/3}\).
        The condition \(\Exp[\max_i f^\star(X_i) - \min_i f^\star(X_i)]=o(n)\) actually holds for all distributions with a finite second moment.
        \citet{Arnold_1979} showed that the \(k\)-th order statistic \(X_{(k)}\) of \(n\) (not necessarily independent) identically distributed random variables with mean \(\mu\) and variance \(\sigma^2\) is such that
        \begin{equation*}
            \mu-\sigma\left(\frac{n-k}{k}\right)^{1/2}\leq \Exp\left[X_{(k)}\right]\leq \mu+\sigma\left(\frac{k-1}{n-k+1}\right)^{1/2}.
        \end{equation*}
        This inequality yields \(\Exp\Big[\max_i f^\star(X_i) - \min_i f^\star(X_i)\Big] \leq 2\sigma_{f^\star}\sqrt{n+1}\),
        so, for any distribution with finite second moment, the loss in predictive performance vanishes asymptotically.
        This bound yields a rate of convergence equivalent to \(n^{-1/6}\).
    \end{rem}

    \subsection{Consistency of isotonic recalibration}
    \label{subsec: consistency}

    The results in the previous subsection provide non-asymptotic upper bounds for the prediction error of the calibrated prediction \(\mathsf{c}(\fhat(X_{n+1}))\) obtained by isotonic recalibration (Algorithm \ref{alg: isotonic_recalibration}).
    In this section, we show that \(\mathsf{c}(\fhat(X_{n+1}))\) is asymptotically consistent for the true conditional mean \(\Exp[Y_{n+1}|X_{n+1}]\), under structural assumptions on the base model \(\fhat\).
    
    Due to the asymptotic nature of the results of this section, we denote by \(\fhat_r\) the base model trained on a sample of size \(r\), and by \(\mathsf{c}_n(\fhat_r(X_{n+1}))\) the calibrated prediction, based on a calibration set of size \(n\).
    We make the following assumptions:

    \begin{assumption}\label{assum: consistency_iid}
    The random pairs \(\{(X_i,Y_i)\}_{i=1}^{\infty}\) are independent and identically distributed.
    \end{assumption}

    \begin{assumption}
    \label{assum: consistency_bounded_mu} The random variables \(\mu(X_{i})=\Exp[Y_i\, |\, X_i]\) are supported on a bounded interval \(I\) and they admit a density with respect to the Lebesgue measure, which is bounded from below by \(C_1>0\).
    \end{assumption}

    \begin{assumption}
    \label{assum: consistency_fhat} There exists a constant \(C_0>0\) and a strictly increasing function \(T:\R\to \R\) such that
    \begin{equation*}
        \lim_{r\to \infty}\Prob\left(\sup_{x\in \X}\left|T\Big(\fhat_r(x)\Big)-\mu(x)\right|\geq C_0\left(\frac{\log r}{r}\right)^{1/3}\right)=0.
    \end{equation*}
    \end{assumption}

    \begin{assumption}
    \label{assum: consistency_subgaussian}
    For all \(i\geq 1\), the variable \(Y_i-\mu(X_i)\) is sub-Gaussian conditionally on \(X_i\), with scale \(\sigma>0\).
    \end{assumption}

    Some comments on the interpretation of these assumptions can be found in Appendix \ref{sec: app_proof_consistency}.
    We now state the main consistency result, whose proof is also deferred to Appendix \ref{sec: app_proof_consistency}.

    \begin{thm}
        \label{thm: consistency}
        Let \(\mathsf{c}_n(\fhat_r(x))=G(\Prob_{n},x)\) be the in-sample calibrated prediction of \citet{wuetr_ziegel_23} and suppose that \(r\leq n\).
        Set \(\delta_r = (\log r/r)^{1/3}\) and \(\X_r\vcentcolon = \{x\in \X: [\mu(x)\pm 5C_0\delta_r]\subset I\}\) for \(r\geq 1\).
        Then, under Assumptions \ref{assum: consistency_iid}-\ref{assum: consistency_subgaussian}, there exists a constant \(C>0\) such that
        \begin{equation*}
            \lim_{n,r\to \infty}\Prob\left(\sup_{x\in \X_r}\left|\mathsf{c}_n(\fhat_r(x))-\mu(x)\right|\geq C\delta_r\right)=0,
        \end{equation*}
    \end{thm}
    This result motivates the use of this prediction, or the clipped version \eqref{eq: best_candidate} as a best candidate instead of the other options listed in Subsection \ref{subsec: best_candidates}.
    As it becomes clear from the proof of the theorem, the same conclusion holds for the out-of-sample calibrated prediction \(G(\Prob_{n+1},X_{n+1})=\iso_{n+1}\Big((\fhat_r(X_1),Y_1),\ldots,(\fhat_r(X_{n+1}),Y_{n+1})\Big)\).



    \section{Width of calibrated confidence intervals}\label{sec: interval_width}

    Theorem \ref{thm:I} provides validity guarantees for \(\C_{n+1,\alpha}\), but it does not address their statistical power.
    There are two standard ways to do this.
    
    One is through the distance (\emph{coverage gap}) between the nominal coverage \(1-\alpha\) and the true coverage \(\Prob(\theta\in \C_{\alpha})\).
    When \(\C_{\alpha}\) is a valid confidence interval, it always holds that \(\Prob(\theta\in \C_{\alpha})-(1-\alpha)\geq 0\).
    Ideally, the value of this quantity is equal to zero.
    Larger values might indicate that we have used an unnecessarily large, overconservative confidence interval.
    In the literature, this approach was used by \citet{Vander_Laan_2024} to show that, under some assumptions, the coverage gap of their prediction sets vanishes asymptotically.
    The coverage gap of conformal prediction sets has also been studied.
    More specifically, conformal prediction sets have a coverage gap of at most \(1/(n+1)\), where \(n\) is the size of the calibration set \citep{Angelopoulos_2023, Vovk_2005}.

    The second way to determine the informativeness of a confidence interval is through its width.
    Narrower confidence intervals are preferable, because they provide a better estimate of the unknown parameter.
    Moreover, in classical statistics, shortest-width intervals are  a standard alternative to equal-tailed ones.

    The width of calibrated confidence intervals derived by simple binning (Algorithm \ref{alg: simple_binning}) is equal to \((h_{n+1}-\ell_{n+1})/(\# D_{j^\star})=\vcentcolon w(\C_{n+1,\alpha})\).
    Using a heuristic argument, we can show that \(w(\C_{n+1,\alpha})\) typically converges to zero (in probability).
    Indeed, if \(I_\alpha(X_{n+1})\) is a conformal prediction set, then, under different types of assumptions \citep{Lei_2018, Yao_2025}, its width converges to \(q_{1-\alpha/2}(Y|X_{n+1})-q_{\alpha/2}(Y|X_{n+1})\) in \(L^1\), where \(q_{\alpha/2}(Y|x),q_{1-\alpha/2}(Y|x)\) are the \(\alpha/2\)- and \((1-\alpha/2)\)-quantiles of \(Y\, | \, X=x\).
    Provided that \(q_{1-\alpha/2}(Y|X_{n+1})-q_{\alpha/2}(Y|X_{n+1})\) is bounded and that \(\# D_{j^\star}\overset{\prob}{\rightarrow} \infty\) as \(r,n\to \infty\), it follows that \(w(\C_{n+1,\alpha})\overset{\prob}{\rightarrow} 0\).
    Therefore, ensuring that bin size goes to infinity yields calibrated confidence intervals with vanishing width.

    In this section, we show that isotonic recalibration (Algorithm \ref{alg: isotonic_recalibration}) has a similar property.
    Along with Proposition \ref{prop: pred_performance_isotonic} and Theorem \ref{thm: consistency}, this property complements the theoretical support for isotonic recalibration and provides further insight on the statistical power of calibrated confidence intervals.
    A relevant result has been shown by \citet{Allen_2025}, who proved that the expected thickness of their calibrated prediction bands is bounded by \(n^{-1/6}\), where \(n\) is the size of the calibration set.
    Their proof is based on a bound on the expected number of distinct levels of the isotonic regression fit \citep[Lemma~1]{Dimitriadis_2023}.
    However, this bound is only valid for isotonic regression with a binary response \(Y\).

    We make the following assumptions:
    \begin{assumption}
        \label{assum: jumps}
        Let \(J_{n+1}\) be the number of distinct values of the isotonic regression fit of \((Y_1,\ldots,Y_{n+1})\) onto \((\fhat(X_1),\ldots,\fhat(X_{n+1}))\).
        Then, there exists a constant \(q>0\) such that \(\Exp[J_{n+1}]=O(n^{1/3}\log^q n)\).
    \end{assumption}

    \begin{assumption}
        \label{assum: width_CP_set}
        The prediction set \(I_\alpha(X_{n+1})=[\ell_{n+1},h_{n+1}]\) in Algorithm \ref{alg: isotonic_recalibration} is such that \(\Exp\big[\big(Y_{n+1}-h_{n+1}\big)^2\big]=O(1)\) and \(\Exp\big[\big(Y_{n+1}-\ell_{n+1}\big)^2\big]=O(1)\)
        as \(n\to \infty\).
    \end{assumption}

    Assumption \ref{assum: jumps} is a technical assumption.
    It is supported theoretically by the results of \citet{Meyer_2000}, who show that this assumption is valid under a Gaussian-noise model.
    The observations of \citet{Groeneboom_2011} provide further evidence that it could also be valid under more general settings.
    Furthermore, \citet[Remark~5.1]{Durot_2012} provide a matching lower bound (up to \(\log\) factors) and conjecture the existence of a sharp upper bound of the same order.
    More recently, \citet{Chatterjee_2021} showed that the same property holds for isotonic quantile regression under an additive-noise model assumption.
    
    Assumption \ref{assum: width_CP_set} is supported by the results of \citet{Lei_2018} and \citet{Yao_2025} that were discussed earlier.
    Based on those results, and as long as \(\Exp[(Y_{n+1}-q_{\alpha/2}(Y\, | \, X_{n+1}))^2]=O(1)\) and \(\Exp[(Y_{n+1}-q_{1-\alpha/2}(Y\, | \, X_{n+1}))^2]=O(1)\), this assumption is realistic.
    
    We now state the main result of this subsection, whose proof can be found in Appendix \ref{sec: app_proof_width}.

    \begin{thm}
        \label{thm: width_isotonic}
        Fix \(\alpha\in (0,1)\).
        Under Assumptions \ref{assum: jumps}, \ref{assum: width_CP_set}, it holds that \(\Exp[w(\C_{n+1,\alpha})]\to 0\) as $n\to \infty$,
        where \(\C_{n+1,\alpha}\) is the calibrated prediction set obtained by Algorithm \ref{alg: isotonic_recalibration}. 
    \end{thm}

    \section{Simulations} \label{sec: simulation}

    In the previous two sections, we demonstrated the favorable properties of isotonic recalibration regarding predictive performance, consistency, and statistical power.
    We also showed that isotonic recalibration is generally superior to simple binning recalibration, even though it is hyperparameter-free and does not rely on user-specific implementation decisions.
    
    In this section, we implement isotonic recalibration on a synthetic multivariate dataset.
    The goal of this simulation study is to evaluate:
    \begin{itemize}
        \item The impact of different base models \(\fhat\);
        \item the width of the calibrated confidence interval for different sample sizes;
        \item the calibration properties of the base model \(\fhat\) and the best-candidate predictions.
    \end{itemize}

    One of the main components of Theorem \ref{thm:I} is the prediction set \(I_{\alpha}(X_{n+1})\).
    More details on how we construct this set are given in Appendix \ref{sec: app_CP}.
    Theorem \ref{thm:I} and Algorithm \ref{alg: isotonic_recalibration} do not require the prediction set \(I_\alpha(X_{n+1})\) to be independent from the training set \(\S_1\) and the calibration set \(\S_2\).
    Therefore, \(I_\alpha(X_{n+1})\) can be constructed using \(\S_1,\S_2\), and no further data-splitting is necessary.

    \subsection{Simulation design} \label{subsec: simulation_design}
    
    The data points \(\{(X_i,Y_i)\}_{i=1}^N\) consist of a three-dimensional covariate \(X_i=(X_{i1},X_{i2},X_{i3})\), where \(X_{i1},X_{i2},X_{i3}\) are independent and \(X_{i1}\sim \cN(0,1),\quad X_{i2}\sim \text{Unif}(0,1),\quad X_{i3}\sim \text{Be}(0.5)\).
    The response variable \(Y_i\) is given by
    \begin{equation}
        \label{eq: simulation_response}
        Y_i=X_{i1}^2-X_{i1}\cdot X_{i3}-\frac{4}{\frac{1}{4}+X_{i1}^2+X_{i2}^2}+\varepsilon_i,\quad i=1,\ldots,N,
    \end{equation}
    where \(\varepsilon_i\sim \text{Gamma}\left(\text{shape } = 1 + \sqrt{|X_{i1}|}, \text{ scale } = 1 + \sqrt{X_{i2}}\right)\).
    We use two different base models: a GAM, and a neural network.
    The GAM has the form \(f(x_1,x_2,x_3)=f_1(x_1)+f_2(x_2)+f_3(x_3)\).
    In particular, it contains no interaction terms, which are necessary to capture the interplay between \(X_{i1},X_{i2},X_{i3}\) in the definition of \(Y_i\).
    
    This intentional model misspecification allows for a comparison between different regimes: a GAM model that does not have the necessary capacity to address the problem complexity, and a neural network, which automatically detects interactions and is expected to fit the data better.
    Due to space constraints, we only include the results related to the GAM in the main body of the paper.
    The second part of the simulation, related to the neural network, is deferred to Appendix \ref{sec: app_simulation}.
    In general, the neural network leads to slightly narrower calibrated confidence intervals.
    Additionally, it appears to be superior in terms of calibration.
    This is in line with the empirical observations of \citet{Niculescu_2005}, who concluded that shallow neural networks have remarkably good calibration properties.

    \subsection{Interval width} \label{subsec: interval_width}

    Figure \ref{fig: GAM_general_plot} provides an illustration of the calibrated confidence intervals derived by Algorithm \ref{alg: isotonic_recalibration} for the data-generating model in \eqref{eq: simulation_response}.
    The full sample of size \(N=3000\) is split into training, calibration, and test sets, with a ratio of \(3:2:1\).
    The error rate is set to \(\alpha = 0.05\).
    
    The points \(\{(\fhat(X_i),Y_i)\}_{i=1}^{1000}\) used for isotonic recalibration are also shown in the figure.
    The red line represents the in-sample calibrated prediction of \citet{wuetr_ziegel_23}.
    At each step, a new test point \(X_{n+1}\) is considered, and the calibrated prediction set \(\C_{n+1,\alpha}(X_{n+1})\) is constructed, following Algorithm \ref{alg: isotonic_recalibration}.
    The corresponding lower and upper bounds are joined by a piecewise constant line and shown as blue curves in the graph.

    \begin{figure}[ht]
        \centering
        \includegraphics[width=0.6\linewidth]{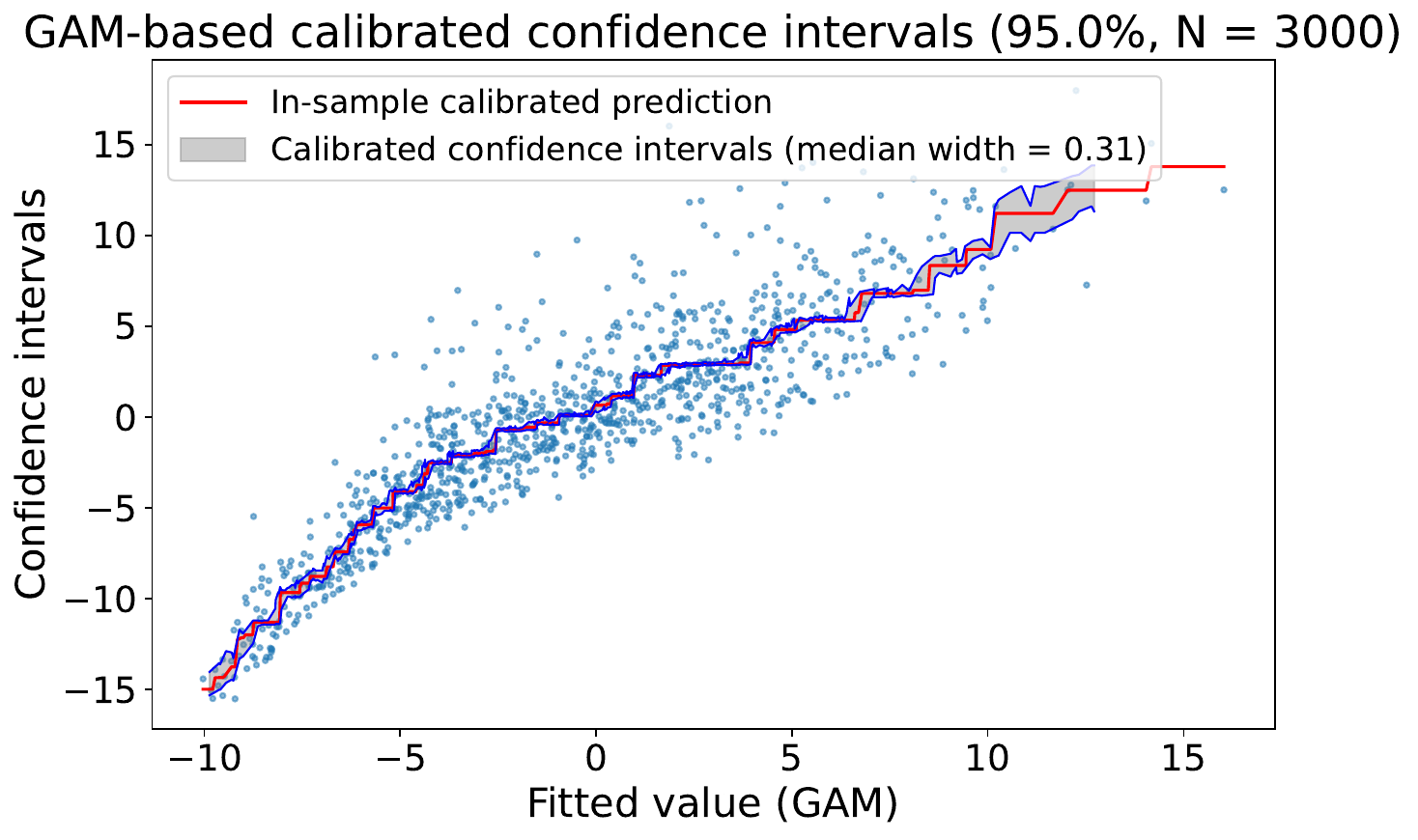}
        \caption{Calibrated confidence intervals based on a GAM model \(\widehat{f}\) and on conformalized quantile regression prediction sets for \(Y_{n+1}\); see Appendix \ref{sec: app_CP}.
        The isotonic regression fit in red corresponds to the in-sample calibrated prediction of \citet{wuetr_ziegel_23}.}
        \label{fig: GAM_general_plot}
    \end{figure}

    Figure \ref{fig: GAM_general_plot} can be compared with Figure \ref{fig: GAM100_intervals}, which corresponds to a smaller sample size.
    The effect of sample size on the width of the confidence intervals can already be seen in these two figures.
    Figure \ref{fig: interval_width} makes this effect clearer by showing the average interval width for different values of \(n\), conditionally on the training and calibration sets.

    \begin{figure}[ht]
        \centering
        \includegraphics[width=0.6\linewidth]{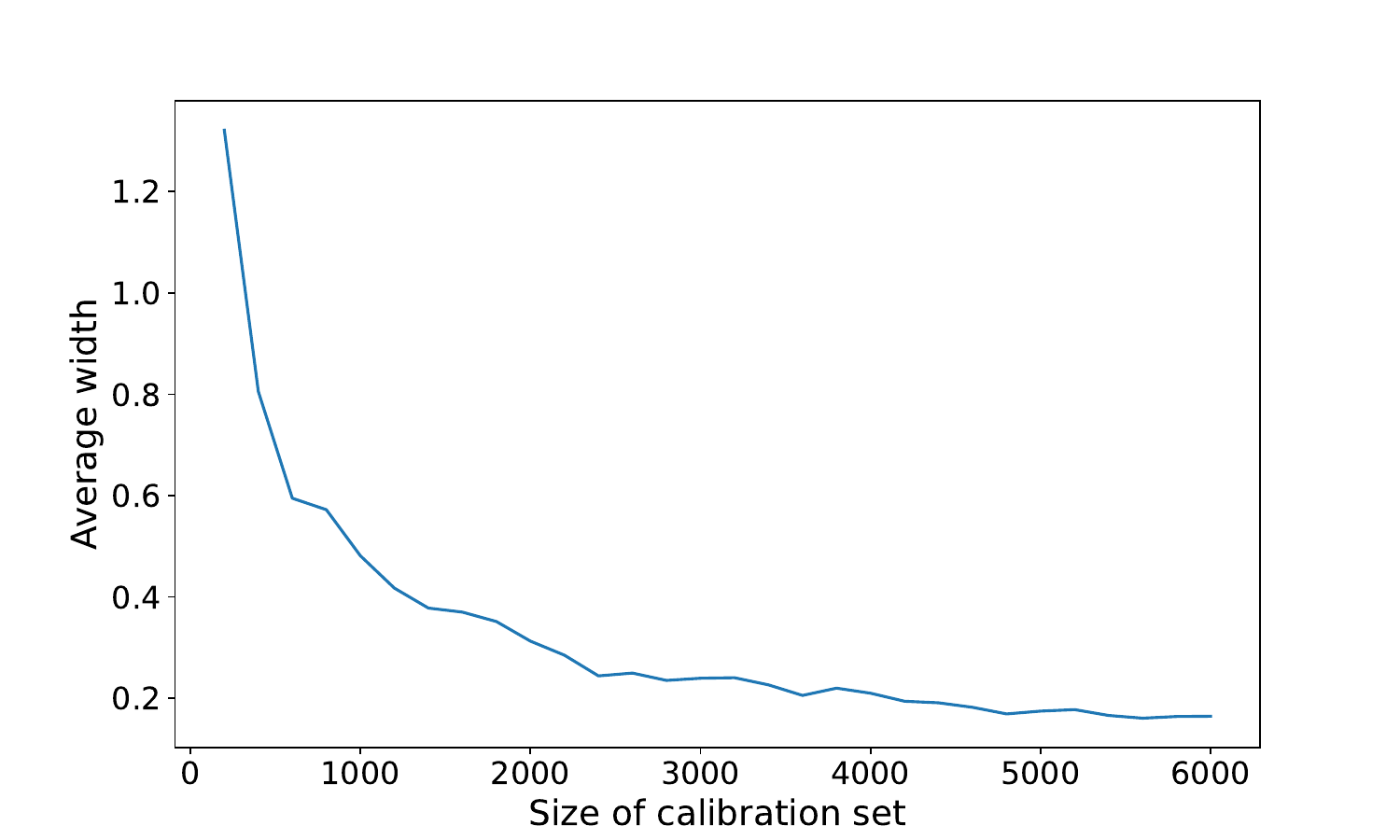}
        \caption{The average width decreases as \(n\to \infty\).
        This provides further empirical evidence for the result of Theorem \ref{thm: width_isotonic}.
        The (conditional) average is taken over the test set, with the training and calibration sets kept unchanged.}
        \label{fig: interval_width}
    \end{figure}

    \subsection{Calibration plots for best candidates} \label{subsec: calibration_plots}

    Calibration plots have been used for the graphical assessment of the calibration properties of a predictive model \citep{Dawid_1986, Zadrozny_2002}.
    Recently, they were studied further by \citet{Gneiting_2023}, who generalized them beyond binary response variables.
    The idea behind calibration plots is to fit isotonic regression to the points \(\{(\fhat(X_i),Y_i)\}_{i=1}^m\), where \(\fhat\) is a predictive model, and assess how close the isotonic regression fit is to the identity function.
    This idea is based on a heuristic argument presented in Appendix \ref{subsec: calibration_plots_motivation}.
    
    Figure \ref{fig: calibration_plots} shows the calibration plots of the GAM (top left), the true conditional mean (bottom right), used as a benchmark, and two of the best candidates discussed in Subsection \ref{subsec: best_candidates}.
    One is the clipped calibrated prediction proposed in \eqref{eq: best_candidate} (bottom left) and the other is the adjusted-midpoint prediction in \eqref{eq: best_candidate_laan}.
    
    \begin{figure}[ht]
    \centering
    \begin{minipage}{0.49\textwidth}
        \centering
        \includegraphics[width=\linewidth]{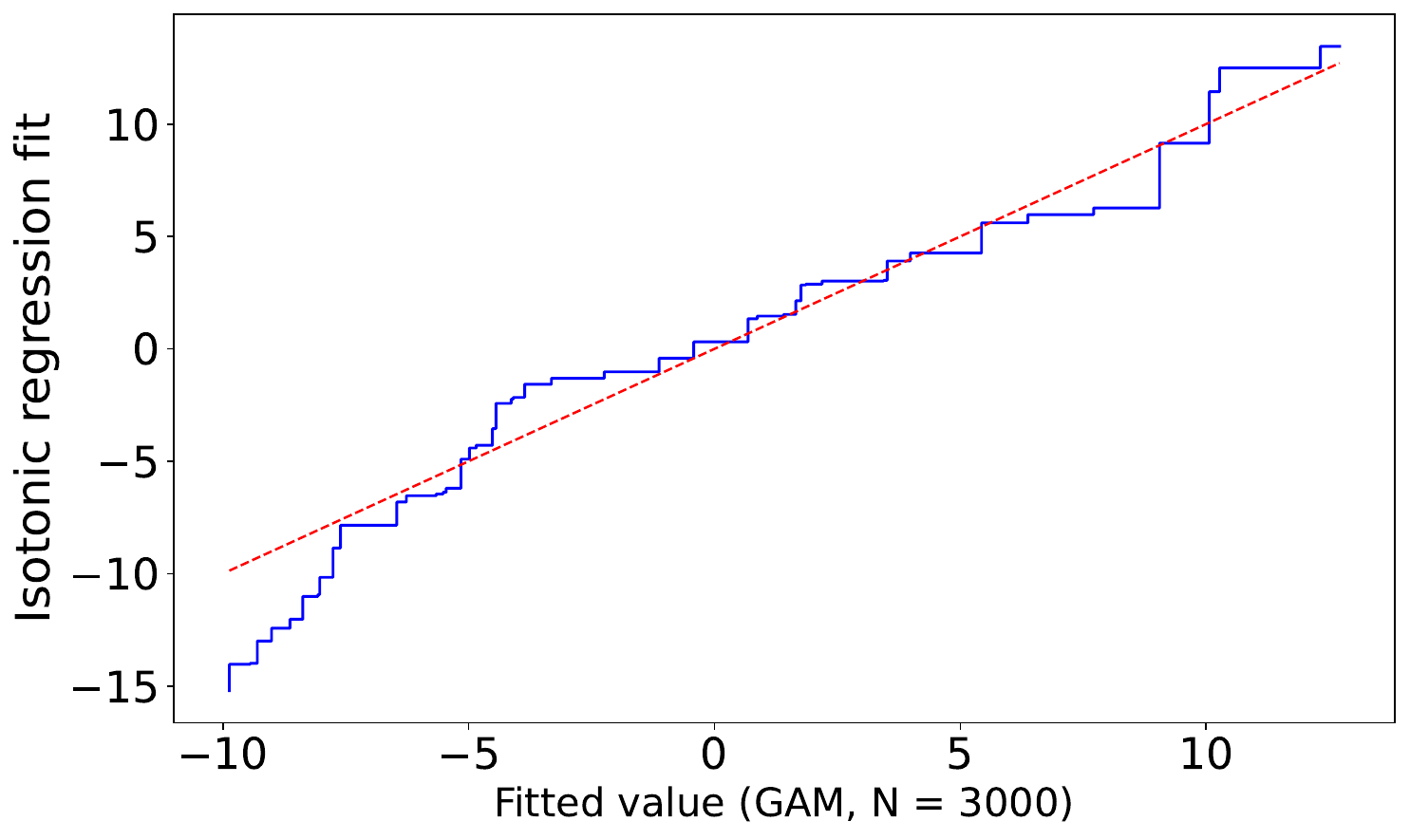}
    \end{minipage}
    \hfill
    \begin{minipage}{0.49\textwidth}
        \centering
        \includegraphics[width=\linewidth]{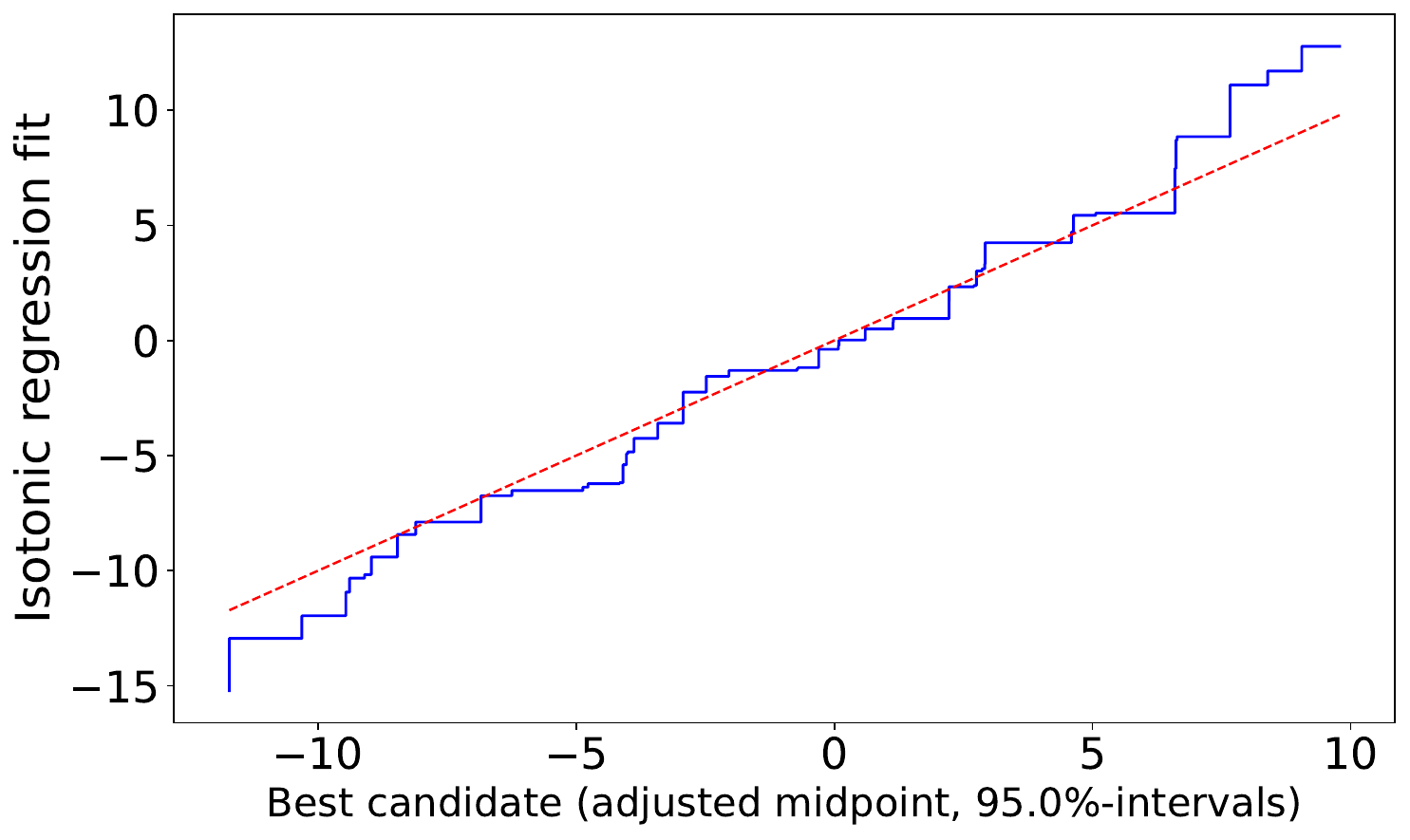}
    \end{minipage}

    \vspace{0.5cm}

    \begin{minipage}{0.49\textwidth}
        \centering
        \includegraphics[width=\linewidth]{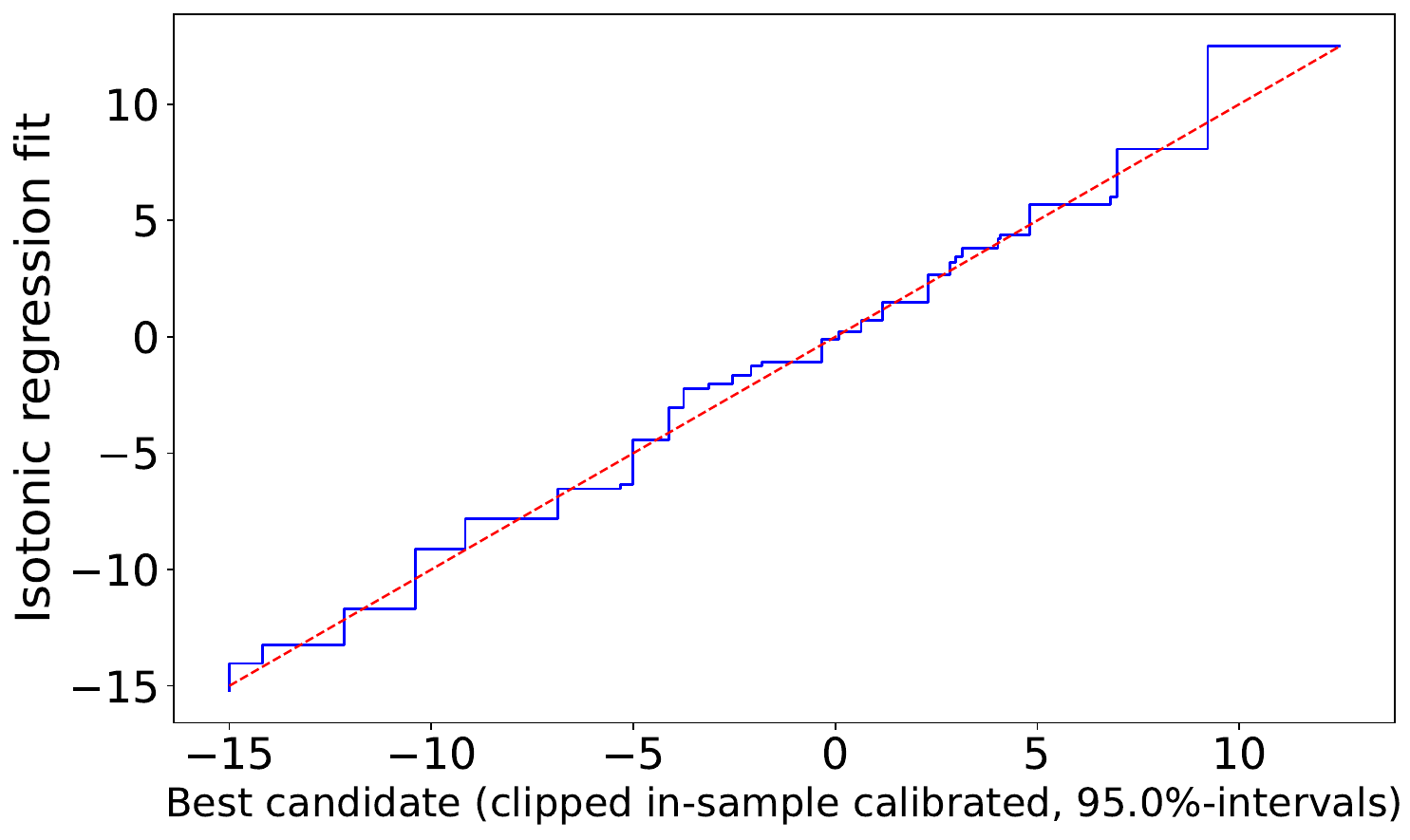}
    \end{minipage}
    \hfill
    \begin{minipage}{0.50\textwidth}
        \centering
        \includegraphics[width=\linewidth]{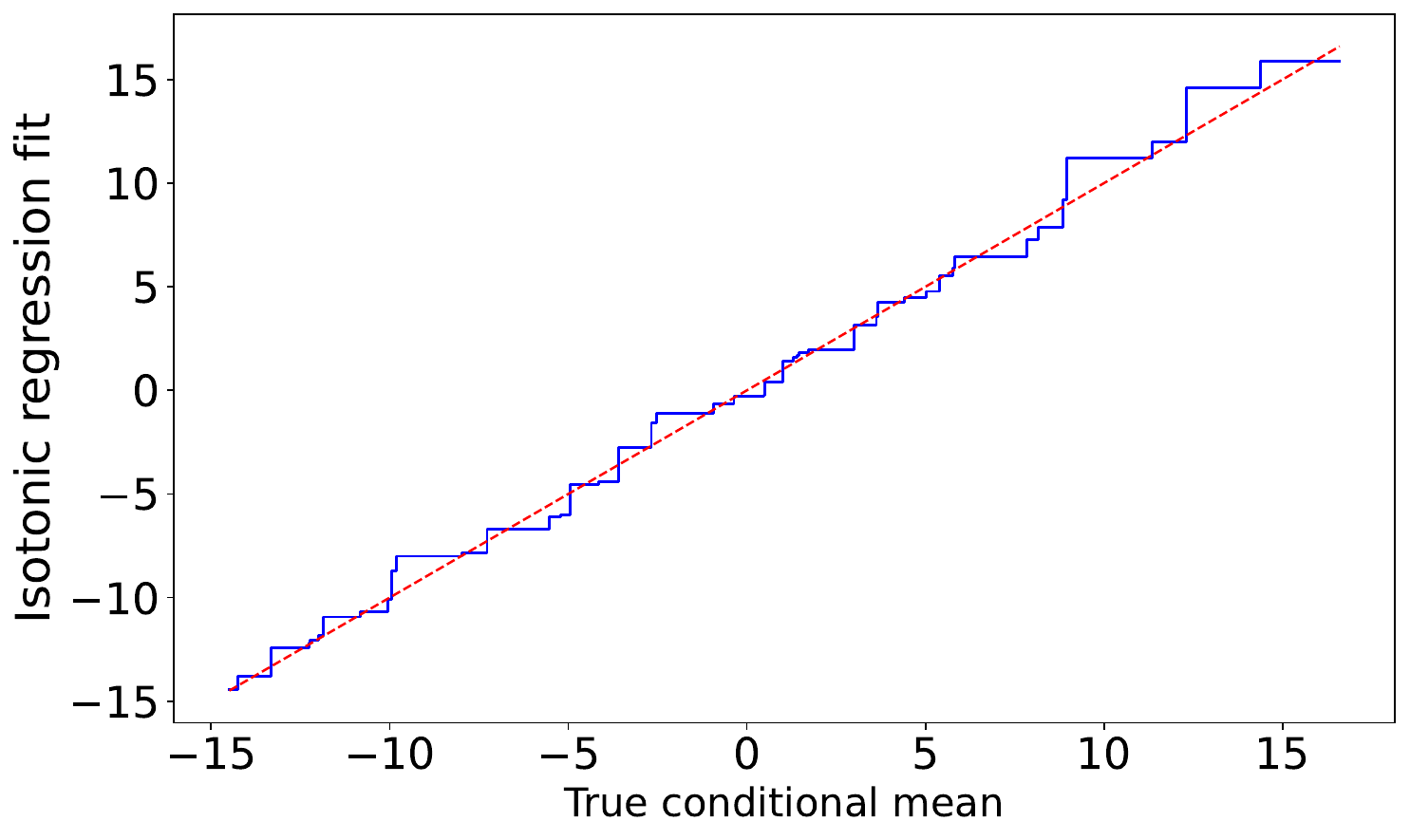}
    \end{minipage}
    \caption{Calibration plots for four different predictions.
    The (clipped) in-sample calibrated prediction seems to deal effectively with the miscalibration of the GAM.}
    \label{fig: calibration_plots}
    \end{figure}

    The main conclusion from Figure \ref{fig: calibration_plots} is that the GAM is adequately calibrated, except for the left tail of the distribution of \(\fhat(X)\).
    The miscalibration in that area is corrected by the two best candidates, particularly by the in-sample calibrated one.
    The adjusted midpoint also performs relatively well, but its performance diminishes at the right end of the graph.
    This could be because of the wider confidence intervals \(\C_{n+1,\alpha}\) observed in that area in Figure \ref{fig: GAM_general_plot}.

    This observation highlights one of the advantages of the in-sample calibrated best candidate, namely, the fact that it does not depend on the error rate \(\alpha\).
    This property makes it immune to uninformative confidence intervals.
    As for the clipped version \eqref{eq: best_candidate}, even though it does depend on \(\alpha\), this dependence only affects cases where the in-sample calibrated prediction takes extreme values that should possibly be corrected.

    \section{Case study: Third-party motor liability claims} \label{sec: case_study}

    In this section, we apply our methods to the dataset \texttt{freMTPL2freq} from the \texttt{R} package \texttt{CASdatasets}.
    This is a widely used benchmarking dataset in actuarial science.
    It is one of the main datasets used in \citet{Wuthrich_2023} and one of the datasets suggested by the Swiss Association of Actuaries for \emph{applying and learning data science methods}\footnote{https://www.actuarialdatascience.org/ADS-Tutorials/}.
    It has previously also been analyzed by \citet{Noll_2020} and \citet{Floser_2025}, and it was employed by \citet{wuetr_ziegel_23} to demonstrate in-sample recalibration.
    
    Given the strong interest in actuarial science in calibrated predictions \citep{wuetr_ziegel_23, Kruger_2021, Denuit_2021}, \texttt{freMTPL2freq} is suitable for an empirical study of the (out-of-sample) calibrated confidence intervals developed in this paper.
    It consists of \(n=678\, 013\) rows, which represent different insurance policies.
    The policies were observed mostly over one year, over an overall period between 2011 and 2013.
    Each row reports the number of claims of the corresponding policy within that year.
    It also contains covariate information, such as the driver's age, the vehicle type, and the duration of the insurance contract.
    A detailed list of all covariates is given in Table \ref{tab: freMTPL2freq} in Appendix \ref{sec: app_case_study}.

    Given the covariate information and the \emph{exposure} of a policy (the fraction of the year for which the policy was active), the goal is to develop models that predict the claim numbers \(N_i=\texttt{ClaimNb}_i\geq 0\), \(i=1,\ldots,678\,013\).
    Since the response takes only non-negative integer values, a Poisson model is a standard approach.
    The Poisson model is expressed via a regression function \(\lambda: \X\to \R_+\), where \(\X\) is the covariate space, such that the claim numbers are independent and satisfy \(N_i\sim \text{Poi}(\lambda(x_i)v_i)\),
    where \(v_i\) is the exposure of policy \(i\in \{1,\ldots,678\, 013\}\).
    In Appendix \ref{sec: app_case_study} we comment on the use of \(v\) as a scaling factor and not as a covariate.

    Many of the covariates are categorical, so some feature pre-processing is necessary before fitting the models.
    We use the same pre-processing steps as \citet{Wuthrich_2023}.

    We split the dataset randomly into a training set \(\S_1\) (60\%), a calibration set \(\S_2\) (20\%), and a test set \(\T\) (20\%).
    The construction of calibrated confidence intervals first requires fitting a regression model on the training set \(\S_1\).
    The target variable is the regression function \(\lambda:\X\to \R_+\) which determines the expectation of the claim number \(N\).
    Motivated by \citet{Wuthrich_2023,Noll_2020}, we compare a Poisson GLM, a boosted regression tree (CatBoost) and a neural network.
    As in \citet{wuetr_ziegel_23}, we also include the intercept-only model (null model) as a benchmark.
    More details about the model parameters can be found in Appendix \ref{sec: app_case_study}.
    Table \ref{tab: pred_performance} shows the training and test error of the four models, measured by means of the deviance loss.
    \begin{table}[h]
        \centering
        \caption{Deviance loss of the four different models.}
        \begin{tabular}{|c|c|c|}
            \hline
            Model &  Training error (\(\S_1\)) & Test error (\(\T\)) \\ \hline
            Null model & 0.477310 & 0.486497 \\ \hline
            Poisson GLM & 0.455478 & 0.466227 \\ \hline
            CatBoost &  0.440605 &  0.457245 \\ \hline
            Neural network & 0.448992 & 0.461538 \\ \hline
        \end{tabular}
        \label{tab: pred_performance}
    \end{table}
    
    The implementation of Algorithm \ref{alg: isotonic_recalibration} also requires a prediction set \(I_\alpha(x_{n+1})\) that contains \(N_{n+1}\) with high probability.
    We obtain such a set using conformal prediction with the Poisson deviance as a nonconformity score. More details can be found in Section \ref{sec: app_case_study}.
    Since the response variables \(N_i\) are bounded between \(0\) and \(4\), setting \(\alpha=0\) and \(I_0(x_{n+1})=\{0,\ldots,4\}\) yields calibrated confidence intervals with almost sure coverage.

    \subsection{Conformal prediction sets and calibrated confidence intervals} \label{subsec: CP_CCI}

    We compare three error-rate levels \(\alpha\), namely \(0, 0.01\), and \(0.05\).
    It turns out that the conformal prediction sets \(I_{0.05}(x_{n+1})\) are all empty or singletons.
    This seems surprising, but it follows naturally from the unbalanced design.
    With more than \(95\%\) of the policies having no reported claims, quantifying uncertainty at level \(\alpha = 0.05\) is not informative.
    For comparison, this is similar to quantifying uncertainty at level \(\alpha = 0.5\) in a balanced dataset with two classes.
    The unbalanced design is also reflected in the fitted values, most of which lie in the interval \([0, 0.4]\).
    



    Different prediction sets give rise to different calibrated confidence intervals \(\mathcal{C}_{n+1,\alpha}\).
    The ones corresponding to the Poisson GLM are illustrated in Figure \ref{fig: GLM_CatBoost_NN}.
    Like before, the \(x\)-axis corresponds to the fitted value \(\widehat{\lambda}(x_j)v_j\) for a test point \((x_j,v_j)\).
    As expected, the intervals corresponding to \(\alpha = 0.05\) are either empty or singletons, just like the associated prediction sets.
    Between \(\alpha = 0.01\) and \(\alpha = 0\) there are only minor differences, which are more pronounced towards the right part of the graph, where there are only a few data points.

    As a result of the scarcity of data in those areas, the width of the intervals increases drastically.
    Indeed, in blocks with few data points, the fitted value of isotonic regression is affected significantly by varying one of them, as Algorithm \ref{alg: isotonic_recalibration} does.

    Figure \ref{fig: GLM_CatBoost_NN} also shows the intervals associated with the CatBoost model and the neural network.
    The \(x\)-axis reports the fitted values of those models, so the scales are slightly different.
    The same is true for the \(y\)-axis.
    Therefore, the three panels are not visually comparable.

    \begin{figure}[ht]
        \centering
        \begin{minipage}{\textwidth}
            \centering
            \includegraphics[width=0.5\linewidth]{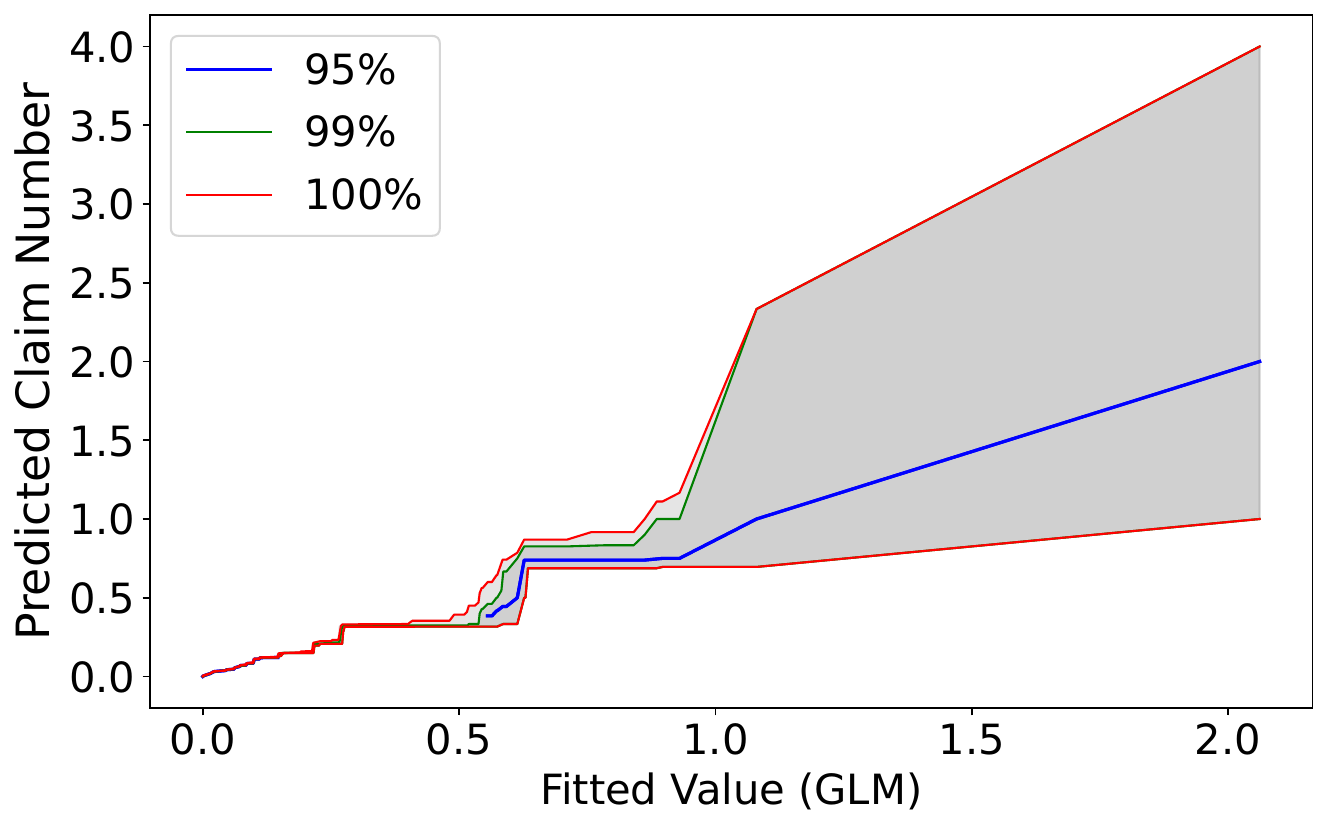}
        \end{minipage}
        \begin{minipage}{0.49\textwidth}
            \centering
            \includegraphics[width=\linewidth]{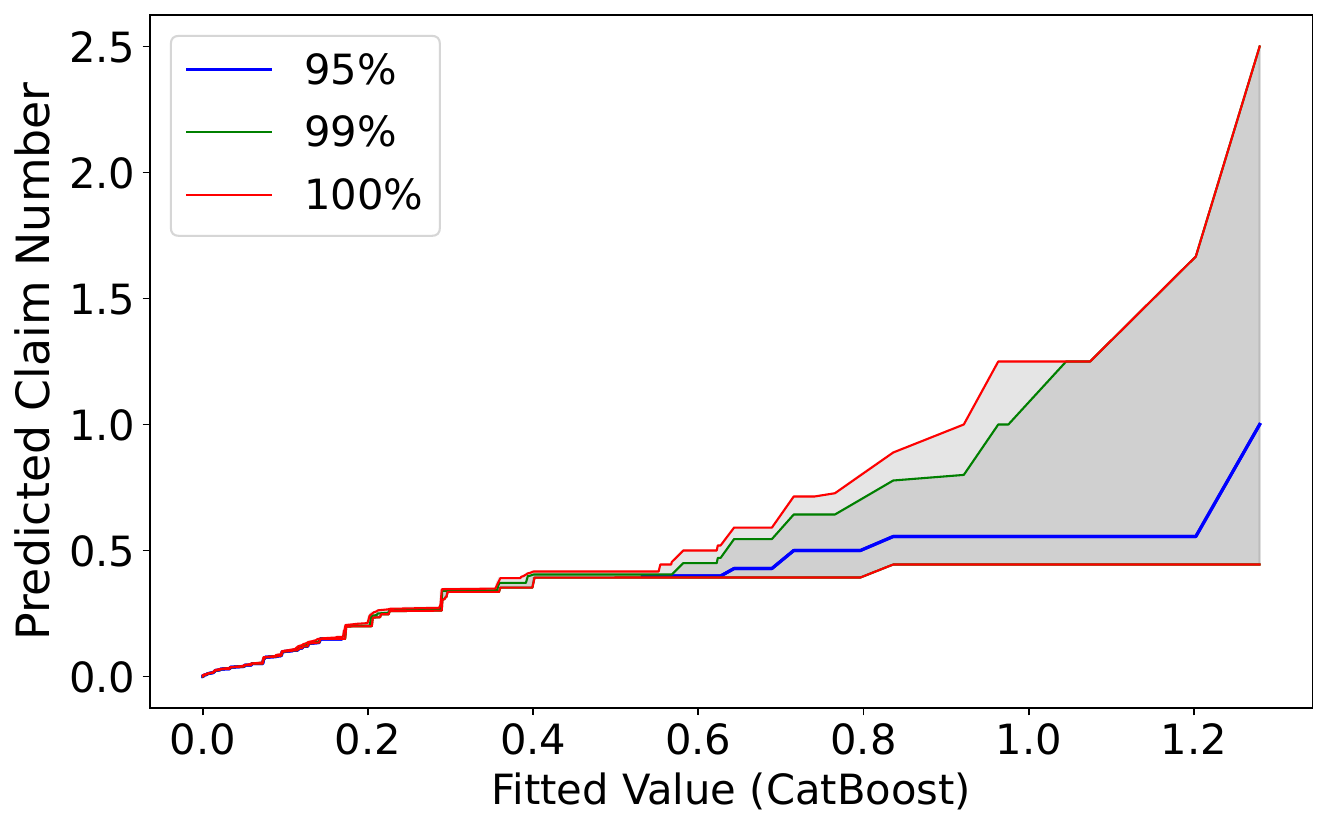}
        \end{minipage}
        \hfill
        \begin{minipage}{0.49\textwidth}
            \centering
            \includegraphics[width=\linewidth]{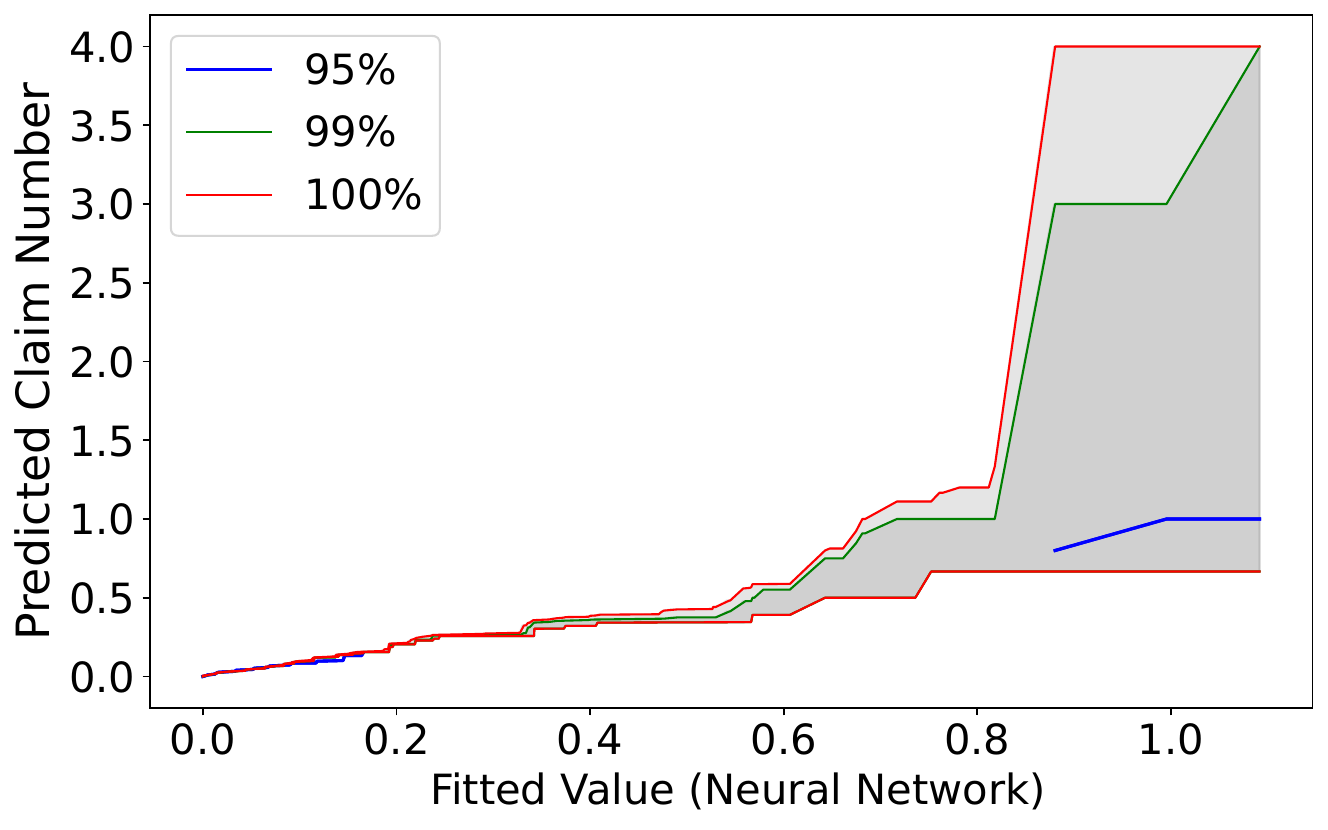}
        \end{minipage}
        \caption{Calibrated confidence intervals for different values of \(\alpha\).}
        \label{fig: GLM_CatBoost_NN}
    \end{figure}
    
    The combination of different axis scales and intervals of tiny width makes it hard to detect which model outputs the most informative confidence intervals.
    In fact, as we demonstrate in Appendix \ref{sec: app_case_study}, none of them is uniformly superior.


    A closer look at Figure \ref{fig: GLM_CatBoost_NN} reveals that, even in areas with a large or moderate number of data points, such as the area where \(\texttt{fitted value}\approx 0.3\), there is often a significant relative increase in width when we move from \(\alpha = 0.01\) to \(\alpha = 0\).
    For the CatBoost model, this is verified by Figure \ref{fig: CatBoost_alphas}.
    Therefore, even when the response variable is bounded, moving from a prediction set \(I_0(x_{n+1})\), equal to the entire domain, to a conformal prediction set \(I_{\alpha}(X_{n+1})\) can lead to a large increase in power.
    This aligns with our discussion in Subsection \ref{subsec: motivating_prediction_set} for the necessity of the conformal prediction set.

    \begin{figure}[ht]
        \centering
        \includegraphics[width=0.55\linewidth]{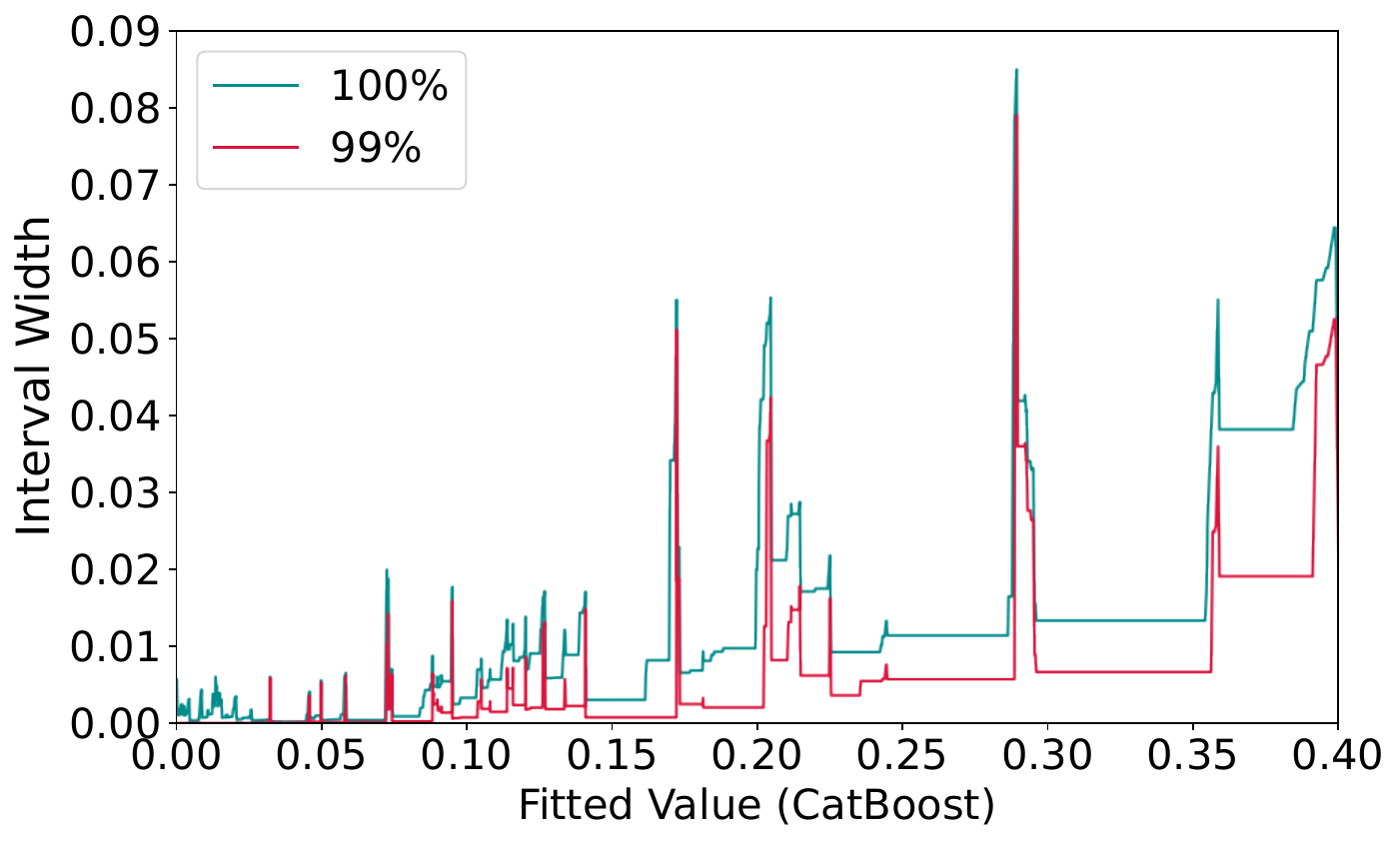}
        \caption{Width of \(\C_{n+1,\alpha}\) for different values of \(\alpha\).
        Even in areas with a large number of data points, the width drops significantly when we change the error rate from \(\alpha = 0\) to \(\alpha = 0.01\).}
        \label{fig: CatBoost_alphas}
    \end{figure}

    \subsection{Model comparison: Calibration} \label{subsec: comparison_calibration}

    We conclude the case study with an evaluation of the calibration of the three selected prediction models.
    In-sample calibration translates to \(\Exp_{(x,N,v)\sim \Prob_n}[N|\widehat{\lambda}(x)v]=\widehat{\lambda}(x)v\), where \(\Prob_n\) is the empirical distribution of the calibration set.
    The law of iterated expectation yields \(\Exp_{(x,N,v)\sim \Prob_n}[N]=\Exp_{(x,N,v)\sim \Prob_n}[\widehat{\lambda}(x)v]\), which is a weaker property, known as \emph{marginal calibration} or \emph{unbiasedness in the large} \citep{Murphy_1967}.
    If \(\{(x_i,N_i,v_i)\}_{i=1}^n\) denotes the calibration set, then marginal calibration is equivalent to
    \begin{equation}
        \label{eq: marginal_calibration}
        \frac{1}{n}\sum_{i=1}^n N_i = \frac{1}{n}\sum_{i=1}^n \widehat{\lambda}(x_i)v_i \Leftrightarrow \frac{\sum_{i=1}^n N_i}{\sum_{i=1}^n v_i}=\frac{\sum_{i=1}^n \widehat{\lambda}(x_i)v_i}{\sum_{i=1}^n v_i},
    \end{equation}
    which essentially requires that the average observed frequency is equal to the average predicted frequency.
    Marginal calibration is considered the minimal calibration property that a predictive model should satisfy.
    It is straightforward to test it, so prior works have used it to assess the calibration properties of prediction models \citep{wuetr_ziegel_23, Noll_2020}.
    Table \ref{tab: insample_frequency} shows that none of the three models is marginally calibrated.
    However, on the training set, the Poisson GLM and the null model fit the true frequency perfectly, see Appendix \ref{sec: app_case_study}.
    \begin{table}[h]
    \caption{True and predicted claim frequency (offset by exposure).}
        \label{tab: insample_frequency}
        \centering
        \begin{tabular}{|l|c|c|c|}
            \hline
            \textbf{Model}  & \textbf{Frequency (\(\S_1\))} & \textbf{Frequency} (\(\S_2\)) & \textbf{Frequency (\(\T\))}\\ \hline
            Null model & 7.3535\% & 7.3535\% & 7.3535\% \\ \hline
            Poisson GLM & 7.3535\% & 7.3477\% & 7.3843\% \\ \hline
            CatBoost & 7.3547\% & 7.3492\% & 7.3713\% \\ \hline
            Neural network & 7.3556\% & 7.3394\% & 7.3730\% \\ \hline\hline
            True frequency & 7.3535\% & 7.2792\% & 7.5018\% \\ \hline
        \end{tabular}
        
    \end{table}

    The best candidates proposed in Subsection \ref{subsec: best_candidates} face similar miscalibration issues, as shown in Table \ref{tab: calibration_candidates}.
    This table corresponds to the Poisson GLM, but the results are similar for the other two models.
    These results support the comment we made in Subsection \ref{subsec: best_candidates}, namely that these candidates rely on heuristic arguments and lack rigorous guarantees.
    \begin{table}[h]
        \centering
        \caption{The candidate predictions proposed in Subsection \ref{subsec: best_candidates} are miscalibrated on \(\T\).}
        \label{tab: calibration_candidates}
        \begin{tabular}{|l|c|}
            \hline
            \textbf{Candidate} & \textbf{Predicted frequency (\(\T\))}\\ \hline
            In-sample calibrated prediction & 7.3128\% \\ \hline \hline
            Lower endpoint \hfill \(\alpha = 0.01\) & 7.3045\%\\ \hline
            Midpoint \hfill \(\alpha = 0.01\) & 7.3346\%\\ \hline
            Upper endpoint \hfill \(\alpha = 0.01\) & 7.3648\% \\ \hline \hline 
            Lower endpoint \hfill \(\alpha = 0\) & 7.3045\% \\ \hline
            Midpoint \hfill \(\alpha = 0\) & 7.4445\% \\ \hline
            Upper endpoint \hfill \(\alpha = 0\) & 7.5845\% \\ \hline \hline
            True frequency  & 7.5018\% \\ \hline
        \end{tabular}
    \end{table}
    For \(\alpha=0\), the observed true frequency lies between the predicted frequencies obtained by the lower and upper endpoints of the associated calibrated confidence interval, whereas for \(\alpha \neq 0\), this is not the case.
    This is not surprising, since the average predicted frequencies are computed conditionally on the training and calibration sets, whereas Theorem \ref{thm:I} yields only marginally valid guarantees.

    \section{Discussion}

    We provide a new methodology for constructing non-parametric, distribution-free confidence intervals for calibrated predictions.
    The main ingredient of this methodology is a coupling of isotonic calibration with conformal prediction.
    We have provided a theoretical framework for these confidence intervals that encompasses both validity and statistical power.

    Directions for further theoretical work include establishing optimal (minimax) convergence rates for the asymptotic expected width of calibrated confidence intervals, and generalizing the assumptions of Theorem \ref{thm: width_isotonic}. 

    
    The flexibility of the main theoretical backbone of the methodology (Theorem \ref{thm:I}) regarding the choice of the prediction set \(I_{\alpha}(X_{n+1})\) provides numerous possibilities for tailoring calibrated confidence intervals to specific applications at hand.
    In addition to standard conformal prediction sets that offer marginal coverage, other options such as conformal prediction sets with sample-conditional coverage \citep{Duchi_2025} or parametric prediction sets could provide interesting choices.

    \section{Acknowledgements}

    The authors would like to thank Kai Schärer for valuable discussions and feedback.

    \bibliographystyle{abbrvnat}
    \bibliography{refs}
    \newpage

    \appendix

    \section{Proof of Theorem \ref{thm:I}}
    \label{sec: appA}

    \begin{proof}[Proof of Theorem \ref{thm:I}]
        Under Assumption \ref{assum: A1}, it holds that \(\Exp|G(\Prob_{n+1},X_{n+1})|<\infty\).
        Indeed, it holds that
        \begin{equation*}
            G(\Prob_{n+1},X_{n+1})=\frac{\sum_{i\in D} Y_i}{\# D},
        \end{equation*}
        for some (random) set \(D\subseteq \{1,\ldots,n+1\}\).
        Since \(Y_1,\ldots,Y_{n+1}\) have a finite first moment, it follows directly that \(G(\Prob_{n+1},X_{n+1})\) also has a finite first moment.
        
        For all indices \(i\in \{1,\ldots,n+1\}\), we denote \(G(\Prob_{n+1},X_i)\) by \(G_i\).
        First, we will show that \(G(\Prob_{n+1},X_{n+1})\) is calibrated, i.e. 
        \begin{equation}
            \label{eq: true_calibrated_prediction}
            \Exp[Y_{n+1}\, | \, G(\Prob_{n+1},X_{n+1})] = G(\Prob_{n+1},X_{n+1}), \quad \text{almost surely}.
        \end{equation}
        To show this, we use the definition of the conditional expectation, according to which it suffices to show that
        \begin{equation}
            \label{eq: conditional_exp_definition}
            \mathbb{E}\Big[Y_{n+1}\one\{G_{n+1}\in B\}\Big]=\mathbb{E}\Big[G_{n+1}\one\{G_{n+1}\in B\}\Big]
        \end{equation}
        for any Borel set \(B\in \BR\).
        The tower property yields
        \begin{equation}
            \label{eq: tower_property}
            \Exp\Big[(Y_{n+1}-G_{n+1})\one\{G_{n+1}\in B\}\Big]
            =  \Exp\Big[\Exp\Big[(Y_{n+1}-G_{n+1})\one\{G_{n+1}\in B\} \Big | \  \Prob_{n+1}\Big]\Big].
        \end{equation}
        We now make use of exchangeability.
        Denote \((\X\times \R)^{r+1}\) by \(\S^{r+1}\).
        The expression \((Y_{n+1}-G_{n+1})\one\{G_{n+1} \in B\}\) is a deterministic function of \(\S_1, \Prob_{n+1}\) and \((X_{n+1},Y_{n+1})\), so we can denote it by \(H(\Prob_{n+1},(X_{n+1},Y_{n+1}),\S_1)\).
        From the disintegration theorem \citep[Theorem~8.5]{Kallenberg_2021}, it follows that
        \begin{equation*}
            \Exp\Big[H(\Prob_{n+1},(X_{n+1},Y_{n+1}),\S_1) \Big | \  \Prob_{n+1}\Big]
             = \int_{\S^{r+1}\times (\X\times \R)} H(\Prob_{n+1},z,s)\, \mu(\Prob_{n+1},d(z,s)),
        \end{equation*}
        where \(\mu(\Prob_{n+1},\cdot)\) is the regular conditional distribution of \(((X_{n+1},Y_{n+1}), \S_1)\) given \(\Prob_{n+1}\).
        Since \(\S_1\) is independent from \(\Prob_{n+1}\) and \((X_{n+1},Y_{n+1})\), its conditional distribution given \(\Prob_{n+1}\) is simply its marginal distribution, which we denote by \(\pi\).
        Therefore, \(\mu\) decomposes as
        \begin{equation*}
            \mu(\Prob_{n+1},d(z,s))=\pi(ds)\mu^\prime(\Prob_{n+1},dz),
        \end{equation*}
        where \(\mu^\prime(\Prob_{n+1},\cdot)\) is now the regular conditional distribution of \(Z_{n+1}\vcentcolon = (X_{n+1},Y_{n+1})\) given \(\Prob_{n+1}\).
        Due to exchangeability, we obtain from \eqref{eq: exchangeability_implication} that \(\mu^\prime(\Prob_{n+1},\cdot)=\Prob_{n+1}\), so
        \begingroup
        \allowdisplaybreaks
        \begin{align*}
            \Exp\Big[H(\Prob_{n+1},(X_{n+1},Y_{n+1}),\S_1)\Big |\  \Prob_{n+1}\Big]
            & = \int_{\S^{r+1}\times(\X\times \R)} H(\Prob_{n+1},z,s)\, \pi(ds)\,  \Prob_{n+1}(dz)\\
            & = \int_{\S^{r+1}}\left(\int_{\X\times \R}H(\Prob_{n+1},z,s)\, \Prob_{n+1}(dz)\right)\, \pi(ds)\\
            & = \int_{\S^{r+1}} \Exp_{(X,Y)\sim \Prob_{n+1}}\Big[H(\Prob_{n+1},(X,Y),s)\Big]\, \pi(ds).
        \end{align*}
        \endgroup
        where the use of Fubini's theorem in the second equality is justified by the fact that \(H\) is integrable by definition.
        For each fixed training set \(s\in \S^{r+1}\), the integrand is equal to
        \begingroup
        \allowdisplaybreaks
        \begin{align*}
            \Exp_{(X,Y)\sim \Prob_{n+1}}
            &\Big[(Y-G(\Prob_{n+1},X))\one\{G(\Prob_{n+1},X) \in B\}\Big]\\
            & = \Exp_{(X,Y)\sim \Prob_{n+1}}\Big[Y\one\{G(\Prob_{n+1},X) \in B\}\Big]-\Exp_{(X,Y)\sim \Prob_{n+1}}\Big[G(\Prob_{n+1},X)\one\{G(\Prob_{n+1},X) \in B\}\Big]\\
            & = 0,
        \end{align*}
        \endgroup
        where the dependence on \(s\) is absorbed by \(G\) and is therefore omitted.
        The last equality follows from the fact that \(G\) is in-sample calibrated, and from the definition of the conditional expectation.
        It now follows that \(\Exp[H(\Prob_{n+1},(X_{n+1},Y_{n+1}),\S_1)|\Prob_{n+1}]=0\), so \eqref{eq: conditional_exp_definition}, \eqref{eq: tower_property} yield that \(G(\Prob_{n+1},X_{n+1})\) is calibrated.
        
        To conclude the proof, we only need to show that $G(\Prob_{n+1},X_{n+1})$ belongs to \(\C_{n+1,\alpha}\) with probability at least \(1-\alpha\).
        This follows from the definition of \(\C_{n+1,\alpha}\) and the assumption on $I_\alpha(X_{n+1})$: 
        \begin{align*}
            \Prob\big(G(\Prob_{n+1},X_{n+1})\in \C_{n+1,\alpha}\, | \, \A \big)
            &=\Prob\left(\exists y\in I_\alpha\left(X_{n+1}\right)\vcentcolon G(\Prob_{n+1},X_{n+1})=G(\Prob_{y},X_{n+1})\, | \, \A \right)\nonumber\\
            &\geq \Prob\left(\exists y\in I_\alpha\left(X_{n+1}\right): Y_{n+1}=y\, | \, \A \right)\\
            &\geq 1-\alpha.
        \end{align*}
    \end{proof}

    \begin{rem}
        \label{rem: CP_Gibbs_2025}
        The prediction set \(I_\alpha(X_{n+1})\) shows up only in the last part of the proof.
        This last part illustrates how the coverage guarantees of \(I_\alpha(X_{n+1})\) carry over to \(\C_{n+1,\alpha}\).
        Those coverage guarantees are expressed via the \(\sigma\)-algebra \(\A\).
        Some examples include:
        \begin{itemize}
            \item Standard split-conformal prediction sets.
            In that case, \(\A\) is the trivial \(\sigma\)-algebra and the coverage offered by \(\C_{n+1,\alpha}\) is only marginally valid.
            \item Conformalized quantile regression \citep{Romano_2019}.
            In that case, \(\A\) is the \(\sigma\)-algebra generated by the training set \(\S_1\).
            More details about this type of conformal prediction set are given in Appendix \ref{sec: app_CP}.
            \item The conformal prediction sets developed by \citet{Gibbs_2025}.
            If \(\{(X_i,Y_i)\}_{i=1}^{n+1}\overset{\text{iid}}{\sim} \prob_{X,Y}\), then these sets satisfy a coverage guarantee of the form
            \begin{equation*}
                1-\alpha
                \leq \Prob\Big(Y_{n+1}\in I_\alpha(X_{n+1})\, \Big|\,  X_{n+1}\in G\Big)
                \leq 1-\alpha + \frac{|\mathcal{G}|}{(n+1)\Prob(X_{n+1}\in G)},\quad \forall G\in \mathcal{G},
            \end{equation*}
            where \(\mathcal{G}\) is a finite collection of groups.
            These sets are more general than Mondrian conformal prediction sets, as they allow for overlapping groups.
            As in the above proof, it can be shown that this property carries over to \(\C_{n+1,\alpha}\).
            Furthermore, \citet{Gibbs_2025} develop randomized conformal prediction sets such that
            \begin{equation}
                \label{eq: stronger_calibration}
                \Exp\Big[f(X_{n+1})\Big(\one\{Y_{n+1}\in I_\alpha(X_{n+1})\}-(1-\alpha)\Big)\Big]=0\quad \text{for all } f\in \mathcal{B},
            \end{equation}
             where \(\mathcal{B} = \left\{\Phi^\top(\cdot)\beta:\, \beta\in \R^d\right\}\) is any finite-dimensional function class generated by a map \(\Phi = (\phi_1,\ldots,\phi_d):\X\to \R^d\).
             
             For recalibration based on isotonic regression, it holds that \(\one\{G(\Prob_{n+1},X_{n+1})\in \C_{n+1,\alpha}\} = \one\{Y_{n+1}\in I_\alpha(X_{n+1})\}\).
             Indeed, by definition it follows that \(\one\{G(\Prob_{n+1},X_{n+1})\in \mathcal{C}_{n+1,\alpha}\}\geq \one\{Y_{n+1}\in I_\alpha(X_{n+1})\}\).
             For the converse, note that if \(I_\alpha(X_{n+1})\) is an interval and if \(Y_{n+1}\notin I_\alpha(X_{n+1})\) then either \(Y_{n+1}<\min I_\alpha(X_{n+1})\) or \(Y_{n+1}>\max I_\alpha(X_{n+1})\).
             Since the isotonic regression fit is pointwise strictly increasing in \(Y_{n+1}\), this would yield that \(G(\mathbb{P}_{n+1},X_{n+1})\notin \mathcal{C}_{n+1,\alpha}\).
             Due to the above argument, \eqref{eq: stronger_calibration} also carries over to \(\C_{n+1,\alpha}\), that is,
        \begin{align*}
            \Exp\Big[f(X_{n+1})\Big(\one\{G(\Prob_{n+1},X_{n+1}) \in \C_{n+1,\alpha}\}-(1-\alpha)\Big)\Big] = 0 \quad \text{for all } f\in \mathcal{B}.
        \end{align*}
        \end{itemize}
        
    \end{rem}

    \section{Over/under-calibration} \label{sec: over_under_calibration}
    The following proposition shows that, with high probability, the upper and lower endpoints \(T_h(\fhat(X_{n+1})), T_{\ell}(\fhat(X_{n+1}))\) of \(\C_{n+1,\alpha}\) are over- and under-calibrated respectively.
    This discourages using best-candidate predictions that lie outside of \(\C_{n+1,\alpha}\) and motivates our proposal of \eqref{eq: best_candidate} as a best candidate. 

    \begin{prop}
        \label{prop: over_under_calibration}
        Suppose that the prediction set \(I_\alpha(X_{n+1})=[\ell_{n+1},h_{n+1}]\) of Theorem \ref{thm:I} satisfies \eqref{eq:Ialpha} for some \(\sigma\)-algebra \(\A\) on the underlying probability space \(\Omega\). 
        Then, for the functions \(T_\ell\) and \(T_h\) of Algorithm \ref{alg: isotonic_recalibration}, it holds with probability at least \(1-\alpha\), conditionally on \(\A\) that
        \begin{equation*}
            \Exp_{(X,Y)\sim \Prob_{n+1}}\Big[Y-T_\ell(\fhat(X))\, \Big|\, T_{\ell}(\fhat(X))\Big]\geq 0\, \text{ and }\, \Exp_{(X,Y)\sim \Prob_{n+1}}\Big[Y-T_h(\fhat(X))\, \Big|\, T_h(\fhat(X))\Big]\leq 0.
        \end{equation*}
    \end{prop}

\begin{proof}

        By assumption, \(\Prob(Y_{n+1}\in [\ell_{n+1},h_{n+1}]\, | \, \A)\geq 1-\alpha\).
        Denote \(\Prob(Y_{n+1}> h_{n+1}\, | \, \A)=\alpha_1\) and \(\Prob(Y_{n+1}< \ell_{n+1}\, | \, \A)=\alpha_2\).
        It holds that \(\alpha_1+\alpha_2\leq \alpha\).
        The function \(T_\ell\) is defined as the isotonic regression fit to the points
        \begin{equation*}
            (\fhat(X_1),Y_1),\ldots, (\fhat(X_n),Y_n),(\fhat(X_{n+1}),\ell_{n+1}),
        \end{equation*}
        so it induces a partition \(D_1,\ldots,D_k\) of the index set \(\{1,\ldots,n+1\}\), such that, for all indices \(j\in D_i\),
        \begin{equation*}
            T_\ell(\fhat(X_j))=\frac{1}{\# D_i}\sum_{s\in D_i} Y_s^\prime=\vcentcolon t_i
        \end{equation*}
        where \(Y_s^\prime=Y_s\) for \(s\neq n+1\) and \(Y_{n+1}^\prime = \ell_{n+1}\).
        For all \(s\in \{1,\ldots,n\}\), it holds that \(Y_s\geq Y_s^\prime\).
        It also holds that \(Y_{n+1}\geq Y_{n+1}^\prime\) with probability \(1-\alpha_2\) conditionally on $\A$.
        Denote by \(j^\star\in \{1,\ldots,k\}\) the index such that \(n+1\in D_{j^\star}\).
        It holds that
        \begingroup
        \allowdisplaybreaks
        \begin{align*}
            \Prob\Big(\Exp_{(X,Y)\sim \Prob_{n+1}}
            &\Big[Y-T_\ell(\fhat(X))\, \Big|\, T_\ell(\fhat(X))\Big]<0\, \Big | \, \A\Big)\\
            &=\Prob\left(\bigcup_{j=1}^k \left\{\Exp_{(X,Y)\sim \Prob_{n+1}}\Big[Y\, \Big|\, T_\ell(\fhat(X))=t_j\Big]<t_j\right\} \, \middle| \, \A\right)\\
            & \leq \sum_{j=1}^k \Prob\left(\Exp_{(X,Y)\sim \Prob_{n+1}}\Big[Y\, \Big|\, T_\ell(\fhat(X))=t_j\Big]<t_j \, \middle| \, \A\right)\\
            &=\sum_{j=1}^k \Prob\left(\frac{1}{\# D_j}\sum_{i\in D_j} Y_i<t_j \, \middle| \, \A\right)\\
            & = \sum_{j=1}^k \Prob\left(\frac{1}{\# D_j}\sum_{i\in D_j} Y_i<\frac{1}{\# D_j}\sum_{i\in D_j} Y_i^\prime \, \middle| \, \A\right)\\
            & = \Prob\left(\frac{1}{\# D_{j^\star}}\sum_{i\in D_{j^\star}} Y_i<\frac{1}{\# D_{j^\star}}\sum_{i\in D_{j^\star}} Y_i^\prime \, \middle| \, \A\right)\\
            & = \Prob\left(Y_{n+1}<Y_{n+1}^\prime\, \middle| \, \A\right)\\
            & = \alpha_2.
        \end{align*}
        \endgroup
        In a similar way, we can show that
        \begin{equation*}
            \Prob\Big(\Exp_{(X,Y)\sim \Prob_{n+1}}\Big[Y-T_h(\fhat(X))\, \Big|\, T_h(\fhat(X))\Big]>0\, \Big | \, \A\Big)\leq \alpha_1.
        \end{equation*}
        It follows that, with probability at least \(1-\alpha\) conditionally on \(\A\),
        \begin{equation*}
            \Exp_{(X,Y)\sim \Prob_{n+1}}\Big[Y-T_\ell(\fhat(X))\, \Big|\, T_{\ell}(\fhat(X))\Big]\geq 0\, \text{ and }\, \Exp_{(X,Y)\sim \Prob_{n+1}}\Big[Y-T_h(\fhat(X))\, \Big|\, T_h(\fhat(X))\Big]\leq 0.
        \end{equation*}
\end{proof}

An interesting direction for future work would be the extension of the above proposition to the population level.
The validity of this extension would give further insight on whether \(T_\ell\) and \(T_h\) could really be seen as \emph{calibration boundaries} or whether calibrated predictions can actually also lie outside of the interval \(\C_{n+1,\alpha}\).

    \section{Proofs of results in Section \ref{sec: predictive_performance}}
    \label{sec: app_Proofs_4}

    \begin{proof}[Proof of Lemma \ref{lem: predictive_power_variance}]
        It holds that
        \begin{align*}
            \Exp\Big[(Y-\ftilde(X))^2\Big]
            & = \Exp\Big[(Y-\mu(X)+\mu(X)-\ftilde(X))^2\Big]\\
            & = \Exp\Big[(Y-\mu(X))^2\Big] + \Exp\Big[(\mu(X)-\ftilde(X))^2\Big] + 2\Exp\Big[(Y-\mu(X))(\mu(X)-\ftilde(X))\Big].
        \end{align*}
        Furthermore,
        \begin{align*}
            \Exp\Big[(Y-\mu(X))\mu(X)\Big]
            & = \Exp\Big[\Exp\Big[(Y-\mu(X))\mu(X)\, \Big|\, \mu(X)\Big]\Big]\\
            & = \Exp\Big[\mu(X)\cdot \Exp[Y-\mu(X)|\mu(X)]\Big] = 0,
        \end{align*}
        where the last step follows from the fact that \(\mu(X)\) is calibrated.
        Thus, it follows from the previous equation that
        \begingroup
        \allowdisplaybreaks
        \begin{align*}
            \Exp\Big[(Y-\ftilde(X))^2\Big]
            & = \Exp\Big[(Y-\mu(X))^2\Big] + \Exp\Big[(\mu(X)-\ftilde(X))^2\Big] - 2\Exp\Big[(Y-\mu(X))\ftilde(X)\Big]\\
            & = \Exp\Big[(Y-\mu(X))^2\Big] + \Exp[(\mu(X))^2] + \Exp[(\ftilde(X))^2] - 2\Exp[Y\ftilde(X)].
        \end{align*}
        \endgroup
        It holds that \(\Exp[Y\ftilde(X)]
            = \Exp\Big[\Exp\Big[Y\ftilde(X)\, \Big |\, \ftilde(X)\Big]\Big]
            = \Exp\Big[\ftilde(X)\cdot \Exp[Y|\ftilde(X)]\Big]
            = \Exp[(\ftilde(X))^2]\).
        Plugging that into the previous equation yields that
        \begin{align*}
            \Exp\Big[(Y-\ftilde(X))^2\Big] + \Exp\Big[(\ftilde(X))^2\Big] = \Exp\Big[(Y-\mu(X))^2\Big]+\Exp\Big[(\mu(X))^2\Big].
        \end{align*}
        Equation \eqref{eq: predictive_power_variance} follows from the fact that \(\Exp[\ftilde(X)]=\Exp[\mu(X)]=\Exp[Y]\).
    \end{proof}

    We now move on to the proof of Proposition \ref{prop: pred_performance_binning}.
    We first show that the provided excess risk bound can be tight in extreme cases.
    Indeed, suppose that the partition consists of only one block.
    In that case, \(\mathsf{c}(f^\star(X_{n+1}))\) is equal to the sample average of \(Y_1,\ldots,Y_{n+1}\).
    As a result, Lemma \ref{lem: predictive_power_variance} yields that the excess risk is equal to
    \begin{align*}
        \Exp\left[\left(Y_{n+1}-\frac{1}{n+1}\sum_{i=1}^{n+1}Y_i\right)^2\right]-\Exp\Big[(Y_{n+1}-f^\star(X_{n+1}))^2\Big]
        & = \mathrm{Var}(f^\star(X_{n+1}))-\mathrm{Var}\left(\frac{1}{n+1}\sum_{i=1}^{n+1} Y_i\right)\\
        & = \sigma_{f^\star}^2-\frac{1}{(n+1)^2}\sum_{i=1}^{n+1}\mathrm{Var}(Y_i)\\
        & = \frac{n}{n+1}\sigma_{f^\star}^2-\frac{1}{n+1}\sigma_{\varepsilon}^2,
    \end{align*}
    which matches the upper bound in Proposition \ref{prop: pred_performance_binning}.
    As \(n\to \infty\), this converges to \(\sigma_{f^\star}^2\).
    \begin{proof}[Proof of Proposition \ref{prop: pred_performance_binning}]
        Let $\Prob_{n+1}$ be the empirical distribution of $(X_1,Y_1),\dots,(X_{n+1},Y_{n+1})$.
        We will compute the first and second moments of \(\mathsf{c}(f^\star(X_{n+1}))\) and then we will use the result of Lemma \ref{lem: predictive_power_variance}.
        It follows by exchangeability that
        \begin{align}
            \label{eq: 2nd_moment_exchangeability}
            \Exp\Big[\mathsf{c}(f^\star(X_{n+1}))^2\Big]
            = \Exp\Big[\Exp\Big[\mathsf{c}(f^\star(X_{n+1}))^2\Big]\, \Big | \, \Prob_{n+1}\Big]\nonumber
            & = \Exp\left[\frac{1}{n+1}\sum_{i=1}^{n+1} \left(\mathsf{c}(f^\star(X_i))\right)^2\right]\nonumber \\
            & = \Exp\left[\frac{1}{n+1}\sum_{j=1}^k \#D_j\left(\frac{1}{\#D_j}\sum_{s \in D_j}Y_s\right)^2\right]\nonumber \\
            & = \Exp\left[\frac{1}{n+1}\sum_{j=1}^k \frac{1}{\#D_j} \left(\sum_{s,t \in D_j}Y_sY_t\right)\right].
        \end{align}
        Also,
        \begin{align}
            \label{eq: second_moment}
            \Exp
            &\left[\frac{1}{n+1}\sum_{j=1}^k \frac{1}{\#D_j}\sum_{s,t \in D_j}Y_s Y_t\, \middle| \, \S_1,X_1,\ldots,X_{n+1}\right]\nonumber \\
            & = \frac{1}{n+1}\sum_{j=1}^k \left(\Exp\left[\frac{1}{\#D_j} \sum_{s,t\in D_j} f^\star(X_s)f^\star(X_t) + \frac{1}{\#D_j}\sum_{s,t\in D_j} \varepsilon_s\varepsilon_t\, \middle | \, \S_1,X_1,\ldots,X_{n+1}\right]\right),
        \end{align}
        where the cross terms have been omitted because they are equal to zero.
        Indeed, since \(\varepsilon_1,\ldots,\varepsilon_{n+1}\) are independent from \(D_1,\ldots,D_k\), and since \(\#D_1,\ldots,\#D_k\) are \((\S_1,X_1,\ldots,X_{n+1})\)-measurable, this term can be written as
        \begin{align*}
            \sum_{j=1}^k\Exp\left[\frac{1}{\#D_j}\sum_{s,t\in D_j} f^\star(X_s)\varepsilon_t\, \middle | \, \S_1,X_1,\ldots,X_{n+1}\right]
            & = \sum_{j=1}^k \frac{1}{\# D_j}\sum_{s,t\in D_j}f^\star(X_s)\Exp[\varepsilon_t\, | \, \S_1, X_1,\ldots,X_{n+1}]\\
            & = \sum_{j=1}^k \frac{1}{\# D_j}\sum_{s,t\in D_j}f^\star(X_s)\Exp[\varepsilon_t]\\
            & = 0,
        \end{align*}
        where we also used the fact that \(\# D_j\) is a measurable function of \(X_1,\ldots,X_{n+1}\).

        \begingroup
        \allowdisplaybreaks
        The second term in \eqref{eq: second_moment} is equal to
        \begin{align*}
            \sum_{j=1}^k\Exp\left[\frac{1}{\#D_j}\sum_{s,t\in D_j} \varepsilon_s\varepsilon_t\, \middle | \, \S_1,X_1,\ldots,X_{n+1}\right]
            & = \sum_{j=1}^k\frac{1}{\# D_j} \sum_{s,t\in D_j}\Exp\left[ \varepsilon_s\varepsilon_t\, \middle | \, \S_1, X_1,\ldots,X_{n+1}\right]\\
            & = \sum_{j=1}^k \frac{1}{\# D_j}\sum_{s,t\in D_j}\Exp[\varepsilon_s\varepsilon_t]\\
            & = \sum_{j=1}^k \frac{1}{\# D_j}\sum_{s\in D_j} \Exp[\varepsilon_s^2]\\
            & = \sum_{j=1}^k \sigma_\varepsilon^2\\
            & = k\sigma_\varepsilon^2.
        \end{align*}
        \endgroup
        The first term in \eqref{eq: second_moment} is equal to
        \begin{equation*}
            \sum_{j=1}^k \Exp\left[(\# D_j)\left(\frac{\sum_{s\in D_j} f^\star(X_s)}{\# D_j}\right)^2\, \middle | \, \S_1,X_1,\ldots,X_{n+1}\right]
        \end{equation*}
        By Cauchy-Schwarz,
        \begin{align*}
            \sum_{j=1}^k (\#D_j)\cdot \sum_{j=1}^k\left[(\# D_j)\left(\frac{\sum_{s\in D_j} f^\star(X_s)}{\# D_j}\right)^2\right]
            & \geq \left[\sum_{j=1}^k (\# D_j)\frac{\sum_{s\in D_j} f^\star(X_s)}{\# D_j}\right]^2\\
            & = \left(\sum_{s=1}^{n+1} f^\star(X_s)\right)^2.
        \end{align*}
        Since \(\sum_{j=1}^k (\# D_j)=n+1\), it follows by \eqref{eq: 2nd_moment_exchangeability}, \eqref{eq: second_moment} that
        \begin{align*}
            \Exp\Big[\mathsf{c}(f^\star(X_{n+1}))^2\Big]
            & = \Exp\Big[\Exp\Big[\mathsf{c}(f^\star(X_{n+1}))^2\, \Big| \, \S_1,X_1,\ldots,X_{n+1}\Big]\Big]\\
            & \geq \Exp\left[\frac{1}{(n+1)^2}\left(\sum_{s=1}^{n+1} f^\star(X_s)\right)^2+ \frac{k}{n+1}\sigma_\varepsilon^2\right]\\
            & = \frac{1}{(n+1)^2}\Exp\left[\left(\sum_{s=1}^{n+1} f^\star(X_s)\right)^2\right]+ \frac{\Exp[k]}{n+1}\sigma_\varepsilon^2,
        \end{align*}
        The computation of \(\Exp[\mathsf{c}(f^\star(X_{n+1}))]\) is simpler.
        Indeed, by exchangeability, it follows that
        \begingroup
        \allowdisplaybreaks
        \begin{align*}
            \Exp\Big[\mathsf{c}(f^\star(X_{n+1}))\Big]
            =\Exp\Big[\Exp\Big[\mathsf{c}(f^\star(X_{n+1}))\, \Big | \, \Prob_{n+1}\Big]\Big]
            & = \Exp\left[\frac{1}{n+1}\sum_{i=1}^{n+1} \mathsf{c}(f^\star(X_i))\right] \\
            & = \Exp\left[\frac{1}{n+1}\sum_{j=1}^k \#D_j\left(\frac{1}{\#D_j}\sum_{s \in D_j}Y_s\right)\right] \\
            & = \Exp\left[\frac{1}{n+1}\sum_{j=1}^k \sum_{\ell\in D_j} Y_\ell\right]\\
            & = \Exp\left[\frac{1}{n+1}\sum_{s=1}^{n+1} Y_s\right]\\
            & = \Exp\Big[f^\star(X_{n+1})\Big].
        \end{align*}
        \endgroup
        Finally, it follows from Lemma \ref{lem: predictive_power_variance} that
        \begingroup
        \allowdisplaybreaks
        \begin{align*}
            \Exp\Big[(Y_{n+1}-
            & \mathsf{c}(f^\star(X_{n+1})))^2\Big]-\Exp\Big[(Y_{n+1}-f^\star (X_{n+1}))^2\Big]\\
            & = \mathrm{Var}(f^\star(X_{n+1}))-\mathrm{Var}(\mathsf{c}(f^\star(X_{n+1})))\\
            & = \Exp\Big[f^\star(X_{n+1})^2\Big]-\left(\Exp\Big[f^\star(X_{n+1})\Big]\right)^2-\Exp\Big[\mathsf{c}(f^\star(X_{n+1}))^2\Big]+\left(\Exp\Big[\mathsf{c}(f^\star(X_{n+1}))\Big]\right)^2\\
            & = \Exp\Big[f^\star(X_{n+1})^2\Big]-\Exp\Big[\mathsf{c}(f^\star(X_{n+1}))^2\Big]\\
            & \leq \Exp\Big[f^\star(X_{n+1})^2\Big]-\frac{1}{(n+1)^2}\Exp\left[\left(\sum_{s=1}^{n+1} f^\star(X_s)\right)^2\right]-\frac{\Exp[k]}{n+1}\sigma_\varepsilon^2\\
            & = \Exp\left[f^\star(X_{n+1})^2-\left(\frac{\sum_{s=1}^{n+1} f^\star(X_s)}{n+1}\right)^2\right]-\frac{\Exp[k]}{n+1}\sigma_\varepsilon^2\\
            & = \mathrm{Var}(f^\star(X_{n+1}))-\mathrm{Var}\left(\frac{\sum_{s=1}^{n+1} f^\star(X_s)}{n+1}\right)-\frac{\Exp[k]}{n+1}\sigma_\varepsilon^2\\
            & = \frac{n}{n+1}\sigma_{f^\star}^2-\frac{\Exp[k]}{n+1}\sigma_\varepsilon^2,
        \end{align*}
        \endgroup
        which finishes the proof.
    \end{proof}

    To prove Proposition \ref{prop: pred_performance_isotonic}, we first need to show the following lemma, which shows that isotonic regression reduces the variance in-sample.

    \begin{lem}
        \label{lem: variance_isotonic_regression}
        In the setting of Proposition \ref{prop: pred_performance_isotonic}, it holds that
        \begin{equation*}
            \mathrm{Var}\left(\mathsf{c}(f^\star(X_{n+1}))\right)\leq \mathrm{Var}(Y_{n+1})=\sigma_{f^\star}^2+\sigma_{\varepsilon}^2.
        \end{equation*}
    \end{lem}

    \begin{proof}
        Without loss of generality, we may assume that \(f^\star(X_1)\leq \ldots\leq f^\star(X_{n+1})\).
        Then, the vector \(\boldsymbol{C}\vcentcolon = (\mathsf{c}(f^\star(X_1)),\ldots,\mathsf{c}(f^\star(X_{n+1})))\) is the projection of \(\boldsymbol{Y}\vcentcolon = (Y_1,\ldots,Y_{n+1})\) onto the convex cone \(\C\vcentcolon = \left\{\boldsymbol{\theta}\in \R^{n+1}:\theta_1\leq \ldots \leq \theta_{n+1}\right\}\).
        Therefore, from the Pythagorean theorem, it follows that \(\lVert \boldsymbol{Y}\rVert_2^2=\lVert \boldsymbol{Y}-\boldsymbol{C}\rVert_2^2+\lVert \boldsymbol{C}\rVert_2^2\geq \lVert \boldsymbol{C}\rVert_2^2\), which is equivalent to
        \begin{equation*}
            \sum_{i=1}^{n+1} \mathsf{c}(f^\star(X_i))^2\leq \sum_{i=1}^{n+1} Y_i^2.
        \end{equation*}
        This yields that
        \begin{equation}
            \label{eq: projection_norm}
            \Exp\left[\mathsf{c}(f^\star(X_{n+1}))^2\right]\overset{\text{exch.}}{=}\frac{1}{n+1}\Exp\left[\sum_{i=1}^{n+1} \mathsf{c}(f^\star(X_i))^2\right]\leq \frac{1}{n+1}\Exp\left[\sum_{i=1}^{n+1} Y_i^2\right]\overset{\text{exch.}}{=}\Exp[Y_{n+1}^2].
        \end{equation}
        The tower property and the fact that \(\mathsf{c}(f^\star(X_{n+1}))\) is calibrated yield that
        \begin{equation*}
            \Exp[Y_{n+1}]=\Exp\Big[\Exp[Y_{n+1}|\mathsf{c}(f^\star(X_{n+1}))]\Big]=\Exp\left[\mathsf{c}(f^\star(X_{n+1}))\right].
        \end{equation*}
        Thus, \eqref{eq: projection_norm} yields that \(\mathrm{Var}\left(\mathsf{c}(f^\star(X_{n+1}))\right)\leq \mathrm{Var}(Y_{n+1})=\sigma_{f^\star}^2+\sigma_{\varepsilon}^2\), which finishes the proof.
    \end{proof}

    \begin{proof}[Proof of Proposition \ref{prop: pred_performance_isotonic}]
        Since \(\Exp[f^\star(X_{n+1})]=\Exp[\mathsf{c}(f^\star(X_{n+1}))]\), we have
        \begin{align}
            \label{eq: pred_performance_intermediate}
            \Big|\mathrm{Var}(f^\star(X_{n+1}))
            &- \mathrm{Var}(\mathsf{c}(f^\star(X_{n+1})))\Big|\nonumber \\
            & = \Big|\Exp\Big[f^\star(X_{n+1})^2 - \mathsf{c}(f^\star(X_{n+1}))^2\Big]\Big| \nonumber \\
            & = \Big|\Exp\Big[\Big(f^\star(X_{n+1}) - \mathsf{c}(f^\star(X_{n+1}))\Big)\Big(f^\star(X_{n+1}) + \mathsf{c}(f^\star(X_{n+1}))\Big)\Big]\Big|.
        \end{align}
        It holds that
        \begin{equation*}
            \Exp\Big[\left(f^\star(X_{n+1})-\mathsf{c}(f^\star(X_{n+1}))\right)\Exp[f^\star(X_{n+1})]\Big]=0,
        \end{equation*}
        so, by the triangle inequality, the last term in \eqref{eq: pred_performance_intermediate} can be bounded by the sum of
        \begin{equation*}
            \Big|\Exp\Big[\Big(f^\star(X_{n+1}) - \mathsf{c}(f^\star(X_{n+1}))\Big)\Big(f^\star(X_{n+1}) -\Exp[f^\star(X_{n+1})]\Big)\Big]\Big|
        \end{equation*}
        and 
        \begin{equation*}
            \Big|\Exp\Big[\Big(f^\star(X_{n+1}) - \mathsf{c}(f^\star(X_{n+1}))\Big)\Big(\mathsf{c}(f^\star(X_{n+1})) -\Exp[\mathsf{c}(f^\star(X_{n+1}))]\Big)\Big]\Big|
        \end{equation*}
        From Cauchy-Schwarz, it follows that the sum of these two terms is upper bounded by
        \begin{align}
            \label{eq: pred_performance_CS}
            \Exp\Big[(f^\star(X_{n+1})-\mathsf{c}(f^\star(X_{n+1})))^2
            &\Big]^{1/2}\Big(\sqrt{\mathrm{Var}\left(f^\star(X_{n+1})\right)}+\sqrt{\mathrm{Var}\left(\mathsf{c}(f^\star(X_{n+1}))\right)}\Big)\nonumber \\
            & \leq \Big(\sigma_{f^\star}+\sqrt{\sigma_{\varepsilon}^2+\sigma_{f^\star}^2}\Big)\Exp\Big[(f^\star(X_{n+1})-\mathsf{c}(f^\star(X_{n+1})))^2\Big]^{1/2}
        \end{align}
        where the last inequality follows from Lemma \ref{lem: variance_isotonic_regression}.
        The quantity
        \begin{equation}
        \label{eq: risk_cond_expectation}
            \frac{1}{n+1}\sum_{i=1}^{n+1}\Exp\Big[\left(f^\star(X_i) - \mathsf{c}(f^\star(X_i))\right)^2 \, \Big|\, f^\star(X_1),\dots,f^\star(X_{n+1})\Big]
        \end{equation}
        can be written as $\Exp\left[\ell^2\left(\bm{f^\star}, \bm{C}\right)\, \middle| \,f^\star(X_1),\dots,f^\star(X_{n+1})\right]$, where
        \begin{equation*}
        \ell^2\left(\bm{a},\bm{b}\right)\vcentcolon =\frac{1}{n+1}\sum_{i=1}^{n+1} \left(a_i-b_i\right)^2.
        \end{equation*}
        and \(\bm{f^\star}=(f^\star(X_1),\ldots,f^\star(X_{n+1}))\),  \(\bm{C}\vcentcolon = (\mathsf{c}(f^\star(X_1)),\ldots,\mathsf{c}(f^\star(X_{n+1})))\),
        Therefore, by \citet[Theorem 2.2]{Zhang_2002}, it follows that
        \begin{multline}
        \label{eq: risk_cond_bound}
        \frac{1}{n+1}\sum_{i=1}^{n+1}\Exp\left[(f^\star(X_i) - \mathsf{c}(f^\star(X_{i})))^2 \, \middle| \,   f^\star(X_1),\dots,f^\star(X_{n+1})\right] \\
        \lesssim \left(\frac{\sigma_{\varepsilon}^2 \left[\max_i f^\star(X_i)-\min_i f^\star(X_i)\right]}{n+1}\right)^{2/3}+\frac{\sigma_{\varepsilon}^2\log n}{n+1}.
        \end{multline}
        The expectation of the left-hand side is equal to $\Exp[(f^\star(X_{n+1}) - \mathsf{c}(f^\star(X_{n+1})))^2]$ due to exchangeability.
        Therefore, \eqref{eq: risk_cond_bound} yields that
        \begin{equation*}
        \Exp\Big[(f^\star(X_{n+1}) - \mathsf{c}(f^\star(X_{n+1})))^2\Big]\leq \Exp\left(\frac{\sigma_{\varepsilon}^2 \left[\max_i f^\star(X_i)-\min_i f^\star(X_i)\right]}{n+1}\right)^{2/3}+\frac{\sigma_{\varepsilon}^2\log (n+1)}{n+1}.
        \end{equation*}
        The claim now follows from \eqref{eq: pred_performance_CS} and Lemma \ref{lem: predictive_power_variance}.
    \end{proof}

    \section{Proof of Theorem \ref{thm: consistency}}\label{sec: app_proof_consistency}

    Before we start the proof of the theorem, we provide some interpretation of Assumptions \ref{assum: consistency_iid}-\ref{assum: consistency_subgaussian}.
    Assumption \ref{assum: consistency_iid} is a standard assumption that we also made in previous sections and Assumption \ref{assum: consistency_bounded_mu} ensures that there are enough points \(\mu(X_i)\) in every interval \(I\subseteq \R\) \citep{Henzi_2023, Moesching_2020}.
    Assumption \ref{assum: consistency_fhat} states that the base model \(\fhat_r\) consistently estimates a monotone transformation of the conditional mean \(\mu\).
    This assumption is convenient for isotonic recalibration as it ensures that \(\fhat_r\) approximates a monotone transformation of the conditional mean \(\mu(x)\).
    This implies that, asymptotically, \(\fhat_r\) can correctly order \(\Exp[Y_1|X_1],\ldots,\Exp[Y_{n+1}|X_{n+1}]\), which justifies fitting isotonic regression to \((\fhat_r(X_1),Y_1),\ldots,(\fhat_r(X_{n+1}),Y_{n+1})\).
    Finally, Assumption \ref{assum: consistency_subgaussian} is necessary for concentration arguments.
    This assumption was not necessary in similar results in \citet{henzi2021isotonic, Henzi_2023, Moesching_2020}, which only treated isotonic regression with binary labels.
    
    We first introduce some notation that will be used in the proof.
    First of all, we define the event
    \begin{equation}
    \label{eq: B_r}
    B_r=\left\{\sup_{x\in \X}\left|T(\fhat_r(x))-\mu(x)\right|\leq C_0\delta_r\right\},
    \end{equation}
    where \(\delta_r = (\log r/r)^{1/3}\).
    Assumption \ref{assum: consistency_fhat} yields that \(\Prob(B_r)\overset{r\to \infty}{\longrightarrow} 1\).
    We define
    \begin{align*}
    w_{uv}=v-u+1,
    \quad
    &F_{uv}=\frac{1}{w_{uv}}\sum_{i=u}^v Y_{i}\\
    \overline{F}_{uv}=\frac{1}{w_{uv}}\sum_{i=u}^v \mu\left(X_{i}\right),
    \quad
    &\widehat{F}_{uv} = \frac{1}{w_{uv}}\sum_{i=u}^v T\left(\fhat_r\left(X_{i}\right)\right).
    \end{align*}
    Set
    \begin{equation}
        \label{eq: M_n_pi}
        M_{n}=\max_{1\leq u\leq v\leq n}w_{uv}^{1/2}\left|F_{uv}-\overline{F}_{uv}\right|.
    \end{equation}
    The proof of Theorem \ref{thm: consistency} is largely based on the following lemma, which is analogous to \citet[Lemma~A.2]{Henzi_2023}.
    
    \begin{lem}
        \label{lem: M_n_pi_bound}
        Under Assumptions \ref{assum: consistency_iid}-\ref{assum: consistency_subgaussian}, there exists a constant \(s=s(\sigma)\), such that
        \begin{equation*}
            \lim_{n\to \infty} \Prob\left(M_{n} \geq s(\log n)^{1/2}\right)=0.
        \end{equation*}
    \end{lem}

    \begin{proof}
        Since \(\{Y_i-\mu(X_i)\}_{i=1}^{\infty}\) are conditionally sub-Gaussian, it follows that, for any fixed \(t>0\),
        \begin{align*}
            \Prob\left(w_{uv}^{1/2}\left|F_{uv}-\overline{F}_{uv}\right|\geq t\right)
             = \Prob\left(\left|F_{uv}-\overline{F}_{uv}\right|\geq \frac{t}{w_{uv}^{1/2}}\right)
            & \leq 2\exp\left\{-\frac{cw_{uv}t^2}{w_{uv}\sigma^2}\right\}\\
            & = 2\exp\left\{-\frac{ct^2}{\sigma^2}\right\}
        \end{align*}
        for some absolute constant \(c\).
        Therefore, a union bound yields
        \begingroup
        \allowdisplaybreaks
        \begin{align*}
            \Prob\left(M_{n} \geq s(\log n)^{1/2}\right)
            & \leq \sum_{1\leq u\leq v\leq n} \Prob\left(w_{uv}^{1/2}\left|F_{uv}-\overline{F}_{uv}\right|\geq s(\log n)^{1/2}\right)\\
            & =2\left(\binom{n}{2}+n\right)\exp\left\{-\frac{cs^2\log n}{\sigma^2}\right\}\\
            & = n(n+1)\cdot n^{-cs^2/\sigma^2}.
        \end{align*}
        \endgroup
        Choosing \(s>\sigma\sqrt{2/c}\) implies that \(\lim_{n\to \infty} \Prob\left(M_{n} \geq s(\log n)^{1/2}\right) = 0.\)
    \end{proof}

    We now restate a result from \citet{Moesching_2020} which we use in the proof of the next lemma.

    \begin{prop}[\citet{Moesching_2020}]
        \label{prop: Moesching_2020}
        Let \(\mathbb{Q}_n\) be the empirical distribution of independent random variables \(X_1,\ldots,X_n\) with distribution \(\mathbb{Q}\).
        Let \(\delta_n>0\) be a sequence such that \(\delta_n\longrightarrow 0\) and \(n\delta_n/\log n \longrightarrow \infty\) as $n\to \infty$.
        Then, there exist numbers \(\epsilon_n\) such that \(\epsilon_n\rightarrow 0\) and
        \begin{equation*}
            \inf\left\{\frac{\mathbb{Q}_n(I)}{\mathbb{Q}(I)}\, \middle| \, \text{intervals } I\subset \R \text{ with } \mathbb{Q}(I)\geq \delta_n\right\}\geq 1-\epsilon_n
        \end{equation*}
        with asymptotic probability 1.
    \end{prop}
    The sequence \(\delta_n=(\log n/n)^{1/3}\) that we defined earlier satisfies the conditions of Proposition \ref{prop: Moesching_2020}.
    The next lemma, which was proved by \citet{Henzi_2023}, shows that the set \(\{T(\fhat_r(X_{j})):j=1,\ldots,n\}\) is asymptotically dense in \(I\).
    For the sake of completeness, we also provide a detailed proof here.
    If no further details are given, asymptotic statements refer to the fully asymptotic regime \(n,r\to \infty\).

    \begin{lem}[\citet{Henzi_2023}]
        \label{lem: g_hat_dense}
        For any set \(B\subset I\), we define
        \begin{equation*}
            \widehat{w}(B)=\# \left\{j\in \{1,\ldots,n\}: T(\fhat_r(X_{j}))\in B\right\}.
        \end{equation*}
        Then, under Assumptions \ref{assum: consistency_iid}-\ref{assum: consistency_subgaussian}, and if \(r\leq n\), the event
        \begin{equation*}
            \left\{\inf\left\{\frac{\widehat{w}(I_r)}{n\lambda(I_r)}: \text{intervals } I_r\subset I \text{ with } \lambda(I_r)\geq 4C_0\delta_r\right\}\geq D\right\}
        \end{equation*}
        has asymptotic probability \(1\) for any \(D<C_1/2\), where \(\delta_r=(\log r/r)^{1/3}\).
    \end{lem}

    \begin{proof}
        Define \(w(B)=\# \{j\in \{1,\ldots,n\}:\mu(X_{j})\in B\}\) for \(B\subseteq I\).
        On the event \(B_r\) defined in \eqref{eq: B_r}, and for any interval \(J\subseteq I\) with \(\lambda(J)\geq 4C_0\delta_r\), it holds that
        \begin{align}
            \label{eq: w_inequality}
            \widehat{w}(J)-w(J)
            &\geq -\# \left\{j\in \{1,\ldots,n\}:T(\fhat_r(X_{j}))\notin J, \mu(X_j)\in J\right\}\nonumber \\
            & \geq -w\left(\{z\in J: z+C_0\delta_r\notin J \text{ or } z-C_0\delta_r\notin J\right\}).
        \end{align}
        The first inequality follows from the fact that
        \begin{align*}
            w(J)-
            &\# \left\{j\in \{1,\ldots,n\}:T(\fhat_r(X_{j}))\notin J,\quad \mu(X_{j})\in J\right\} \nonumber \\
            & =\# \{j\in \{1,\ldots,n\}:\mu(X_{j})\in J\} - \# \left\{j\in \{1,\ldots,n\}:T(\fhat_r(X_{j}))\notin J,\quad \mu(X_{j})\in J\right\}\nonumber \\
            & = \# \{j\in \{1,\ldots,n\}: \mu(X_{j})\in J,\quad T(\fhat_r(X_{j}))\in J\}\nonumber \\
            & \leq \widehat{w}(J).
        \end{align*}
        The second inequality follows from the fact that, if \(T(\fhat_r(X_{j}))\notin J\) and \(\mu(X_{j})\in J\), then, since we are under \(B_r\), it must hold that either \(\mu(X_{j})+C_0\delta_r\notin J\) or \(\mu(X_{j})-C_0\delta_r\notin J\) (that is, \(\mu(X_{j})\) must be close to the endpoints of \(J\), because it is close to \(T(\fhat_r(X_{j}))\), which does not belong to \(J\)).
        By \eqref{eq: w_inequality}, it follows that
        \begin{equation}
            \label{eq: what_inequality}
            \widehat{w}(J)\geq w\Big(J\setminus\{z\in J: z+C_0\delta_r \notin J \text{ or } z-C_0\delta_r\notin J\Big\}\Big).
        \end{equation}
        For any interval \(I_r\subset I\) of length at least \(4C_0\delta_r\), the set
        \begin{equation*}
            \widetilde{I}_r=I_r\setminus \{z\in I_r:z+C_0\delta_r\notin I_r \text{ or } z-C_0\delta_r\notin I_r\}
        \end{equation*}
        is an interval of length
        \begin{equation*}
            \lambda(\widetilde{I}_r)=\lambda(I_r)-2C_0\delta_r\geq \lambda(I_r)-\frac{\lambda(I_r)}{2}=\frac{\lambda(I_r)}{2}.
        \end{equation*}
        That is because \(\widetilde{I}_r\) essentially consists of the points that are at least \(C_0\delta_r\) away from both endpoints of \(I_r\).
        Thus, we need to exclude those points that are too close to the boundaries.
    
        By \eqref{eq: what_inequality}, it follows that \(\widehat{w}(I_r)\geq w(\widetilde{I}_r)\), for any interval \(I_r\subseteq I\) with \(\lambda(I_r)>4C_0\delta_r\).
        Thus,
        \begin{align*}
            \widehat{m}_{r,n}\vcentcolon
            &= \inf\left\{\frac{\widehat{w}(I_r)}{n\lambda(I_r)}:\text{intervals } I_r\subset I\text{ with } \lambda(I_r)\geq 4C_0\delta_r\right\}\\
            & \geq \frac{1}{2}\inf\left\{\frac{w(\widetilde{I}_r)}{n\lambda(\widetilde{I}_r)}:\text{intervals } \widetilde{I}_r\subset I \text{ with } \lambda(\widetilde{I}_r)\geq 2C_0\delta_r\right\}\\
            & =\vcentcolon\frac{1}{2} m_{r,n}.
        \end{align*}
        Define \(\widehat{A}_{r,n}=\{\widehat{m}_{r,n}\geq D\}\) and \(A_{r,n}=\{m_{r,n}\geq 2D\}\) for \(D<C_1/2\).
        Then, \(A_{r,n}\cap B_r\subseteq \widehat{A}_{r,n}\) and
        \begin{equation*}
            \Prob(\widehat{A}_{r,n})\geq \Prob(\widehat{A}_{r,n}\cap B_r)\geq \Prob(A_{r,n})+\Prob(B_r)-1\rightarrow 1\text{ as } n,r\rightarrow \infty,
        \end{equation*}
        by Assumption \ref{assum: consistency_fhat} and Proposition \ref{prop: Moesching_2020}.
        In this setting, Proposition \ref{prop: Moesching_2020} is applied as follows: we consider the sequence \(\eta_n=2C_0C_1(\log n/n)^{1/3}\) and the random variables \(\mu(X_1),\ldots,\mu(X_n)\) with empirical measure \(\mathbb{Q}_n\) and true distribution \(\mathbb{Q}\).
        Then,
        \begin{equation*}
            \frac{w(\widetilde{I}_r)}{n\lambda(\widetilde{I}_r)}=\frac{\mathbb{Q}_n(\widetilde{I}_r)}{\lambda(\widetilde{I}_r)}\geq C_1\frac{\mathbb{Q}_n(\widetilde{I}_r)}{\mathbb{Q}(\widetilde{I}_r)},
        \end{equation*}
        where the last inequality follows from Assumption \ref{assum: consistency_bounded_mu}.
        Since \(r\leq n\), it follows that \(\lambda(\widetilde{I}_r)\geq 2C_0(\log r/r)^{1/3}\geq 2C_0(\log n/n)^{1/3}\), so the assumptions of Proposition \ref{prop: Moesching_2020} for the sequence \(\eta_n\) are satisfied.
    \end{proof}

    We are now ready to prove Theorem \ref{thm: consistency}.
    Due to its asymptotic nature, our proof covers both the in-sample calibrated version \(G(\Prob_n,\cdot)=\mathsf{c}(\fhat_r(\cdot)),\) and the out-of-sample calibrated one \(G(\Prob_{n+1},\cdot)=\mathsf{c}(\fhat_r(\cdot))\).

    \begin{proof}[Proof of Theorem \ref{thm: consistency}]
        Without loss of generality, we assume that \(\fhat_r(X_{1})\leq \ldots \leq \fhat_r(X_{n})\).
        By Lemma \ref{lem: g_hat_dense} it follows that, for all \(x\in \X_{r}\vcentcolon =\{x^\prime\in \X:[\mu(x^\prime)\pm 5C_0\delta_r]\subset I\}\), the indices
        \begin{align*}
            & j(x)=\max\{j\in \{1,\ldots,n\}:T(\fhat_r(X_{j}))\leq T(\fhat_r(x))+4C_0\delta_r\}\\
            & i(x)=\min\{j\in \{1,\ldots,n\}:T(\fhat_r(X_{j}))\geq T(\fhat_r(x))\}
        \end{align*}
        are well defined with asymptotic probability one, because \([T(\fhat_r(x)),T(\fhat_r(x))+4C_0\delta_r]\) is of length \(4C_0\delta_r\) and is contained in \(I\):
        indeed, on the event \(B_r\) it holds that
        \begin{equation*}
            \left|T(\fhat_r(x))-\mu(x)\right|\leq C_0\delta_r,
        \end{equation*}
        so
        \begin{align*}
            \mu(x)-C_0\delta_r\leq T(\fhat_r(x))
            & \leq T(\fhat_r(x))+4C_0\delta_r\\
            & \leq \mu(x)+C_0\delta_r+4C_0\delta_r\\
            & = \mu(x)+5C_0\delta_r.
        \end{align*}
        This shows that \(T(\fhat_r(x)),T(\fhat_r(x))+4C_0\delta_r\in I\), given the definition of \(\X_r\).
        Lemma \ref{lem: g_hat_dense} shows that the points \(T(\fhat_r(X_{j}))\) are asymptotically dense in \(I\), so the above inequality yields that \(i(x),j(x)\) are the minimum and the maximum of asymptotically nonempty sets.
        Therefore, they are well-defined and it holds that \(i(x)\leq j(x)\) and
        \begin{equation*}
            T(\fhat_r(x))\leq T(\fhat_r(X_{i(x)}))\leq T(\fhat_r(X_{j(x)}))\leq T(\fhat_r(x))+4C_0\delta_r.
        \end{equation*}
        Lemma \ref{lem: g_hat_dense} yields that, with asymptotic probability one,
        \begin{equation*}
            \widehat{w}_{i(x)j(x)}=\# \left\{j\in \{1,\ldots,n\}:T(\fhat_r(x))\leq T(\fhat_r(X_j))\leq T(\fhat_r(x))+4C_0\delta_r\right\}\geq 4C_0Dn\delta_r
        \end{equation*}
        for some absolute constant \(0<D<C_1/2\).
        Therefore, it holds with asymptotic probability 1 (with respect to \(n\)) that 
        \begingroup
        \allowdisplaybreaks
        \begin{align*}
        \mathsf{c}(\fhat_r(x))-\mu(x)
        & \leq \mathsf{c}\left(\fhat_r(X_{i(x)})\right)-\mu(x)\\
        & \leq \min_{v\geq j(x)}\max_{u\leq i(x)} F_{uv}-\mu(x)\\
        & \leq \max_{u\leq i(x)} F_{uj(x)}-\mu(x)\\
        & \leq w_{i(x)j(x)}^{-1/2}M_n+\max_{u\leq i(x)} \overline{F}_{uj(x)}-\mu(x)\\
        & \leq (4C_0Dn\delta_r)^{-1/2}M_n+\max_{u\leq i(x)} \left(\overline{F}_{uj(x)}-\widehat{F}_{uj(x)}+\widehat{F}_{uj(x)}\right)-\mu(x)\\
        &\leq (4C_0Dn\delta_r)^{-1/2}M_n+\sup_{x\in \X}\left|T(\fhat_r(x))-\mu(x)\right|+\max_{u\leq i(x)} \widehat{F}_{uj(x)}-\mu(x)\\
        & \leq (4C_0Dn\delta_r)^{-1/2}M_n+\sup_{x\in \X}\left|T(\fhat_r(x))-\mu(x)\right|+T\left(\fhat_r\left(X_{j(x)}\right)\right)-\mu(x)\\
        & \leq (4C_0Dn\delta_r)^{-1/2}M_n+\sup_{x\in \X}\left|T(\fhat_r(x))-\mu(x)\right|+T(\fhat_r(x))-\mu(x)+4C_0\delta_r\\
        & \leq (4C_0Dn\delta_r)^{-1/2}M_n+2\sup_{x\in \X}\left|T(\fhat_r(x))-\mu(x)\right|+4C_0\delta_r,
        \end{align*}
        where we used the following:
        \begin{itemize}
            \item First step: monotonicity of \(\mathsf{c}(\fhat_r(\cdot))\) with respect to \(\fhat_r(\cdot)\), along with \(\fhat_r(x)\leq \fhat_r(X_{i(x)})\).
            \item Second step: min-max formula for monotone regression, since \(\mathsf{c}(\cdot)\) is defined as the isotonic regression fit to the points \((\fhat_r(X_1),Y_1),\ldots,(\fhat_r(X_n),Y_n)\).
            \item Fourth step: for all \(s\leq i(x)\), it holds that \(w_{sj(x)}=j(x)-s+1\geq j(x)-i(x)+1 = w_{i(x)j(x)}\).
            Also, \(w_{i(x)j(x)}^{-1/2}\leq (4C_0Dn\delta_r)^{-1/2}\).
            \item Sixth step: we used the convexity of the max function and the trivial inequality
            \begin{equation*}
                \left|\overline{F}_{sj(x)}-\widehat{F}_{sj(x)}\right|=\left|\frac{1}{w_{sj(x)}}\sum_{i=s}^{j(x)} \mu(X_i)-T\left(\fhat_r(X_i)\right)\right|\leq \sup_{x\in \X}\left|T(\fhat_r(x))-\mu(x)\right|
            \end{equation*}
            \item Third last step: for all \(s\leq i(x)\) it holds that \(\widehat{F}_{sj(x)}\leq T(\fhat_r(X_{j(x)}))\).
            \item Second last step: definition of \(j(x)\).
        \end{itemize}
        \endgroup
        An analogous argument provides an identical upper bound for \(\mu(x)-\mathsf{c}(\fhat_r(x))\).
        By Assumption \ref{assum: consistency_fhat} and Lemma \ref{lem: M_n_pi_bound}, it follows that the event
        \begin{equation*}
            \left\{M_n\leq s(\log n)^{1/2}\right\}\cap \left\{\sup_{x\in \X}\left|T(\fhat_r(x))-\mu(x)\right|<C_0\delta_r\right\}
        \end{equation*}
        has asymptotic probability 1 for \(s>\sigma\sqrt{2/c}\).
        On this event, it holds that
        \begin{align*}
        \sup_{x\in \X_r}\left|\mathsf{c}(\fhat_r(x))-\mu(x)\right|
        &\leq s(4C_0Dn\delta_r)^{-1/2}(\log n)^{1/2}+6C_0\delta_r\\
        & \lesssim \left(\frac{r}{\log r}\right)^{1/6}\left(\frac{\log n}{n}\right)^{1/2} + \left(\frac{\log r}{r}\right)^{1/3}.
        \end{align*}
        Since \(r\leq n\), this term is bounded by \(2(\log r/r)^{1/3}\), which finishes the proof.
    \end{proof}

    \section{Proof of Theorem \ref{thm: width_isotonic}} \label{sec: app_proof_width}



    The proof of Theorem \ref{thm: width_isotonic} relies on Lemma \ref{lem: change_one_coordinate}, which explains how the isotonic projection of a vector changes when we modify only one of its coordinates.
    The isotonic projection \(\iso(y)\) of a vector \(y\in \R^{n+1}\) is defined as the projection of \(y\) onto the isotonic cone \(\C\vcentcolon = \{x\in \R^{n+1}:x_1\leq \ldots \leq x_{n+1}\}\), that is
    \begin{equation*}
        \iso(y)=\argmin_{x\in \R^{n+1}, x_1\leq \ldots\leq x_{n+1}}\norm{y-x}_2,
    \end{equation*}
    where \(\norm{\cdot}_2\) is the standard Euclidean norm.
    It can be computed by PAVA \citep[Appendix~A.1]{wuetr_ziegel_23} and it relies on a partition of \(\{1,\ldots,n+1\}\) that consists of several blocks \(I_1^y,\ldots,I_{\ell}^y\subseteq \{1,\ldots,n+1\}\) in which its value is constant and equal to
    \begin{equation}
        \label{eq: isotonic_projection}
        \iso(y)_k=\frac{1}{\# I_j^y}\sum_{\ell\in I_j^y} y_\ell\quad \text{for all } k\in I_j^y.
    \end{equation}
    The isotonic projection is essentially the isotonic regression fit, with the order not determined by a covariate but by the natural order on the index space \(\{1,\ldots,n+1\}\).
    
    Lemma \ref{lem: change_one_coordinate} is related to the construction of the abridged version of PAVA by \citet{Henzi_2022}.
    It is relevant for Algorithm \ref{alg: isotonic_recalibration} because the central idea of this algorithm is to fit isotonic regression to two different vectors that differ only by one coordinate.

    \begin{lem}
        \label{lem: change_one_coordinate}
        Let \(x,y\in \R^{n+1}\) be two vectors and let \(i_0\in \{1,\ldots,n+1\}\) be an index such that \(x_i=y_i\) for all \(i\neq i_0\), and \(x_{i_0}\geq y_{i_0}\).
        Let \(I_1^y,\ldots,I_{\ell}^y\) and \(I_1^x,\ldots,I_s^x\) be the partitions of \(\{1,\ldots,n+1\}\) induced by \(\iso(y),\iso(x)\) respectively.
        Suppose that \(i_0\in I_{\ell_0}^y\cap I_{s_0}^x\).
        Then, it holds that
        \begin{equation}
            \label{eq: intervals_shift}
            \min I_{\ell_0}^y\leq \min I_{s_0}^x\quad \text{and}\quad \max I_{\ell_0}^y\leq \max I_{s_0}^x.
        \end{equation}
    \end{lem}

    \begin{proof}
        Since \(x_i\geq y_i\) componentwise, it follows that \(\iso(x)_i\geq\iso(y)_i\) for all \(i\in \{1,\ldots,n+1\}\).
        As shown in \citet{Henzi_2022}, the partition induced by \(\iso(y)\) will be a coarsening of the partition with the following blocks:
        \begin{equation*}
            I_t^x \text{ for } 1\leq t<s_0,\quad \{\min I_{s_0}^x,\ldots, i_0\},\quad \{j\} \text{ for } i_0< j\leq \max I_{s_0}^x,\quad I_t^x \text{ for } s_0<t\leq s.
        \end{equation*}
        Since \(i_0\) and \(\min I_{s_0}^x\) definitely belong to the same block in \(\iso(y)\), it follows directly that \(\min I_{\ell_0}^y\leq \min I_{s_0}^x\).
        
        Moreover, it is shown in \citet{Henzi_2022} that \(\iso(y)_j=\iso(x)_j\) for \(j>\max I_{s_0}^x\).
        But, for all \(j>\max I_{s_0}^x\), it holds that \(\iso(x)_j>\iso(x)_{i_0}\geq \iso(y)_{i_0}\); therefore \(\iso(y)_j>\iso(y)_{i_0}\).
        This shows that \(I_{\ell_0}^y\) cannot extend beyond \(\max I_{s_0}^x\), which finishes the proof.
    \end{proof}

    The next lemma was proved by \citet{Allen_2025}.
    It provides an upper bound for the distance between the isotonic projections of vectors that differ only in one coordinate.

    \begin{lem}
        \label{lem: isox_isoy}
        Let \(x,y\in \R^{n+1}\) be two vectors and let \(i_0\in \{1,\ldots,n+1\}\) be an index such that \(x_i=y_i\) for all \(i\neq i_0\), and \(x_{i_0}\geq y_{i_0}\).
        Let \(I_1^y,\ldots,I_{\ell}^y\) and \(I_1^x,\ldots,I_s^x\) be the partitions of \(\{1,\ldots,n+1\}\) induced by \(\iso(y),\iso(x)\) respectively.
        Suppose that \(i_0\in I_{\ell_0}^y\cap I_{s_0}^x\).
        Then,
        \begin{equation}
            \label{eq: choose_k_m}
            0\leq \iso(x)_{i_0}-\iso(y)_{i_0}
            \leq  \frac{|x_{i_0}-y_{i_0}|}{\max I_{\ell_0}^y-\min I_{s_0}^x+1}.
        \end{equation}
    \end{lem}

    We can now move on to the proof of Theorem \ref{thm: width_isotonic}. This is similar to the proof of \citet[Theorem~3.1]{Allen_2025}.
    One of the main differences is that we use Assumption \ref{assum: jumps} instead of the related result of \citet{Dimitriadis_2023}, which is only applicable to binary labels.
    Throughout the proof, we use the notation \(A_n\lesssim B_n\) to denote the existence of a universal constant \(C>0\) such that \(A_n\leq C B_n\) for all values of \(n\).

    \begin{proof}[Proof of Theorem \ref{thm: width_isotonic}]
        In Algorithm \ref{alg: isotonic_recalibration},
        \begin{equation*}
            \left(T_{\ell}(\fhat(X_1)),\ldots,T_{\ell}(\fhat(X_{n+1}))\right)\quad \text{and}\quad \left(T_{h}(\fhat(X_1)),\ldots,T_{h}(\fhat(X_{n+1}))\right)
        \end{equation*}
        are the isotonic projections of the (rearranged) vectors \((Y_1,\ldots,Y_n,\ell_{n+1})\) and \((Y_1,\ldots,Y_n,h_{n+1})\), where the order in the associated isotonic cone is not determined by the indices \(\{1,\ldots,n+1\}\), but by the values \(\fhat(X_1),\ldots,\fhat(X_{n+1})\).
        The calibrated prediction
        \begin{equation*}
            \left(\mathsf{c}(\fhat(X_1)),\ldots,\mathsf{c}(\fhat(X_{n+1}))\right)
        \end{equation*}
        is the isotonic projection of the vector \((Y_1,\ldots,Y_{n+1})\) onto the same isotonic cone.
        
        Let \(\pi\in S_{n+1}\) be a permutation such that \(\fhat(X_{\pi(1)})\leq \ldots \leq \fhat(X_{\pi(n+1)})\) and set \(i_0=\pi^{-1}(n+1)\).
        Then, the expected width of the interval \(\C_{n+1,\alpha}\) is equal to
        \begin{equation*}
            \Exp[w(\C_{n+1,\alpha})]=\Exp\Big[T_h(\fhat(X_{n+1}))-T_{\ell}(\fhat(X_{n+1}))\Big]=\Exp\Big[\iso(\bm{Y}^h)_{i_0}-\iso(\bm{Y}^{\ell})_{i_0}\Big],
        \end{equation*}
        where
        \begin{equation*}
            \bm{Y}^{\square}=\left(Y_{\pi(1)},\ldots,Y_{\pi(i_0-1)},\square_{n+1},Y_{\pi(i_0+1)},\ldots,Y_{\pi(n+1)}\right),\quad \square\in \{\ell,h\}.
        \end{equation*}
        We also define the vector \(\bm{Y}=\left(Y_{\pi(1)},\ldots,Y_{\pi(n+1)}\right)\).
        Let \(\{I_1^{\bm{Y}},\ldots,I_{\ell}^{\bm{Y}}\},\, \{I_1^{\bm{Y}^\ell},\ldots,I_m^{\bm{Y}^\ell}\}\), and \(\{I_1^{\bm{Y}^h},\ldots,I_{s}^{\bm{Y}^h}\}\) be the partitions induced by \(\iso(\bm{Y}),\iso(\bm{Y}^\ell)\), and \(\iso(\bm{Y}^h)\) respectively.
        Assumption \ref{assum: jumps} yields \(\Exp[J_{n+1}]\lesssim n^{1/3}\log^q n\) for large enough values of \(n\).
        Conditionally on \(\S_1\) and the empirical distribution \(\Qrob_{n+1}\)
        of the points \((\fhat(X_i),Y_i),\, i=1,\ldots,n+1\), the probability that \(i_0\) falls into any index interval \(I\subseteq \{1,\ldots,n+1\}\) is equal to \(|I|/(n+1)\).
        Suppose that \(i_0\in I_{\ell_0}^{\bm{Y}}\cap I_{m_0}^{\bm{Y}^{\ell}}\cap I_{s_0}^{\bm{Y}^h}\) for some indices \(\ell_0\in \{1,\ldots,\ell\},\, m_0\in \{1,\ldots,m\}\), and \(s_0\in \{1,\ldots,s\}\).
        We consider two different cases:
        \paragraph{Case I: \(Y_{n+1}\geq h_{n+1}\).}
        In that case, we split each interval \(I_{j}^{\bm{Y}}\) into two consecutive intervals \(I_j^\prime\) and \(I_j^{\prime\prime}\), such that \(I_j^{\prime}\) contains the \(\lceil |I_j^{\bm{Y}}|\cdot n^{-\delta}\rceil\) smallest elements, where \(\delta\in (0,1/3)\).
        If \(i_0\) falls into \(I_{\ell_0}^{\prime\prime}\), then \(I_{\ell_0}^{\prime}\subseteq I_{\ell_0}^{\bm{Y}}\cap I_{s_0}^{\bm{Y}^h}\).
        Indeed, by construction it holds that \(I_{\ell_0}^{\prime}\subseteq I_{\ell_0}^{\bm{Y}}\) and, since \(i_0\in I_{\ell_0}^{\prime\prime}\), it follows that \(\max I_{\ell_0}^{\prime}< i_0 \leq \max I_{s_0}^{\bm{Y}^h}\).
        From Lemma \ref{lem: change_one_coordinate}, it follows that \(\min I_{s_0}^{\bm{Y}^h}\leq \min I_{\ell_0}^{\bm{Y}}=\min I_{\ell_0}^{\prime}\), so \(I_{\ell_0}^{\prime}\) is also a subset of \(I_{s_0}^{\bm{Y}^h}\).
        The inclusion \(I_{\ell_0}^{\prime}\subseteq I_{\ell_0}^{\bm{Y}}\cap I_{s_0}^{\bm{Y}^h}\) yields
        \begin{equation*}
            \max I_{s_0}^{\bm{Y}^h}-\min I_{\ell_0}^{\bm{Y}}+1\geq \lceil |I_{\ell_0}^{\bm{Y}}|\cdot n^{-\delta}\rceil.
        \end{equation*}
        Therefore, \eqref{eq: choose_k_m} implies that
        \begingroup
        \allowdisplaybreaks
        \begin{align}
            \label{eq: lower_bound_iso}
            0
            & \leq \Exp\left[\one\{Y_{n+1}\geq h_{n+1}\}\left|\iso(\bm{Y})_{i_0}-\iso(\bm{Y}^h)_{i_0}\right|\right]\nonumber \\
            & \leq \Exp\left[\one\{Y_{n+1}\geq h_{n+1}\}\frac{|h_{n+1}-Y_{n+1}|}{\max I_{s_0}^{\bm{Y}^h}-\min I_{\ell_0}^{\bm{Y}}+1}\right]\nonumber \\ 
            & \leq \Exp\left[n^{\delta}|Y_{n+1}-h_{n+1}|\cdot \sum_{j=1}^{J_{n+1}} \frac{\one\{i_0\in I_j^{\prime\prime}\}}{|I_j^{\bm{Y}}|}\right] + \Exp\left[|h_{n+1}-Y_{n+1}|\cdot \sum_{j=1}^{J_{n+1}} \one\{i_0\in I_j^{\prime}\}\right],
        \end{align}
        \endgroup
        where we used the fact that \(\max I_{s_0}^{\bm{Y}^h}-\min I_{\ell_0}^{\bm{Y}}+1\geq \lceil |I_j^{\bm{Y}}|\cdot n^{-\delta}\rceil\) in the last step.
        We now treat the two terms separately.
        The first term can be bounded by Cauchy-Schwarz and the fact that \(\Prob(i_0\in I_j^{\bm{Y}}|\Qrob_{n+1})=|I_j^{\bm{Y}}|/(n+1)\).
        Indeed,
        \begingroup
        \allowdisplaybreaks
        \begin{align*}
            \Exp\left[n^{\delta}|Y_{n+1}-h_{n+1}|\cdot \sum_{j=1}^{J_{n+1}} \frac{\one\{i_0\in I_j^{\prime\prime}\}}{\left|I_j^{\bm{Y}}\right|}\right]
            & \leq n^\delta \left(\Exp\left[\sum_{j=1}^{J_{n+1}}\frac{\one\{i_0\in I_j^{\prime\prime}\}}{\left|I_j^{\bm{Y}}\right|^2}\right]\right)^{1/2}\left(\Exp\left[(Y_{n+1}-h_{n+1})^2\right]\right)^{1/2}\\
            &\leq  O\left(n^{\delta} \left(\Exp\left[\Exp\left[\sum_{j=1}^{J_{n+1}}\frac{\one\{i_0\in I_j^{\prime\prime}\}}{\left|I_j^{\bm{Y}}\right|^2}\, \middle| \, \S_1, \Qrob_{n+1}\right]\right]\right)^{1/2}\right)\\
            & \leq O\left(n^{\delta} \left(\Exp\left[\sum_{j=1}^{J_{n+1}} \frac{1}{(n+1)\left|I_j^{\bm{Y}}\right|}\right]\right)^{1/2}\right)\\
            & \leq O\left(n^{\delta} \left(\Exp\left[J_{n+1} n^{-1}\right]\right)^{1/2}\right)\\
            & = O\left( n^{\delta-1/3}\log^{q/2}n\right),
        \end{align*}
        \endgroup
        where we used Assumption \ref{assum: jumps} in the second-to-last step.

        The second term in \eqref{eq: lower_bound_iso} is treated in a similar way.
        We have
        \begingroup
        \allowdisplaybreaks
        \begin{align*}
            \Exp\left[|h_{n+1}-Y_{n+1}|\cdot \sum_{j=1}^{J_{n+1}} \one\{i_0\in I_j^{\prime}\}\right]
            & \leq \left(\Exp(Y_{n+1}-h_{n+1})^2\right)^{1/2}\left(\Exp\left[\sum_{j=1}^{J_{n+1}} \one\{i_0\in I_j^\prime\}\right]\right)^{1/2}\\
            & \leq O\left(\Exp\left[\Exp\left[\sum_{j=1}^{J_{n+1}} \one\{i_0\in I_j^\prime\}\, \middle| \, \Qrob_{n+1}\right]\right]^{1/2}\right)\\
            & = O\left(\Exp\left[\sum_{j=1}^{J_{n+1}} \frac{\left|I_j^\prime\right|}{n+1}\right]^{1/2}\right)\\
            & \leq O\left(\Exp\left[\sum_{j=1}^{J_{n+1}} \frac{\left|I_j^{\bm{Y}}\right|n^{-\delta}+1}{n+1}\right]^{1/2}\right)\\
            & = O\left(\Exp\left[\frac{n^{-\delta}(n+1)+J_{n+1}}{n+1}\right]^{1/2}\right)\\
            & = O\left(n^{-\delta/2}+n^{-1/3}\log^{q/2} n\right),
        \end{align*}
        \endgroup
        which also converges to zero.
        Along with \eqref{eq: lower_bound_iso}, this shows that
        \begin{equation}
            \label{eq: case1_upper}
            \Exp\Big[\one\{Y_{n+1}\geq h_{n+1}\}\left|\iso(\bm{Y})_{i_0}-\iso(\bm{Y}^h)_{i_0}\right|\Big]\overset{n\to \infty}{\longrightarrow} 0.
        \end{equation}
        The whole argument can be used with \(\ell_{n+1}\) instead of \(h_{n+1}\).
        This yields
        \begin{equation}
            \label{eq: case1_lower}
            \Exp\Big[\one\{Y_{n+1}\geq \ell_{n+1}\}\left|\iso(\bm{Y})_{i_0}-\iso(\bm{Y}^\ell)_{i_0}\right|\Big]\overset{n\to \infty}{\longrightarrow} 0.
        \end{equation}
        \paragraph{Case II: \(Y_{n+1}\leq \ell_{n+1}\).}
        This case can be reduced to the first one by considering the random variables \(-Y_i\) as the response variables.
        In that case, the prediction set is given by \(I_\alpha(X_{n+1})=[-h_{n+1},-\ell_{n+1}]\) and it holds that \(\iso(\bm{-Y}_{\text{inv}})=-(\iso(\bm{Y}))_{\text{inv}}\), where \(\bm{Y}_{\text{inv}}=(Y_{\pi(n+1)},\ldots,Y_{\pi(1)})\).
        Similar relationships hold for \(\bm{Y}^\ell,\bm{Y}^h\).
        Therefore, the proof can use the same argument as in the previous case.
        We conclude that
        \begin{equation}
            \label{eq: case2_lower}
            \Exp\left[\one\{Y_{n+1}\leq \ell_{n+1}\}\left|\iso(\bm{Y})_{i_0}-\iso(\bm{Y}^\ell)_{i_0}\right|\right]\overset{n\to \infty}{\longrightarrow} 0.
        \end{equation}
        The argument used in Case I can also be reversed to cover the event \(\{Y_{n+1}\leq h_{n+1}\}\).
        This yields
        \begin{equation}
            \label{eq: case2_upper}
            \Exp\left[\one\{Y_{n+1}\leq h_{n+1}\}\left|\iso(\bm{Y})_{i_0}-\iso(\bm{Y}^h)_{i_0}\right|\right]\overset{n\to \infty}{\longrightarrow} 0.
        \end{equation}
        From \eqref{eq: case1_upper}, \eqref{eq: case1_lower}, \eqref{eq: case2_lower}, \eqref{eq: case2_upper}, it follows that
        \begin{align*}
            \Exp[w(\C_{n+1,\alpha})]
            &=\Exp\Big[\iso(\bm{Y}^h)_{i_0}-\iso(\bm{Y}^\ell)_{i_0}\Big]\\
            & = \Exp\Big[\iso(\bm{Y}^h)_{i_0}-\iso(\bm{Y})_{i_0}+\iso(\bm{Y})_{i_0}-\iso(\bm{Y}^\ell)_{i_0}\Big]\\
            & \leq \Exp\Big[\left|\iso(\bm{Y}^h)_{i_0}-\iso(\bm{Y})_{i_0}\right|\Big]+\Exp\Big[\left|\iso(\bm{Y})_{i_0}-\iso(\bm{Y}^\ell)_{i_0}\right|\Big]
            \overset{n\to \infty}{\rightarrow} 0,
        \end{align*}
        which finishes the proof.
    \end{proof}

    \section{Construction of the conformal prediction set} \label{sec: app_CP}
    
    Theorem \ref{thm:I} does not enforce any restrictions on the choice of the prediction set \(I_\alpha(X_{n+1})\).
    As we have discussed in Subsection \ref{subsec: calibrated_intervals}, it is often practical to use conformal prediction sets, which are valid under exchangeability.
    In our simulations in Section \ref{sec: simulation}, we derived \(I_\alpha(X_{n+1})\) using conformalized quantile regression (CQR) \citep{Romano_2019}.
    CQR is a conformal-prediction-based method that is adaptive to the test point \(X_{n+1}\).
    Adaptivity here translates to higher flexibility with respect to the size of the conformal prediction set.
    CQR obtains its adaptive properties by leveraging estimates of the upper and lower quantiles of \(Y\, | \, X\) and by taking the number of data points that are available from each part of the \(X\)-domain into account.
    Split CQR consists of three basic steps:
    \begin{itemize}
        \item The dataset is split into a training set \((X_{-r},Y_{-r}),(X_{-r+1},Y_{-r+1}),\ldots,(X_0,Y_0)\) and a calibration set \((X_1,Y_1),\ldots,(X_n,Y_n)\).
        \item A quantile prediction model \(\widehat{q}_{\gamma}\) is fitted on the training set, where \(\gamma\in (0,1)\).
        This model can be linear quantile regression, a quantile regression forest \citep{Meinshausen_2006}, a quantile neural network \citep{Taylor_2000}, or any other model that predicts the \(\gamma\)-quantile of the conditional distribution \(\L(Y|X=x)\).
        In our simulation, we use quantile regression forests, which perform well in terms of prediction set size and conservativeness \citep[Section~6]{Romano_2019}.
        \item We then compute the scores
        \begin{equation*}
            s_i=s(X_i,Y_i)=\max\left\{\widehat{q}_{\alpha/2}(X_i) - Y_i,\ Y_i-\widehat{q}_{1-\alpha/2}(X_i)\right\},\quad i=1,\ldots,n.
        \end{equation*}
        Let \(\widehat{q}\) be the \(\left\lceil (1-\alpha)(n+1)\right\rceil\)-th smallest calibration score -- \(\widehat{q}=+\infty\) if \((1-\alpha)(n+1)>n\).
        For any test point \(X_{n+1}\), we output the prediction set
        \begin{equation*}
            \mathcal{C}(X_{n+1})=\Big[\widehat{q}_{\alpha/2}(X_{n+1})-\widehat{q}, \widehat{q}_{1-\alpha/2}(X_{n+1})+\widehat{q}\Big].
        \end{equation*}
    \end{itemize}
    The general idea of CQR is to penalize \emph{bad} predictions via the score \(s\) and then inflate the miscalibrated prediction set \(\left(\widehat{q}_{\alpha/2}(X_i), \widehat{q}_{1-\alpha/2}(X_i)\right)\) depending on the overall penalty.
    Under exchangeability, the CQR prediction sets enjoy coverage guarantees conditionally on the training set \citet[Theorem~1]{Romano_2019}.

    An interesting feature of CQR is that its validity does not depend on the coverage guarantees of the base model \(\widehat{q}_{\gamma}\).
    Even if this model is arbitrarily bad, the CQR prediction sets will still satisfy the marginal coverage condition \(\Prob(Y_{n+1}\in \mathcal{C}(X_{n+1}))\geq 1-\alpha\).
    Therefore, even if we choose a different value of \(\gamma\) for the base model \(\widehat{q}_\gamma\), CQR will continue to be valid.
    
    Inspired by this observation, and in order to improve power, \citet{Romano_2019} suggest that we tweak the parameter \(\gamma\) of the two fitted quantile prediction models.
    In their implementation, they choose \(\gamma= 0.15\) by default, instead of \(\gamma = 0.1\).
    However, they also mention that this approach can be optimized even further by treating \(\gamma\) as a hyperparameter, tuned by cross-validation.
    The cross-validation error is defined as the size of the prediction set \(\C(X_{n+1})\).
    Therefore, the chosen pair \(\gamma_1,\gamma_2\) -- possibly a different value for each quantile level -- is the one that leads on average to the smallest sets \(\C(X_{n+1})\).

    \section{Additional details for the simulation study} \label{sec: app_simulation}

    In this section, we provide additional details for the simulation study in Section \ref{sec: simulation}.
    We start with the neural network architecture, which consists of the following elements:
    \begin{enumerate}
        \item 2 hidden layers, with \(32\) neurons each.
        \item Swish activation function, which is an alternative to ReLU.
        This function is defined as
        \begin{equation*}
            \phi(x)=\frac{x}{1+\exp(-\beta x)},
        \end{equation*}
        where \(\beta\) is a trainable parameter.
        \item Batch normalization with momentum.
        \item Early stopping.
    \end{enumerate}
    To fit the network, we split the training dataset into a learning set and a validation set, with a ratio of \(9:1\).
    Figure \ref{fig: NN_general_plot} shows the calibrated confidence intervals produced by this base model, as well as the in-sample calibrated prediction.

    \begin{figure}[H]
        \centering
        \includegraphics[width=0.65\linewidth]{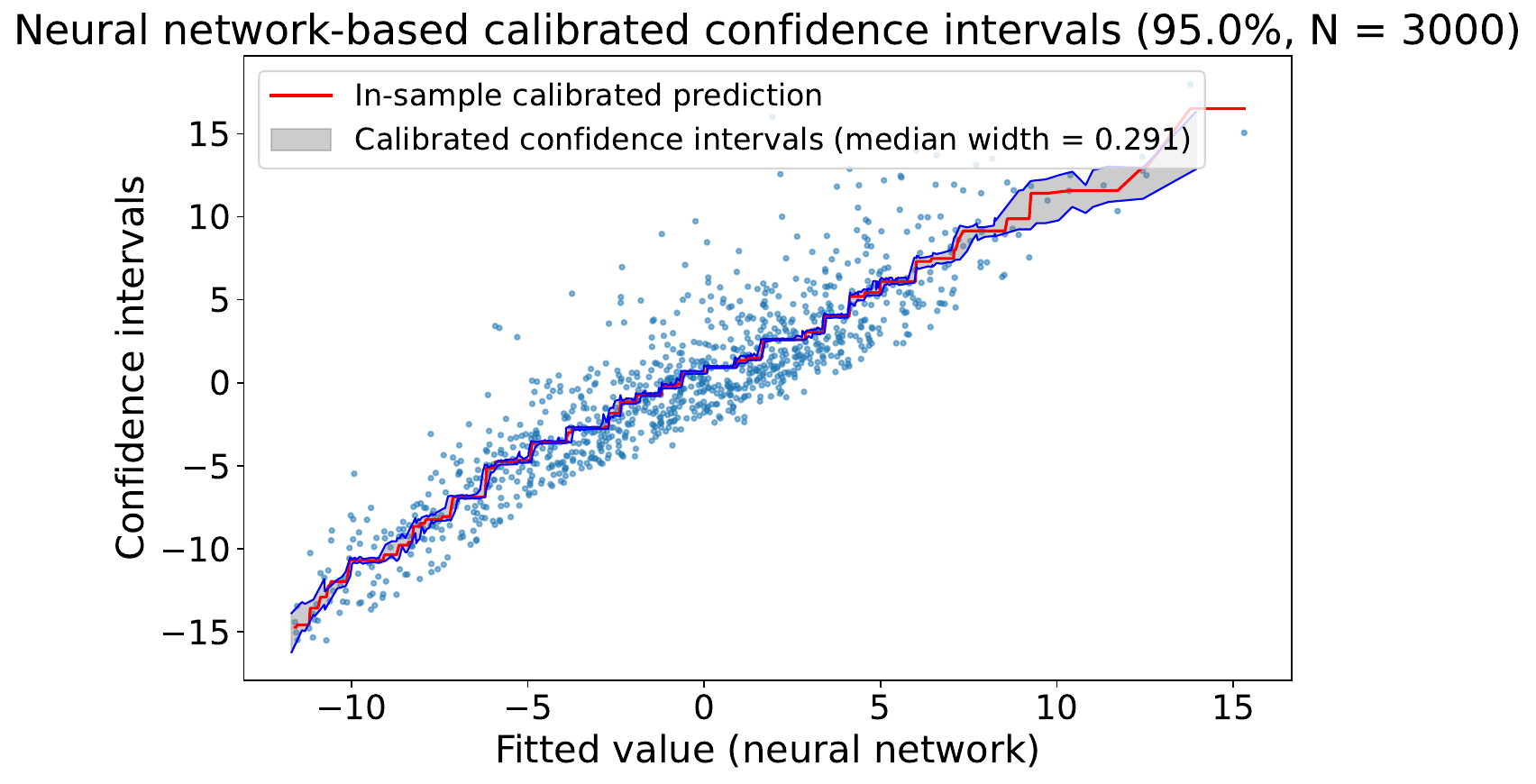}
        \caption{Calibrated confidence intervals based on a neural network \(\fhat\).
        The points \((\fhat(x_i),y_i)\), \(i=1,\ldots,n\) are also shown.
        The red curve is the isotonic regression fit on these points, which is in-sample calibrated.}
        \label{fig: NN_general_plot}
    \end{figure}

    The neural network appears to be adequately calibrated according to the calibration plots in Figure \ref{fig: calibration_plots_NN}.
    This figure includes two of the best candidates discussed in Subsection \ref{subsec: best_candidates}, which seem to be slightly less calibrated.
    The calibration plots of the neural network and the best candidates can be compared with that of the true conditional mean.
    
    \begin{figure}[ht]
    \centering
    \begin{minipage}{0.49\textwidth}
        \centering
        \includegraphics[width=\linewidth]{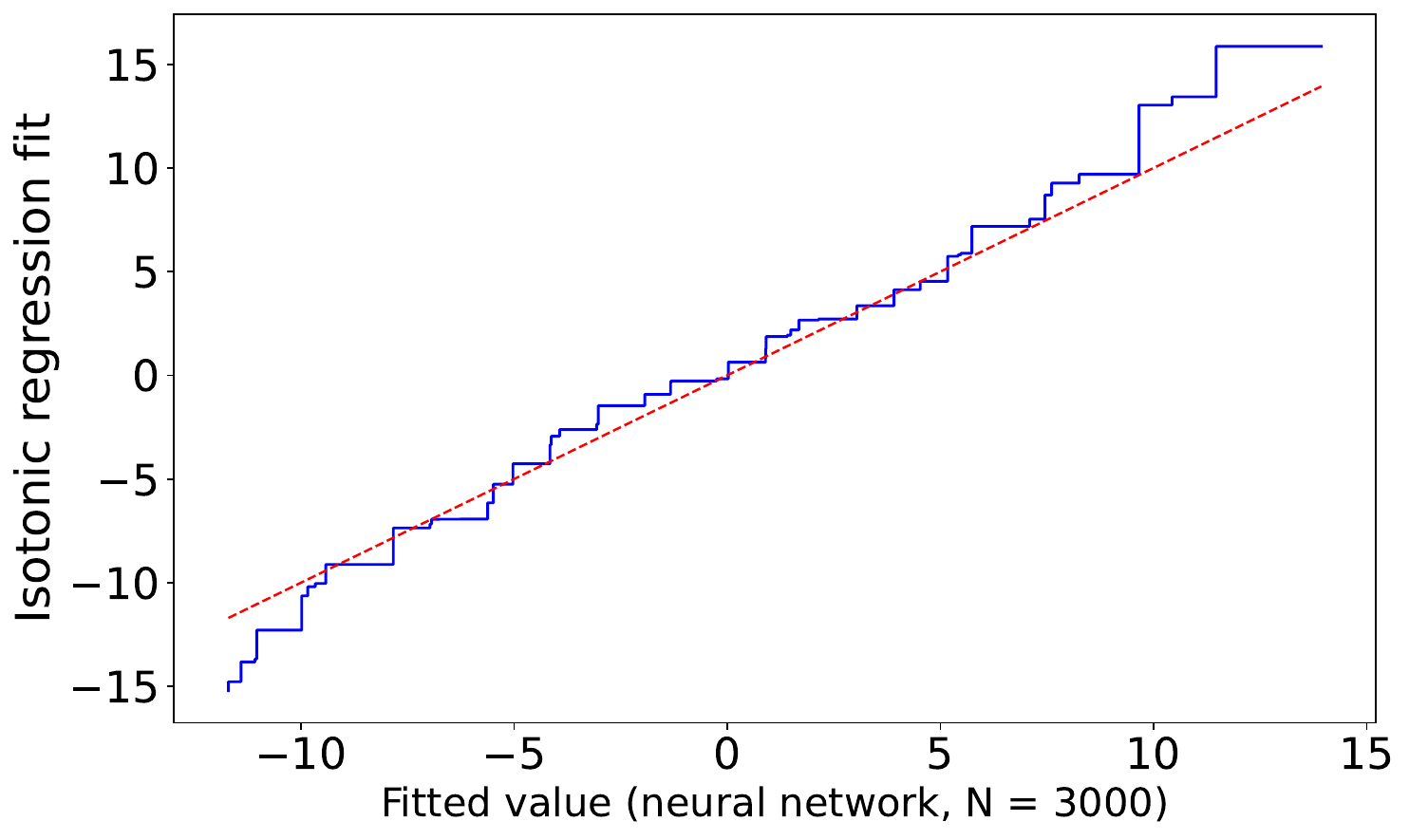}
    \end{minipage}
    \hfill
    \begin{minipage}{0.49\textwidth}
        \centering
        \includegraphics[width=\linewidth]{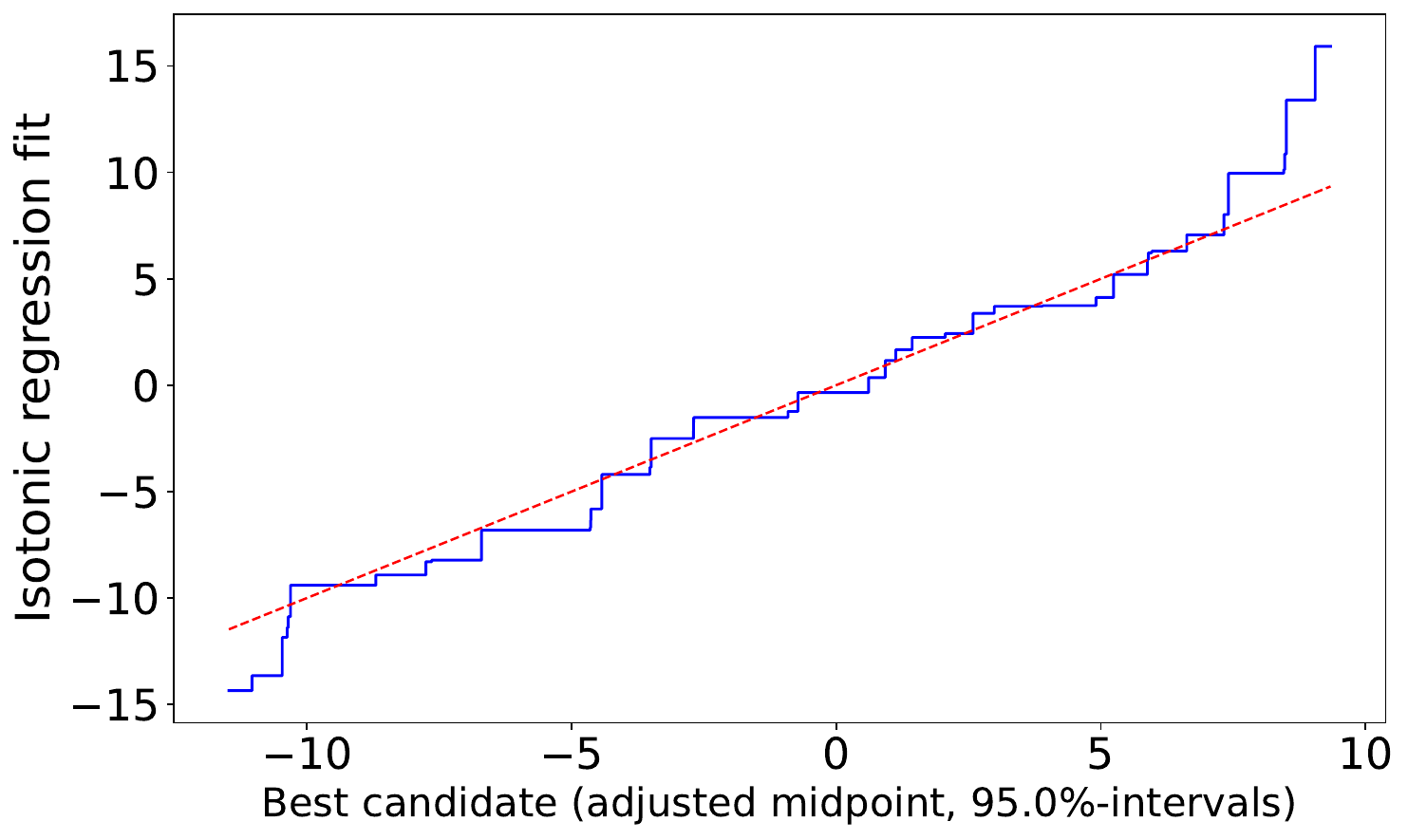}
    \end{minipage}

    \vspace{0.5cm}

    \begin{minipage}{0.49\textwidth}
        \centering
        \includegraphics[width=\linewidth]{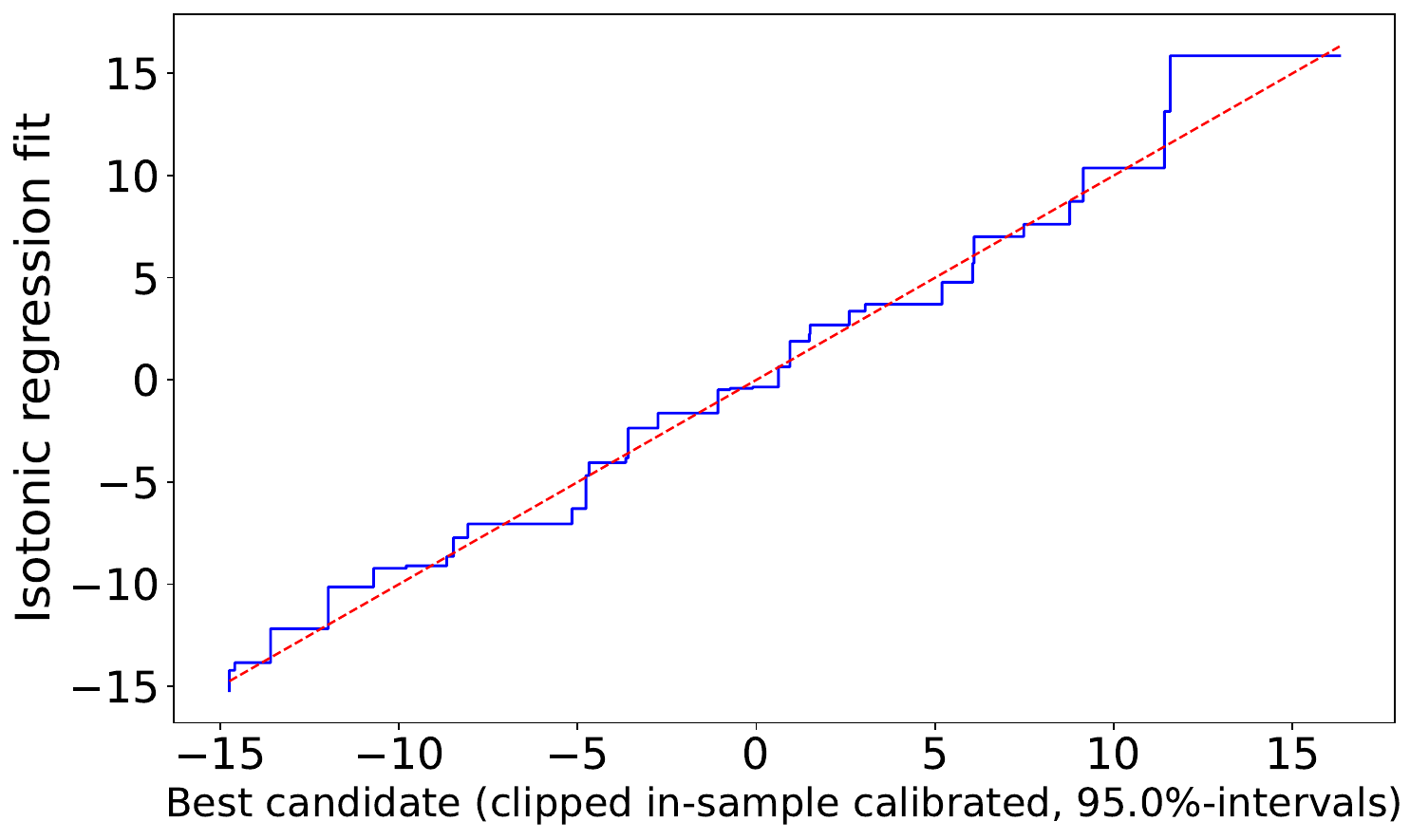}
    \end{minipage}
    \hfill
    \begin{minipage}{0.50\textwidth}
        \centering
        \includegraphics[width=\linewidth]{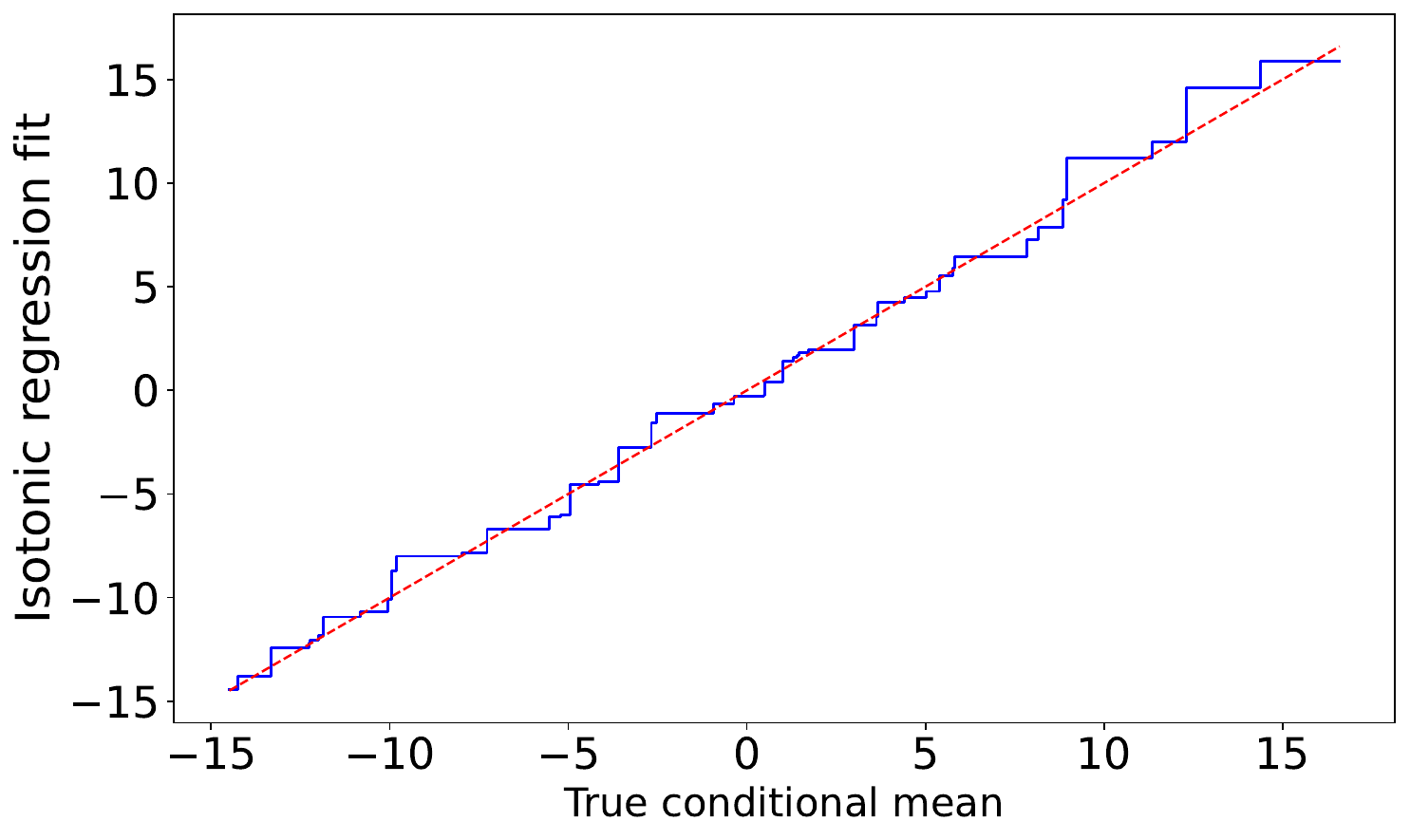}
    \end{minipage}
    \caption{Calibration plots for four different predictions.
    The neural network seems to be adequately calibrated.
    Two of the best candidates proposed in Subsection \ref{subsec: best_candidates} seem to be slightly miscalibrated.}
    \label{fig: calibration_plots_NN}
    \end{figure}

    \subsection{Motivation for calibration plots} \label{subsec: calibration_plots_motivation}
    
    In Subsection \ref{subsec: calibration_plots}, we presented the calibration plots corresponding to the predictions of the GAM and to the best-candidate predictions proposed in Subsection \ref{subsec: best_candidates}.
    The calibration plot has been established as a graphical diagnostic tool due to the following (heuristic) idea:
    let \(T\) denote the isotonic regression fit, which induces a partition \(B_1,\ldots,B_k\) of \(\fhat(\X)\subseteq \R\) into consecutive intervals and an associated partition \(D_1,\ldots,D_k\) of the set \(\{1,\ldots,m\}\) such that \(D_j=\{i\in \{1,\ldots,m\}:\fhat(X_i)\in B_j\}\).
    Let \(x\in \X\) be an arbitrary point such that \(\fhat(x)\in B_{j_0}\) for some index \(j_0\in \{1,\ldots,k\}\) and denote \(\fhat(x)\) by \(y\).
    Then, if \(T\) is close to the identity function, it follows that
    \begingroup
    \allowdisplaybreaks
    \begin{align*}
        y\approx T(y)= \frac{1}{\# D_{j_0}}\sum_{j: \widehat{f}(X_j)\in B_{j_0}} Y_j
        & = \Exp_{\Prob_n}[Y\, | \, \widehat{f}(X)\in B_{j_0}]\\
        &\approx \Exp[Y\, | \, \widehat{f}(X)\in B_{j_0}]\\
        & = \frac{1}{\Prob\left(\widehat{f}(X)\in B_{j_0}\right)}\int_{B_{j_0}} \Exp[Y\, | \, \widehat{f}(X)=u]\, d\Prob_{\widehat{f}(X)}(u)\\
        & \approx \frac{1}{\Prob\left(\widehat{f}(X)\in B_{j_0}\right)}\int_{B_{j_0}} \Exp[Y\, | \, \widehat{f}(X)=y]\, d\Prob_{\widehat{f}(X)}(u)\\
        & = \Exp[Y\, | \, \widehat{f}(X)=y],
    \end{align*}
    \endgroup
    so \(\fhat\) is approximately calibrated.
    This is only a heuristic argument that uses the Law of Large Numbers and the assumption that \(\widehat{f}(X)\) does not vary a lot over \(B_{j_0}\).
    Even so, this argument shows why calibration plots are meaningful as a diagnostic tool.

    \section{Additional details for the case study} \label{sec: app_case_study}

    In this section, we give more details about the case study in Section \ref{sec: case_study}.
   Table \ref{tab: freMTPL2freq} presents a summary of the dataset \texttt{freMTPL2freq}.
    
    \begin{table}[h]
        \centering
        \begin{tabular}{|>{\centering\arraybackslash}m{3cm}|m{13cm}|}
            \hline
            \textbf{Column} & \textbf{Interpretation} \\ \hline
            \texttt{IDpol} & Policy number.\\ \hline
            \texttt{ClaimNb} & Number of claims.\\ \hline
            \texttt{Exposure} & Exposure of a given policy. For example, an exposure equal to \(0.5\) denotes that the policy was active for half of the accounting year.\\ \hline
            \texttt{Area} & Area code (categorical, ordinal). \\ \hline
            \texttt{VehPower} & Power of the insured car (categorical, ordinal). \\ \hline
            \texttt{VehAge} & Age of the insured car (in years). \\ \hline
            \texttt{DrivAge} & Age of the driver (in years). \\ \hline
            \texttt{BonusMalus} & Variable that takes values between \(50\) and \(230\) and indicates the driver's past behaviour. Any accidents for which the insured driver is liable lead to a higher bonus-malus score. That score is subsequently used to determine the driver's insurance premium.\\ \hline
            \texttt{VehBrand} & Car brand (categorical, nominal).\\ \hline
            \texttt{VehGas} & Diesel or regular fuel (binary). \\ \hline
            \texttt{Density} & Number of inhabitants per \(\text{km}^2\) in driver's place of residence. \\ \hline
            \texttt{Region} & Administrative regions of France between 1982 and 2015 (categorical).\\ \hline
        \end{tabular}
        \caption{Summary of the dataset \texttt{freMTPL2freq}.}
        \label{tab: freMTPL2freq}
    \end{table}

    \subsection{Poisson GLM}\label{subsec:Poisson_GLM}

    The Poisson GLM maintains the universal model assumption that \(N_{-r},\ldots,N_n\) are independent random variables such that \(N_i\sim \text{Poi}(\lambda(x_i)v_i)\), \(i=-r,\ldots,n\).
    The intensity function \(\lambda:\X\to \mathbb{R}_+\) represents the expected annual claim frequency and is assumed to be a log-linear function, that is  \(\log \lambda(x)=\beta_0 + \sum_{\ell = 1}^d \beta_\ell x_\ell\).

    Using \(v_i\) as an offset is a sensible choice: If policies \(i\) and \(j\) have the same covariate values, but \(v_i=2v_j\), i.e. \(i\) has been active for twice as long as \(j\), then the expected number of claims of \(i\) should be twice that of \(j\).
    Since the expectation of a Poisson random variable \(\text{Poi}(\lambda)\) is equal to \(\lambda\), the choice of the exposure as a weight function agrees with our intuition.
    
    The coefficients \(\beta_0, \ldots, \beta_d\) are estimated by minimizing the deviance loss
    \begin{align}
        \label{eq: deviance_loss}
        \L(\beta_0,\ldots,\beta_d)
        & =\sum_{i=0}^r 2N_{-i}\left[\frac{\lambda(x_{-i})v_{-i}}{N_{-i}}-1-\log\left(\frac{\lambda(x_{-i})v_{-i}}{N_{-i}}\right)\right]\geq 0 \nonumber \\
        & = \sum_{i=0}^r \left[2
        \lambda(x_{-i})v_{-i}-2N_{-i}+2N_{-i}\log\left(\frac{N_{-i}}{\lambda(x_{-i})v_{-i}}\right)\right]\geq 0
    \end{align}
    on the training set \(\{(x_{-i},N_{-i},v_{-i})\}_{i=0}^r\), where \(N_{-i}\) are the observed claim counts and \(v_{-i}\) the corresponding exposure values.
    To compute the deviance loss in case \(N_{-i}=0\), we use the convention \(0\cdot \log(0)=0\).
    This minimization problem is equivalent to maximum-likelihood estimation.
    Indeed, the negative log-likelihood of the sample \(\{(x_{-i},N_{-i})\}_{i=0,\ldots,r}\) is given by
    \begingroup
    \allowdisplaybreaks
    \begin{align*}
    -\sum_{i=0}^r \log \ell(x_{-i},N_{-i})
    & = \sum_{i=0}^r -\log \left(\frac{e^{-\lambda(x_{-i})v_{-i}}\cdot (\lambda(x_{-i})v_{-i})^{N_{-i}}}{N_{-i} !}\right)\\
    & = \sum_{i=0}^r \left[\lambda(x_{-i})v_{-i} - N_{-i}\log\left(\lambda(x_{-i})v_{-i}\right)+\log(N_{-i} !)\right],
    \end{align*}
    \endgroup
    which has the same minimizer as
    \begin{equation*}
        \sum_{i=0}^r 2N_{-i} \left[\frac{\lambda(x_{-i})v_{-i}}{N_{-i}}-1-\log \left(\frac{\lambda(x_{-i})v_{-i}}{N_{-i}}\right)\right],
    \end{equation*}
    because \(N_{-i}\) is regarded as a constant.
    To summarize, the implementation of Algorithm \ref{alg: isotonic_recalibration} in this setting includes the following steps:
    \begin{enumerate}
        \item Estimate \(\beta_0,\ldots,\beta_d\) by minimizing the Poisson deviance on the training set \(\S_1\).
        \item Compute the squared Poisson deviance residuals
        \begin{equation}
            \label{eq: Pearson_scores}
            s\left(N_i, \widehat{\lambda}(x_i),v_i\right)=2N_i\left(\frac{\widehat{\lambda}(x_i)v_i}{N_i}-1-\log\left(\frac{\widehat{\lambda}(x_i)v_i}{N_i}\right)\right),\quad i=1,\ldots,n
        \end{equation}
        on the calibration set \(\C=\{(x_i,N_i,v_i)\}_{i=1}^{n}\) and their \(\lceil (1-\alpha)(n+1)\rceil/n\) quantile \(\widehat{q}_{1-\alpha}\).
        The deviance residual is a good choice of a nonconformity score because the model was fitted to the data using the Poisson deviance loss, which is simply the sum of these residuals.
        \item For any covariate test point \(x^\prime\), the set \(I_\alpha(x^\prime)=\left\{y\in \{0,1,2,3,4\}:s(y,\widehat{\lambda}(x^\prime),v^\prime)\leq \widehat{q}_{1-\alpha}\right\}\) is a \((1-\alpha)\)-conformal prediction set for the unobserved claim number \(N^\prime\).
        \item For all test points \(x^\prime\), consider the covariate vector
        \begin{equation*}
            \Big(\widehat{\lambda}(x_1)v_1,\ldots,\widehat{\lambda}(x_n)v_n, \widehat{\lambda}(x^\prime)v^\prime\Big)
        \end{equation*}
        and the response vectors \(\left(N_1,\ldots,N_n,\min I_\alpha(x^\prime)\right)\) and \(\left(N_1,\ldots,N_n,\max I_\alpha(x^\prime)\right)\).
        Run isotonic regression of each of these response vectors against the above covariate vector to obtain the isotonic fits \(T_\ell\) and \(T_h\).
        \item Output the prediction set \(\C_{\alpha}(x^\prime)=\left[T_\ell(\widehat{\lambda}(x^\prime)v^\prime),T_h(\widehat{\lambda}(x^\prime)v^\prime)\right]\).
        If \(I_\alpha(x^\prime)=\emptyset\), output \(\C_{\alpha}(x^\prime)=\emptyset\).
        According to Theorem \ref{thm:I}, this interval contains an expectation-calibrated prediction for \(N^\prime\) with probability at least \(1-\alpha\).
        This prediction is the isotonic regression fit of \((N_1,\ldots,N_n, N^\prime)\) against the same covariate vector as above.
    \end{enumerate}

    \subsection{Boosted regression tree}\label{subsec:Boosted_Regression_Tree}

    The second model we fit to the data is CatBoost \citep{Prokhorenkova_2018}.
    This is a gradient boosting approach that is particularly suitable for categorical features.
    As in \citet{Noll_2020}, we do not need any feature preprocessing for this model.

    Gradient boosting models combine several \emph{weak} base learning models \(f_1,\ldots,f_M\) to come up with an accurate ensemble model \(F_M\).
    The base models come from a pre-specified function class (e.g. trees) and they are added to the ensemble sequentially.
    The objective function at step \(t\) is given by \(\L^{(t)}=\sum_{i=0}^r \ell(y_{i}, F_{t-1}(x_{i})+f_t(x_{i}))\), and it is minimized over \(f_t\) by fitting a base model to the negative gradients
    \begin{equation*}
        g_{it}=-\left.\frac{\partial \ell(y_{i},s)}{\partial s}\right|_{s = F_{t-1}(x_{i})}
    \end{equation*}
    using least squares regression.
    CatBoost uses trees as base learners and minimizes the weighted Poisson negative log-likelihood \(\ell(N_i, \lambda(x_i),v_i) =\lambda(x_i)v_i-N_i\log(\lambda(x_i))\).
    This is the same objective function as that of a Poisson GLM.
    The difference is that, in Poisson GLMs, \(\log(\lambda(x_i))\) is modeled as a linear function of \(x_i\), whereas in CatBoost the log-frequency is modeled as a tree ensemble.
    More details about gradient boosting can be found in \citet[Chapter~10]{Hastie_2009}.

    \subsection{Neural network}\label{subsec:FFNN}

    We use the same network configuration as \citet[Section~7.3.2]{Wuthrich_2023}.
    Like before, \texttt{Exposure} and \texttt{ClaimNb} are clipped at levels \(1\) and \(4\) respectively.
    Feature pre-processing differs from that used for Poisson GLMs.

    \begin{itemize}
        \item \texttt{VehBrand} and \texttt{Region} are encoded as categorical features using one-hot encoding.
        They have 11 and 22 classes respectively.
        \item \texttt{VehGas} is encoded as a binary numerical variable.
        \item \texttt{VehAge}, \texttt{DrivAge}, and \texttt{BonusMalus} are clipped at \(20, 90, 150\) respectively.
        \item \texttt{Area} is encoded as a numerical random variable with values in \(\{0,1,\ldots,6\}\).
    \end{itemize}

    A \texttt{StandardScaler} is applied to all numerical features before inputting them to the network. 
    We use the same transformation as in \citet[Equation~7.29]{Wuthrich_2023}.
    We also use a batch size of 5\,000 and the Nadam optimizer.
    We split the dataset into the following subsets:
    \begin{enumerate}
        \item A learning dataset \(\S_1\) containing \(60\%\) of the data points.
        During training, this is further split into training and validation sets with a ratio $9:1$.
        \item A calibration set \(\S_2\) containing \(20\%\) of the data points.
        \item A test set \(\T\) containing the remaining \(20\%\) of the data points.
    \end{enumerate}
    We train the neural network for \(100\) epochs.
    The objective function is the Poisson negative log-likelihood \(\sum_{i=1}^r \left(\lambda(x_{-i})v_{-i}-N_{-i}\log\lambda(x_{-i})\right)\), where \(\S_1=\{(x_{-i},N_{-i},v_{-i})\}_{i=1}^r\) is the training set.

    \subsection{Model comparison}
    
    \subsubsection{Marginal Calibration}\label{subsec: app_marginal_calibration}

    By definition, the null model is marginally calibrated only on the sample on which it is fitted.
    Indeed, it holds that
    \begin{equation*}
        \widehat{\lambda}(x)=\frac{\sum_{i=0}^r N_{-i}}{\sum_{i=0}^r v_{-i}}
    \end{equation*}
    for all \(x\in \X\).
    Therefore,
    \begin{equation*}
        \frac{\sum_{i=0}^r \widehat{\lambda}(x_{-i})v_{-i}}{\sum_{i=0}^r v_{-i}}=\frac{\sum_{i=0}^r N_{-i}}{\sum_{i=0}^r v_{-i}}\cdot \frac{\sum_{i=0}^r v_{-i}}{\sum_{i=0}^r v_{-i}}= \frac{\sum_{i=0}^r N_{-i}}{\sum_{i=0}^r v_{-i}},
    \end{equation*}
    which proves that \eqref{eq: marginal_calibration} holds.
    When \(\S_2\neq \S_1\), we can recalibrate the null model by setting \(\widehat{\lambda}(x)\) to be equal to the exposure-weighted claim frequency over \(\S_2\).
    Under this definition, the null model remains marginally calibrated.

    The Poisson GLM is also marginally calibrated when \(\S_1=\S_2\).
    Indeed, this model minimizes the negative Poisson log-likelihood
    \begin{align*}
        -\sum_{i=0}^r \log \ell(x_{-i}, N_{-i})\propto \sum_{i=0}^r \left(\lambda(x_{-i})v_{-i}-N_{-i}\log(\lambda(x_{-i}))\right),
    \end{align*}
    where \(\lambda(x_{-i})=\exp\left\{\beta_0+\sum_{j=1}^d \beta_{j}x_{-ij}\right\}\).
    It holds that
    \begin{align*}
        \frac{\partial}{\partial \beta_0} \sum_{i=0}^r \left(\exp\left\{\beta_0+\sum_{j=1}^d \beta_{j}x_{-ij}\right\}v_{-i}-N_{-i}\left(\beta_0+\sum_{j=1}^d \beta_{j}x_{-ij}\right)\right)
        & = \sum_{i=0}^r (\lambda(x_{-i})v_{-i}-N_{-i}).
    \end{align*}
    Setting this to zero yields \(\sum_{i=0}^r \widehat{\lambda}(x_{-i})v_{-i} = \sum_{i=0}^r N_{-i}\), which proves \eqref{eq: marginal_calibration}.

    \subsubsection{Width of calibrated confidence intervals}
    As we explained in Section \ref{sec: case_study}, none of the three models produces uniformly narrower calibrated confidence intervals.
    This is verified by Figure \ref{fig: width_intervals}.
    Even if we disregard the frequent spikes, we still observe that different models dominate in different areas of the domain.
    In this figure, we restricted the \(x\)-axis to \([0,0.4]\), where the vast majority of the points lie.

    \begin{figure}[ht]
        \centering
        \includegraphics[width=0.5\linewidth]{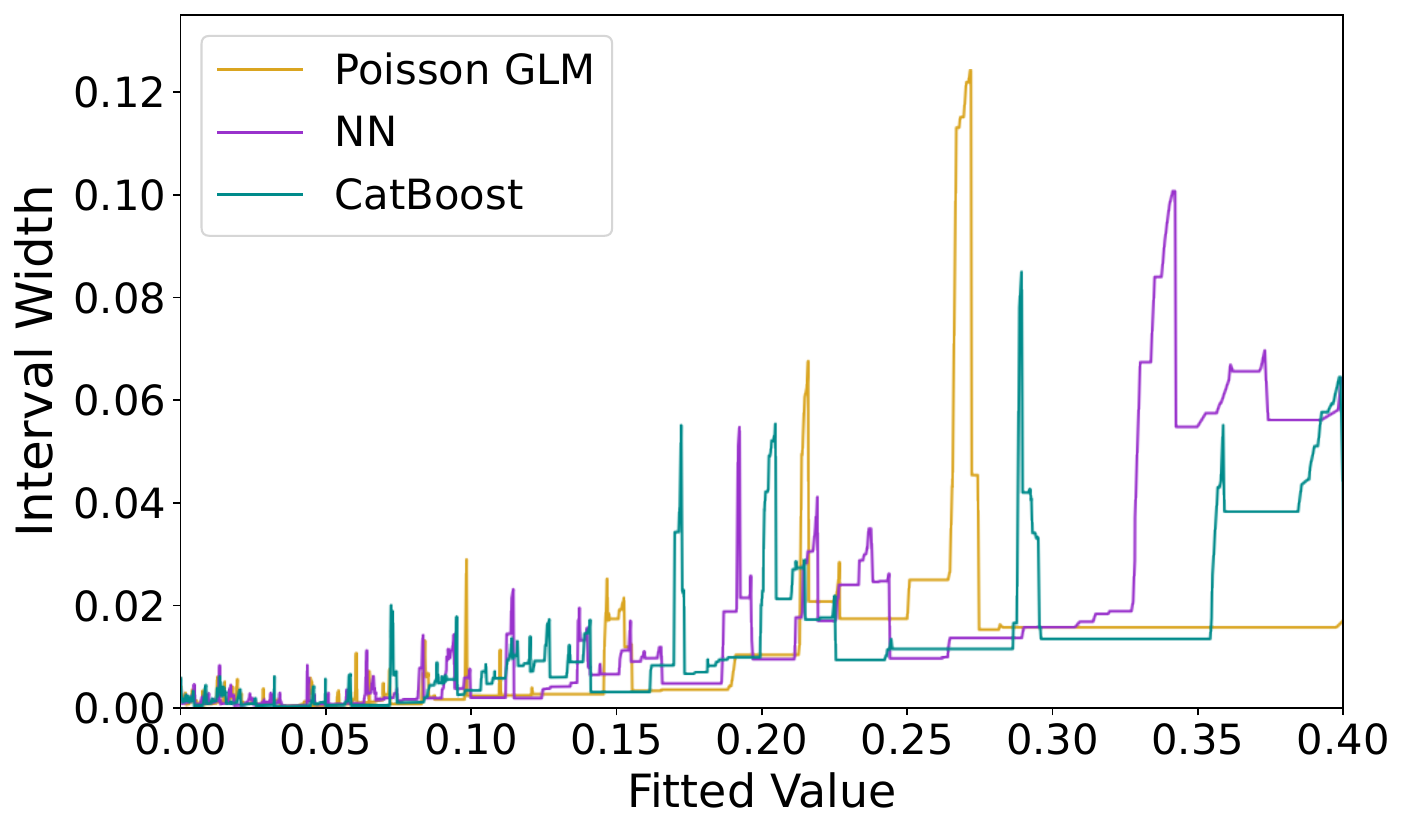}
        \caption{Widths of the calibrated confidence intervals display multiple spikes, which occur if the test point \(\fhat(X_{n+1})\) belongs to a block with few data points.
        Even if we ignore these spikes, we see that no model is uniformly superior over the entire domain.}
         \label{fig: width_intervals}
    \end{figure}
\end{document}